\documentclass{article}
\usepackage{amsmath}
\allowdisplaybreaks
\usepackage{amsthm}
\usepackage[utf8]{inputenc}
\usepackage[english]{babel}

\newtheorem{theorem}{Theorem}
\usepackage{mathtools}
\usepackage{subfigure} 
\usepackage[justification=centering]{caption}
\usepackage{hyperref}
\makeatletter
\def\tagform@#1{\maketag@@@{\bfseries(\ignorespaces#1\unskip\@@italiccorr)}}
\renewcommand{\eqref}[1]{\textup{{\normalfont(\ref{#1}}\normalfont)}}
\makeatother

\usepackage{thmtools}
\usepackage[dvipsnames]{xcolor}
\declaretheoremstyle[
spaceabove=0.5\topsep,
spacebelow=0.5\topsep,
name={Proposition},
]{propsty}

\declaretheoremstyle[
spaceabove=0.5\topsep,
spacebelow=0.5\topsep,
name={Proposition},
]{propsty}
\declaretheoremstyle[
spaceabove=0.5\topsep,
spacebelow=1\topsep,
name=Proof,
headfont=\itshape, numbered=no
]{proofsty}

\usepackage{amssymb}

\usepackage[figuresright]{rotating}

\author{Wenbo Zhang$^1$\and Harsha Honnappa$^2$\and Satish V. Ukkusuri$^1$\footnote{corresponding author, email:sukkusur@purdue.edu}\\}
\date{
	$^1$Lyles School of Civil Engineering, Purdue University\\
	$^2$School of Industrial Engineering, Purdue University\\[2ex]	
}

\begin{document}

\title{\uppercase{Modeling Urban Taxi Services with E-Hailings: A Queueing Network Approach}}

\maketitle

\vspace{1em}
\section*{Abstract}
The rise of e-hailing taxis have significantly altered urban transportation and resulted in an competitive taxi market with both traditional street-hailing and e-hailing taxis. The new mobility services provide similar door-to-door rides as the traditional one and there is competition across these various services. Meanwhile, the increasing e-hailing supply, together with traditional taxicab flows, influence the urban road network performance, which can also in turn affect taxi mode choice and operation. In this study, we propose an innovative modeling structure for the competitive taxi market and capture the interactions not only within the taxi market but also between the taxi market and urban road system. 

The model is built on a network consisting of two types of queueing theoretic approaches for both the taxi and urban road system. Considering both the passenger and vehicle arrivals, we utilize an assembly-like queue $SM/M/1$ for passenger-vehicle matching within the taxi system, which controls how many and how frequently vehicles drive from the taxi system to the urban road system. A common multi-server $M/M/c$ queue that can account for road capacity is proposed for the urban road system and a feedback of network states are sent back to the taxi system. Moreover, within the taxi system, we introduce state-dependent service rate to account for the stochasticity of passenger-vehicle matching efficiency. Then the stationary state distributions, as well as asymptotic properties, of the queueing network are discussed.

An example is designed based on data from New York City. Numerical results show that the proposed modeling structure, together with the corresponding approximation method, can capture dynamics within high demand areas using multiple data sources. Overall, this study shows how the queueing network approach can measure both the taxi and urban road system performance at an aggregate level. The model can be used to estimate not only the waiting/searching time during passenger-vehicle matching but also the delays in the urban road network. Furthermore, the model can be generalized to study the control and management of taxi markets.  

\textbf{Keywords:} E-hailings; Urban Taxi Market; Queueing Network; Vehicle-Passenger Matching; Synchronization Process; Road Congestion 

\section{Introduction}
The rise of flexible transportation companies (TNCs), enabled by the rapid adoption of smartphones, has significantly altered the urban transportation landscape by providing anywhere, anytime mobility. These new mobility services are mainly based on smartphone apps and thus are named as e-hailings, app-based taxi services (ATS), or mobility-on-demand (MoD) services. Competition between ATS's and traditional taxi services (TTS's) has resulted in significant changes in the urban taxi market. For instance, the price of a medallion in the TTS system in NYC has dropped by ninety percent in the most recent five years due to the entry of ATS's \cite{bresiger_2017}. The lack of barriers in this new taxi market (which is a mix of TTS and ATS taxis) has lead to a competitive market where a small number of firms dominate the ATS and TTS market. The ATS and TTS vehicles have similar door-to-door mobility services (note that we only focus on regular or economy taxi services rather than any shared, premium, or special services) but utilize different hailing methods; and the dominant firms (or taxi authorities) create technical (or monetary) barriers to new entries. The emerging market structure is very different from the previous TTS-dominated taxi market where all cabs apply for permits from a central taxi authority, have mostly the same service characteristics and pricing schemes, and taxi authorities control the fleet size by determining the number of licenses to sell. In particular, there are two fundamental shifts in the ATS market which warrant highlighting: 
\begin{itemize}
	\item Change from a TTS-dominated to a competitive market with both ATS and TTS. In the TTS, the authorities can manage fleet size in a direct and centralized way (e.g. medallions) and has direct interaction with the drivers. On the other hand, the ATS provide a matching platform to pair customers and drivers, and acts as a middleman between riders and drivers;
	\item Flexibility and Reduced barriers of entry. ATS allows flexible service hours and almost no barriers for entry. One ATS driver can become available for service anytime and anywhere within the city where the vehicle registers. There are no limits on service hours per shift and fleet size unlike TTS. Moreover, the ATS platforms  updates the fare rate based on supply and demand conditions instead of fixed fare rate in TTS.   
\end{itemize}

Advances in data science, as well as availability of taxi GPS traces, have contributed to a better understanding of the dynamics and operation of taxi markets. In particular, TTS studies include vehicle/passenger movement patterns \cite{CAI2016590,zhang_pricing}, mode choice \cite{ShaheenSusanA2016Ccit}, service efficiency \cite{RAYLE2016168,ZhanXianyuan2016AGAt}, ride-sharing \cite{Alonso-Mora462,QianXinwu2017Oaai,Vazifeh2018}, to system modeling \cite{HE201593,YANG20101067,ZHA2016249}. However, almost all of these studies focus on a unique taxi service rather than modeling the comprehensive taxi market comprising of both the ATS and TTS, except for a few studies which model the competitive market equilibrium \cite{QIAN201743,Heilker2018196}. Given the limited number of studies in the literature on the systematic modeling of taxi systems, this study makes an important contribution towards the development of fundamental tools for measuring system-wide performance metrics of the new taxi market.       

Currently, there are limited modeling tools that allow for the quantification of the quality of service and capture the inherent stochasticity that arises in the taxi markets. There are three challenges that should be addressed in developing system wide tools for modeling the new taxi market:
\begin{itemize}
	\item Spatial Heterogeneity. Since the taxi activities are highly associated with socio-economic variables, the taxi rides also shows significant spatial heterogeneity, particularly, during rush hours. In other words, most of the taxi pickups may be concentrated in some high economic activity zones, such as central business district, while dropoffs are concentrated in other zones. The modeling has to consider properties that dictate the spatial heterogeneity in an urban area.  
	\item Network externalities. Whenever a customer engages a vehicle, this not only decreases the instantaneous vehicle availability at the source location, but also affects the future vehicle availability at all other locations within short timescales. The impacts of network externalities are more significant under dynamic pricing in ATS. Since changes in vehicle availability can be reflected in dynamic pricing this affects future vehicle supply (e.g. induced supply).      
	\item Role of Stochasticity. In a two-sided market, not only do customers choose when to request a ride but also drivers choose when to work, how long to work, and where to search for customers. Moreover, the platforms frequently examine local states and develop corresponding fare rate for a specific time interval and location. This will in turn affect demand and supply. Even if the transportation network is symmetric (i.e. uniform arrival rates and routing choice at all nodes over a regular grid network), the stochastic nature of arrivals will also quickly drive the system out of balance and hence leads to instability. 
\end{itemize}   
The large-scale and stochastic nature of the system makes it challenging, but critical and important, to develop an understanding of the system dynamics by considering both ATS and TTS within the same framework. More importantly, the system modeling should be not only rich enough to capture the salient features of both passengers' and vehicles' behaviors but also the stochastic nature of the demand-supply dynamics and the resulting stability of the system. With this background, the goal of this study is to develop a queuing based methodology to model the combined ATS and TTS system (or `new taxi system') dynamics and determine various performance metrics of the system.

A short review of the literature is summarized below. While not comprehensive, the literature highlights key studies relevant to this paper and organizes them based on various categories.  
\begin{itemize}
	\item The first line of research is the aggregated models that formulate the relationships among system performance metrics, for instance, nonlinear simultaneous equations of system performance  \cite{Yang2000}, queueing theory \cite{Mu_queue,wong_bilateral_search}, and neural network \cite{xu_neural}. These models are primarily focused on modeling the TTS system and derive the whole system performance without any considerations of spatial variations. Then, they employ different modeling structures to explain interactions among the whole system performance. Hence, the major limitation of aggregated models is the capability of addressing stochastic nature of components and spatial dependencies in the demand and supply. 
	\item The second line of research is the equilibrium models that investigate dynamics of drivers and passengers, taking internal and external factors into account. Based on different definitions of equilibrium, there are three modeling structures: spatial demand-supply equilibrium  \cite{WONG2008985,YANG20101067}, competitive equilibrium \cite{wang_game,yang_taxi_OR}, and other defined stationary distribution states \cite{Buchholz2015SpatialES,frchette2016frictions}. However, almost all equilibrium based models are built on weak assumptions of passenger-driver matching. They also do not capture the stochastic nature of the market dynamics, which leads to unreliable estimations on utilities and state distributions. Moreover, the supply-demand equilibrium is mainly designed for a market with perfect information, homogeneous products, no barriers to entry, and profit maximization of service providers. This is not appropriate for mixed TTS and ATTS market which is the focus of this study.
	\item The third line of research generally emulates individual behaviors and their interactions with others through representing system participants or rides as agents or nodes in a graph, for instance, graph theory \cite{ZhanXianyuan2016AGAt} and agent-based simulation \cite{MACIEJEWSKI2015358}. Obviously, the large-scale system will include considerable agents or nodes, as well as their interactions. These methods are not suitable for capturing the large scale system dynamics due to the significant computational time they incur. Current cases in the literature are mostly based on the small hypothetical networks.
	\item The last line of research is the queueing network approaches for station based transportation system. For example, most studies assume a station-based autonomous taxi/vehicle system where customers arrive at predefined stations for autonomous vehicle rental, drive to destinations, and drop off rental autonomous vehicles. Most of these studies model each station as one $M/M/1$ queue with the assumptions of Poisson arrivals and exponential service then connect queues based on routing \cite{GEORGE2011198,Banerjee:2017:POS:3033274.3085099,zhang_robotic}. The process yields a closed Jackson network with product-form stationary distribution. Queueing networks allow analysts to incorporate two important forms of qualitative prior knowledge: first, the structure of the queueing network can be used to capture known connectivity, such as road network, and second, the queueing model inherently incorporates the assumption that the response time explodes when the workload approaches the system’s maximum capacity, which is useful for examining system performance under heavy flow and worst cases. However, the literature emphasizes the approximation techniques to find optimal control policies under queueing network structure, other than validating queueing network and assumptions for transportation systems. In particular, the regular $M/M/1$ queue can not fully explain the driver-passenger matching, as well as strategic behaviors of drivers and passengers under dynamic controls. Several studies also considered both the passenger and driver arrivals in taxi system and proposed a double-ended or synchronization process queue \cite{SHI20161024,shi_non_zero}.     
\end{itemize}
Besides the aforementioned literature, there are a few adaptions of equilibrium models to the emerging competitive new taxi market. The studies in \cite{HE201593,WANG2016212} applied demand-supply equilibrium models to the competitive taxi market but assumed a static pricing scheme rather than a dynamic one by the ATS. Although the study \cite{QIAN201743} considers dynamic pricing for ATS under a framework of competitive equilibrium between ATS and TTS, they failed to precisely formulate the vehicle-passenger matching process, as well as the utilities of passengers and drivers. 

In this study, we take advantage of both queueing network models and matching queues to investigate the large-scale and stochastic nature of the competitive taxi market, yielding quantitative performance measures. Within the queueing network, we have two types of queues/nodes representing two different subsystems. The first type of nodes are the taxi passenger-vehicle matching subsystem. Since neither of the ATS and TTS have stations, we assume that one homogeneous spatial unit is modeled as one taxi subsystem. Vehicles are matched with passengers based on a synchronization process, forming a `synchronized' $SM/M/1$ queue. The latter differs from the regular $M/M/1$ queue in that it has two independent arrival flows of both passengers and vehicles thus processes synchronized passenger-vehicle pairs that match based on arrival sequences and zero matching time. Although there exist certain differences, the $SM/M/1$ queue can be closely approximated by the simple $M/M/1$ queue by taking the minimum of the two arrival rates as the effective input rate; this approximation is further detailed below. 

The second type of nodes are the road transport subsystem. Again the road network is split into homogeneous units and each homogeneous unit is modeled as one road subsystem in the form of an $M/M/c$ queue. The road subsystem can be connected with both road and taxi subsystems based on geographical contiguity. However, the taxi subsystems only serve their spatial unit and can not directly interact with other taxi subsystems, because vehicles departing from taxi subsystems should travel through road network to get to their destinations. Further, we also model balking behavior of vehicles in the $SM/M/1$ queue given the fact that empty vehicles can drive to another spatial unit if s/he can not find passengers at one spatial unit. Finally, we also account for the dynamics of matching efficiency with state-dependent service rates. Note that one of the novelties of our model is that we explicitly combine a model of `virtual' infrastructure (the passenger-vehicle matching queue) with a model of `physical' infratructure (the urban road network) to obtain a holistic view of the taxi system dynamics.

Queueing theoretic approaches provide insight into system performance under a range of workload conditions. In particular, it allows us to identify the degree of load that will cause a system to become highly congested without actually cascading into failure. On the other hand, queueing approaches have a reputation for making unrealistic distributional assumptions and of lacking robustness to divergence of the actual system from modeling assumptions. Here, we use multiple statistical hypothesis tests at various spatio-temporal resolutions to justify our queueing theoretic model. In general, queueing theoretic models are data agnostic, and provide sufficient conditions under which one can compute performance metrics of interest. Having said that, it is not {\it a priori} apparent that a given queueing network is appropriate at a given spatio-temporal resolution, for the taxi system. We combine the queueing network model with extensive statistical hypothesis testing to justify an appropriate spatio-temporal aggregation scale at which the observed arrival and service conditions are sufficiently `homogeneous,' thereby yielding empirical validation for our modeling assumptions. 


The remainder of the paper is organized as follows: Section 2 presents the modeling structures for the competitive taxi market with both ATS and TTS; Section 3 investigates the stationary state distributions and asymptotic behaviors under the proposed modeling structure; Section 4 provides a case study based on TTS and Uber in NYC; and Section 5 concludes the paper and points out future research directions.


\begin{table} 
	
	
	
	\centering
	\begin{tabular}{r c p{10cm}}
		\hline \multicolumn{1}{c}{\textbf{Symbols}} & \multicolumn{1}{c}{\textbf{}} & \multicolumn{1}{c}{\textbf{Descriptions}} \\ \hline 
		$\lambda_{i}^{p}$ & $ $ & The total passenger arrival rate of both the ATS and TTS at spatial unit $i$, and superscript $p$ denotes passengers\\
		
		$p_{i}^{*}$& $ $ &The probability of passengers using service * at spatial unit $i$, and superscript * can be either ATS or TTS\\
		
		$\lambda_{i}^{p,*}$&$ $&The passenger arrival rate at spatial unit $i$, superscript $p$ denotes passengers, and superscript * can be either ATS or TTS\\
		
		$\lambda_{i}^{v,*}$& $ $ &The external arrival rate of vehicles at spatial unit $i$, $V$ denotes vehicles, and superscript * can be either ATS or TTS\\
		
		$\lambda_{i}^{pv,*}$& $ $&The arrival rate of passenger-vehicle pairs at spatial unit $i$, superscript $pv$ denotes passenger-vehicle pairs, and superscript * can be either ATS or TTS\\
		
		$\hat{\lambda}_{i}^{v,*}$& $ $ &The effective arrival rate of vehicles at spatial unit $i$, and superscript * can be either ATS or TTS\\
		
		$\mu_{i}^{*}$& $ $ &The service rate for passenger-vehicle pairs at spatial unit $i$, and superscript * can be either ATS or TTS\\
		
		$D_{i}^{*}$& $ $ &The departure flow rate of passenger-vehicle pairs at spatial unit $i$, and superscript * can be either ATS or TTS\\

		$\lambda_{i}^{r}$& $ $ &The overall vehicle arrival rate for road queue at spatial unit $i$, and superscript $r$ denotes the road queue\\
		
		$p_{i,.}^{r,*}$& $ $ &The portion of vehicle flows at spatial unit $i$, subscript $.$ can be either $O$ (occupied vehicles) or $E$ (empty vehicles); and superscript * can be either ATS or TTS\\

		$\mu_{i}^{r}$& $ $ &The service rate at road queue of spatial unit $i$, and superscript $r$ denotes the road queue\\
	
		$p_{i}^{p,*}$& $ $ &The probability of empty vehicles successfully picking up passengers at spatial unit $i$, superscript * can be either ATS or TTS, and superscript $p$ denotes passenger pickups\\
		
		$p_{ij}^{.,*}$& $ $ &The probability of vehicles moving from spatial unit $i$ to $j$, superscript $.$ can be either $O$ (occupied vehicles) or $E$ (empty vehicles); and superscript * can be either ATS or TTS\\
		
		$p_{i0}^{.,*}$& $ $ &The probability of vehicles exiting system at spatial unit $i$, superscript $.$ can be either $O$ (occupied vehicles) or $E$ (empty vehicles); and superscript * can be either ATS or TTS\\
		
		$F_{i,in}$& $ $ &The incoming vehicle flow rate at spatial unit $i$, regardless of service types and vehicle status\\
		
		$F_{i,out}$& $ $ &The outgoing vehicle flow rate at spatial unit $i$, regardless of service types and vehicle status\\
		
		$I$& $ $ &The set of spatial units with cardinality of $|I|$\\
		
		$a_{ij}$& $ $ &The connectivity between spatial units $i$ and $j$ with physical road network\\	
		
		$x_i^{*}$& $ $ &The number of vehicle in taxi queue at spatial unit $i$, and superscript * can be either ATS or TTS\\
		
		$x_i^{.,*}$& $ $ &The number of vehicle in road queue at spatial unit $i$, superscript $.$ can be either $O$ (occupied vehicles) or $E$ (empty vehicles); and superscript * can be either ATS or TTS\\

		$c_{i}$& $ $ &Number of road servers at spatial unit $i$\\ \hline
	\end{tabular}
\end{table}

\section{Modeling Structures}
\subsection{Network Presentation of the Competitive Taxi Market}
In the TTS-dominated taxi market, the dynamics can be simply explained by a bilateral passenger-vehicle matching, as well as vehicle movements among spatial units. In particular, the passenger-vehicle matching behaviors are critical to the system performance, for instance, waiting/searching time and utilization. Since the existence of spatiotemporal mismatch between one drop-off and the next pickup, taxicab drivers always search around for passengers. Although ATS also operates like TTS, the centralized platform with real-time controls introduces more complexity, as shown in Fig.\ref{fig:TTS_ATS_system}. The first addition is the competition for passengers between ATS and TTS. The second addition is the ATS drivers' flexible working hours. Moreover, the ATS platform examines demand and supply frequently and utilizes dynamic pricing for seeking balance of demand and supply. Overall, the competitive taxi market receives two types of external flows: passengers and vehicles. The vehicles are operated by both the taxi fleet and ATS driver partners. Due to the market model of free entry, the ATS vehicles can enter and exit the system in a frequent manner. Thus, this is an open system involving external arrivals and exits.              
\begin{figure}[!htbp]
	\centering
	\includegraphics[width=0.9 \columnwidth]{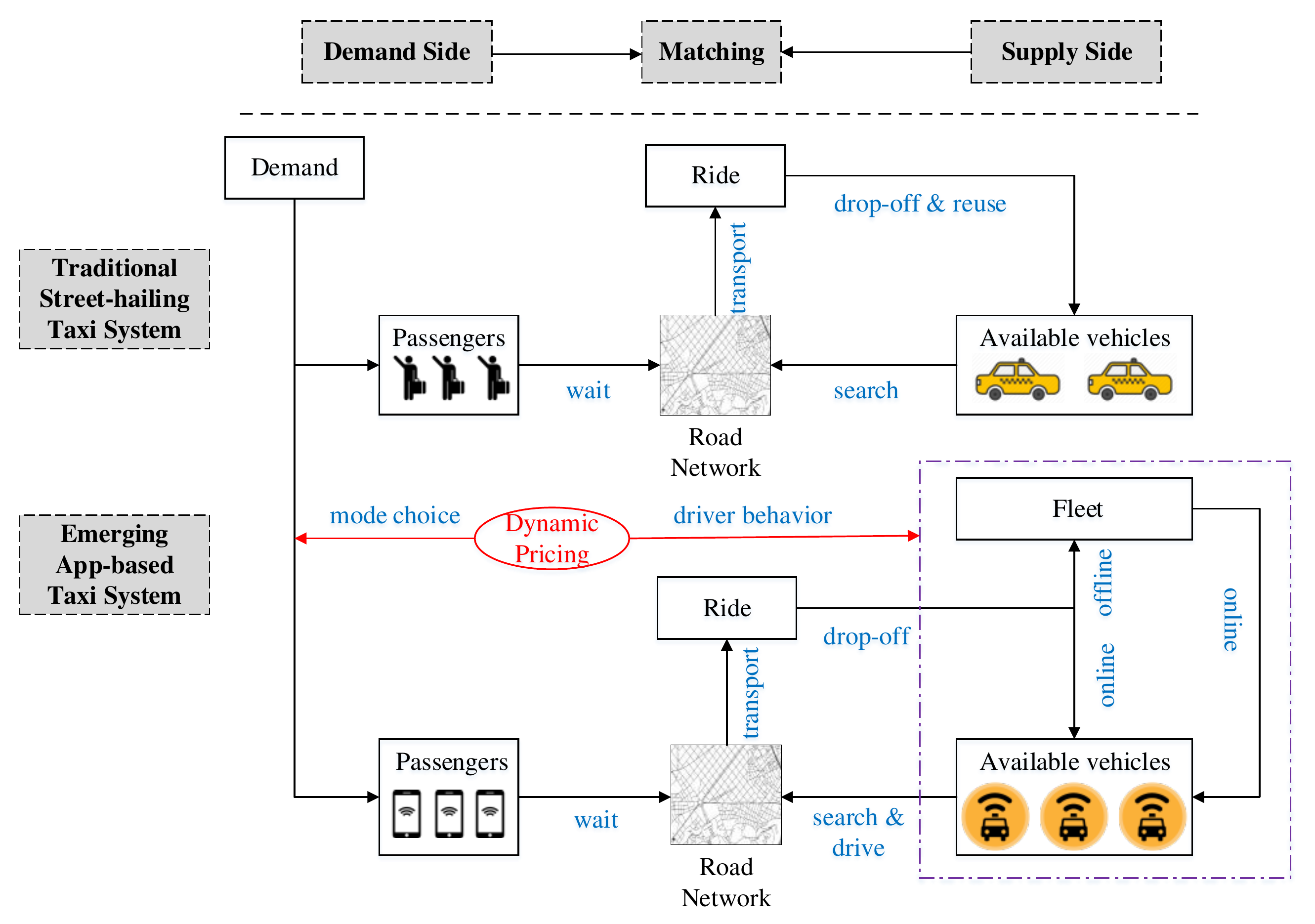}
	\caption{Competitive taxi market with both ATS and TTS}\label{fig:TTS_ATS_system}
\end{figure}

As mentioned before, the major behaviors of the taxi system are passenger-vehicle matching and occupied/empty vehicle movements over urban road network. Thus, the analyses should involve not only the taxi system itself but also the urban road network. It is well known that the both systems are spatially unbalanced, for instance downtown with high ride requests, but with slower ground traffic speed. We first divide the whole system (e.g. a city) into multiple subsystems (e.g. homogeneous spatial units). The empty vehicles can meet passengers in each subsystem, called taxi subsystem and then the matched pairs of passengers and vehicles (or occupied vehicles) travel within and across road subsystems. Therefore, it forms a system of systems $G(N,A)$, where $N$ is the combination of all divided subsystems (or spatial units), each of which operates one taxi and one road subsystem; and $A$ is the set of directed links indicating the connections across subsystems, consisting of $a_{ij}$. Moreover, the directed links are weighted with routing probabilities $p_{ij}$ to describe the routing choices by vehicles. We further classify into four routing probability matrices that all can be derived from our empirical datasets, depending on service types (ATS or TTS) and vehicle status (occupied or empty). Regarding each unique spatial unit, Fig. \ref{network_presentation} illustrate major taxi activities and segment based on vehicle status. One spatial unit generally receives two external arrivals of both vehicles and passengers (e.g. p1,p2, and p3). In specifics, external vehicle arrivals may originally generate within the spatial unit (e.g. e2) or transfer from neighboring spatial units, regardless of occupied (e.g. o6, o10, o12) and empty (e.g. e1, e2, e7) vehicles. The detail structures of taxi subsystems will be presented in the next section, addressing not only two external arrivals but also more complicated behaviors of dropoff followed by pickup, for instance o6 and e3. On the other hand, each road subsystem only describes the vehicle and taxicab movement over road network directed by the corresponding routing probabilities. In addition, we introduce a virtual node $N_0$ as the exit node from the system and describe the exiting vehicle flows. In summary, each taxi subsystem only addresses the matching dynamics between empty vehicles and passengers and then transmits matched pairs of passengers and vehicles to the road subsystems. Each road subsystem moves both occupied and empty vehicles among taxi subsystems. In the following two sections, we model each subsystem using queueing models.          

\begin{figure}[!htbp]
	\centering
	\includegraphics[width=0.85 \columnwidth]{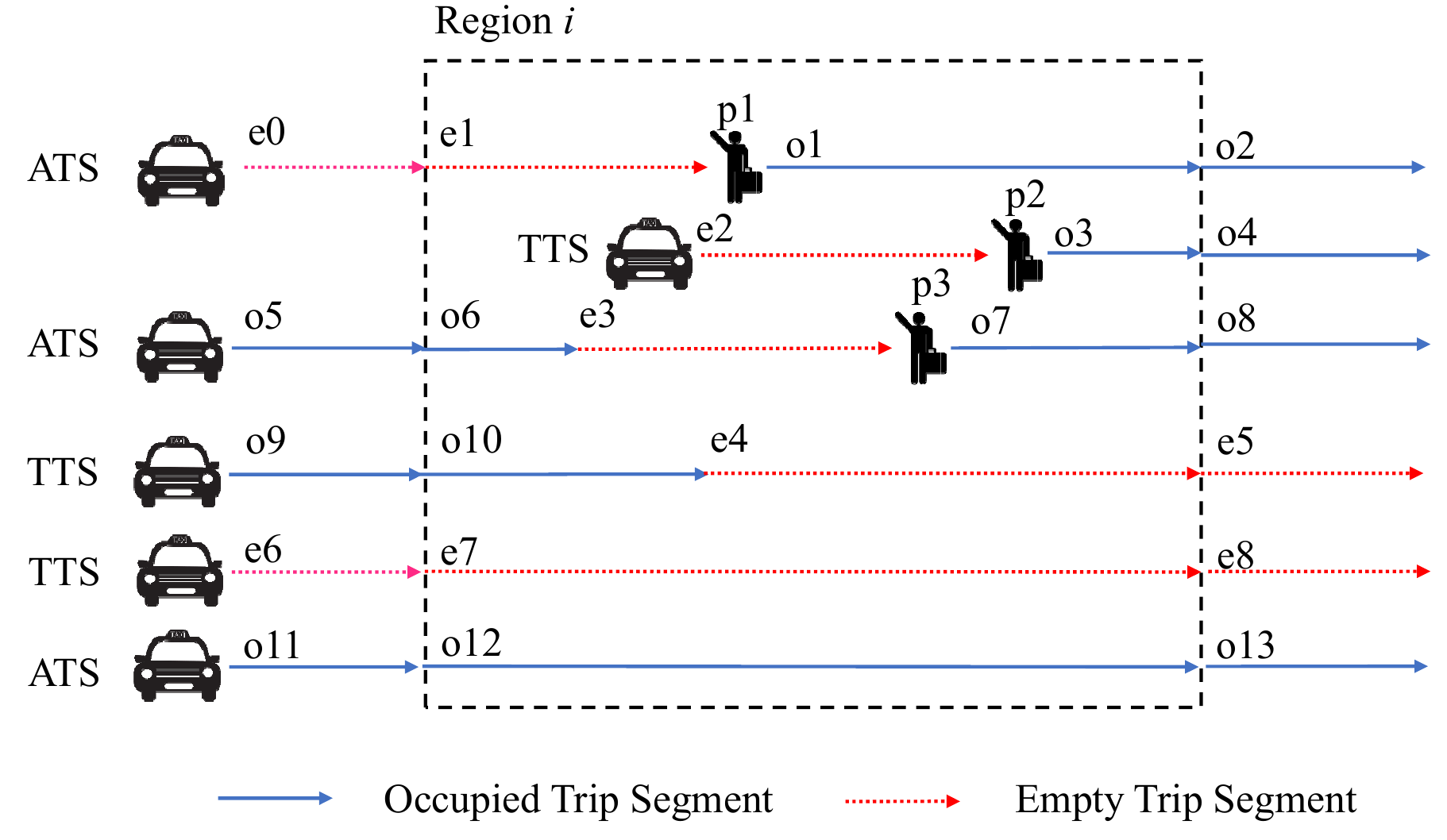}
	\caption{Illustration of major taxi activities in one spatial unit}
	\label{network_presentation}
\end{figure}

\subsection{Passenger-Vehicle Matching}
The taxi system, regardless of whether it is an ATS or TTS, requires matching passengers and vehicles. Standard approaches for modeling matching include nearest distance and Cobb-Douglas production function. However, the former is inappropriate, since it has been observed that even when drivers have perfect knowledge, they do not apply a nearest distance heuristic to find a passenger. On the other hand, it is typically very hard to calibrate Cobb-Douglas production function from available data. A more appropriate approach would be to use matching or assembly-like queues \cite{Harrison_assemble}. Here, passengers and vehicles are queued up in separate ``buffers" and are matched on a first-come-first-served (FCFS) basis. The arrival flow of passengers and vehicles is determined by a ``synchronized" stochastic process, defined as the minimum of the individual flows. Assuming that the individual arrival processes are Poisson processes and that the `service' (i.e. matching time) times are exponentially distributed, we model the matching process as a $SM/M/1$ assembly-like, synchronized queue. As noted above, we assume that matching is conducted on a FCFS basis, which is reasonable way in which passengers and vehicles are matched. The service time, on the other hand, models how quickly a vehicle can reach a passenger's location within a subsystem. This is critical for modeling the dynamics of ATS, in particular. 

Since the ATS and TTS coexist and compete in each subsystem, we introduce an $SM/M/1$ queue for ATS and TTS each and deploy a parallel layout together with demand splitting, as shown in Fig. \ref{fig:synchronization}. Given a taxi subsystem $i\in I$, the overall passenger arrivals to this spatial unit follow a Poisson process with rate $\lambda_{i}^P$. With Bernoulli splitting, the passengers are split into two Poisson processes with rates $\lambda_{i}^{P,ATS}=p_{i}^{ATS} \lambda_{i}^P$ and $\lambda_{i}^{P,TTS}=p_{i}^{TTS} \lambda_{i}^P$ with $p_{i}^{ATS}+p_{i}^{TTS}=1$. The available ATS vehicle arrival $\hat{\lambda}_{i}^{V,ATS}$ consists of two sources: (1) the Poisson process of newly joined vehicles (e.g. e2 in Fig. \ref{network_presentation}) with a rate $\lambda_{i}^{V,ATS}$; and (2) empty vehicles who are searching for passengers and successfully pick up in final, but originate from neighboring spatial units (e.g. e1 and e3  in Fig. \ref{network_presentation}), $F_{i,in}p_i^{p,ATS}$. The effective arrival rate of vehicle is shown in equation \ref{eq1}. The derivations of vehicle incoming flows $F_{i,in}$ will be shown in the next section on network flow balance, since they are based on departure flows from all other spatial units. Similarly, we can also obtain the effective TTS vehicle arrival rate in equation \ref{eq2}.

In addition, we derive service rate $\mu_i^{ATS}$ and $\mu_i^{TTS}$, directly from empirical observations on vehicle searching time, for instance, duration of processes e1, e2, and e3. Before figuring out the service rate measurements, we should clarify several key points. First, the $M/M/1$ queues for both service types are built at zone levels, other than taxi stands or points of interest. It may be related to zonal road network configurations and length but are less likely to be observed in reality. Second, the vehicle searching time are observable, only by counting empty trips that are fully or partially inside spatial unit $i$. The outside trip segments even for same vehicles are assumed to be not related to matching efficiency of the spatial unit $i$. Thus, under the $M/M/1$ modeling structure, we can derive service rate based on observed total system time (i.e. vehicle searching time from begins of passenger searching to pickups). In $M/M/1$, the system time follows exponential distribution, as shown in equation \ref{service_dist} and \ref{service_dist_ats}. Thus, the difference between service rate and arrival rate should be the mean system time as definitions of exponential distribution, as shown in equation \ref{mu_deriv} and \ref{mu_deriv_ats}.           
\begin{gather}
\hat{\lambda}_{i}^{v,ATS}\coloneqq \lambda_{i}^{v,ATS}+F_{i,in}p_i^{p,ATS}   \label{eq1}\\
\hat{\lambda}_{i}^{v,TTS}\coloneqq \lambda_{i}^{v,TTS}+F_{i,in}p_i^{p,TTS} \label{eq2} \\
w(t_{i}^{ATS})=\left(\mu_{i}^{ATS}-\lambda_{i}^{pv,ATS}\right)e^{-\left(\mu_{i}^{ATS}-\lambda_{i}^{pv,ATS}\right)} \label{service_dist}\\
w(t_{i}^{TTS})=\left(\mu_{i}^{TTS}-\lambda_{i}^{pv,TTS}\right)e^{-\left(\mu_{i}^{TTS}-\lambda_{i}^{pv,TTS}\right)} \label{service_dist_ats}\\
\mu_{i}^{ATS}=\lambda_{i}^{pv,ATS}+\hat{t_{i}}^{ATS} \label{mu_deriv}\\
\mu_{i}^{TTS}=\lambda_{i}^{pv,TTS}+\hat{t_{i}}^{TTS} \label{mu_deriv_ats}
\end{gather}
where, $w(\cdot)$ is the probability density function of random variable; $t_{i}^{*}$ is observed vehicle searching time (* can be either ATS or TTS); and $\hat{t_{i}^{*}}$ is empirical mean vehicle searching time (* can be either ATS or TTS). 
\begin{figure}[!htbp]
	\centering
	\includegraphics[width=0.9 \columnwidth]{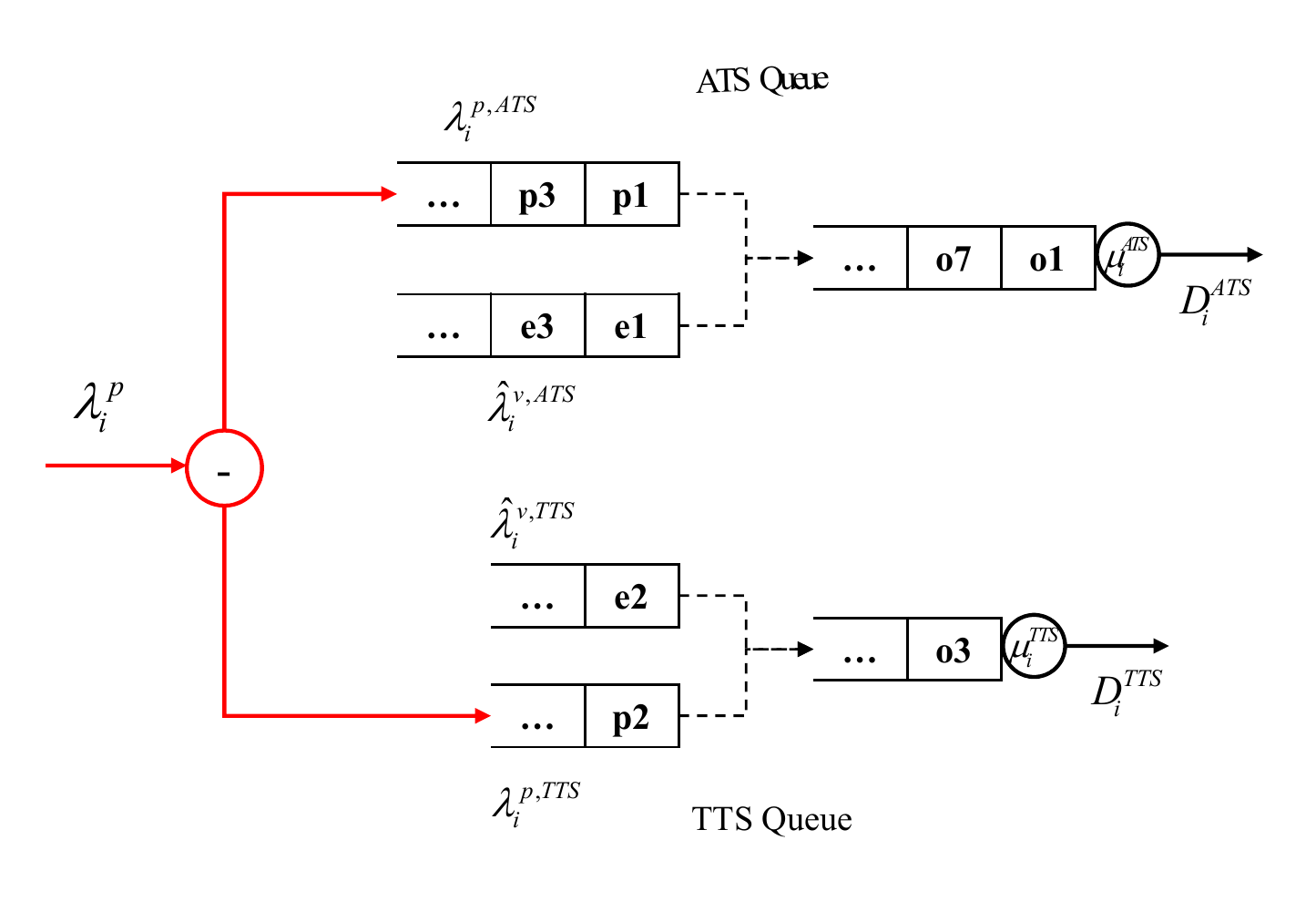}
	\caption{The synchronization process for passenger-vehicle matching at taxi queue $i$}\label{fig:synchronization}
\end{figure}

\subsection{Inclusions of Road Network Performance}
Since taxicabs transport passengers through the road network, there is a close interaction between the taxi- and urban road- systems. As mentioned before, we also split the urban road system into multiple homogeneous subsystems, each of which is modeled as a $\cdot/M/c$ queue, as shown in Fig. \ref{fig:road_server}, where $1<c<\infty$ represents the road capacity. Once a vehicle enters the road subsystem, it queues up and waits for available road space. The derivation of number of servers in each homogeneous road subsystem, $c$, is based on the idea of Macroscopic Fundamental Diagram (MFD) proposed and applied in recent years \cite{GEROLIMINIS2008759,RAMEZANI2017}. MFD models the relationship of traffic accumulation (or network density) and production (outgoing flows) and indicates a critical accumulation leading to a congested road network. The $c$ corresponds to the critical taxi accumulations. Since the both terms reveal the maximum number of vehicles can be processed without delays. On the other hand, the derivation of service rate at each server is similar as taxi queues in equations \ref{service_dist} to \ref{mu_deriv_ats}, by counting vehicle travel time in one specific spatial unit and computing based on exponential distribution of observed travel time.  

The last component of interest in the road queues is the effective arrival and departure flows. Since the road network does not differentiate service types and vehicle status. The effective arrivals should be a pooled flow from both the ATS and TTS containing two types of vehicle flows: (1) matched pairs (i.e. occupied vehicles transporting passengers to destinations) from taxi subsystem $i$, $D_{i}^{TTS}$ and $D_{i}^{ATS}$; and (2) remaining vehicle arrivals in $F_{i,in}$, who just driving through the spatial unit $i$, regardless of searching (e.g. e7 and e4 in Fig.\ref{network_presentation}) or transporting passengers (e.g. o12 and o10 in Fig.\ref{network_presentation}). The effective arrival process is in equation \ref{eq:road_arrival}. The detail analyses on the pooled flows will be presented in the next section on network flow balance. One more interesting point is about the departure flow of $M/M/c$. More complicated than occupied vehicle flow departure from $M/M/1$, the departure flow from road queue will have multiple vehicle status (occupied or empty) and service types (ATS or TTS). Considering different movement patterns, we further distribute departure flow depending on vehicle status and service types. Different type of vehicles are assigned with special routing probabilities for distribution over road network. Identification of vehicle types is primarily based on their portions in incoming flows of spatial units, which are consistent regardless of queue arrival and departure flows. 

\begin{figure}[!htbp]
	\centering
	\includegraphics[width=0.9 \columnwidth]{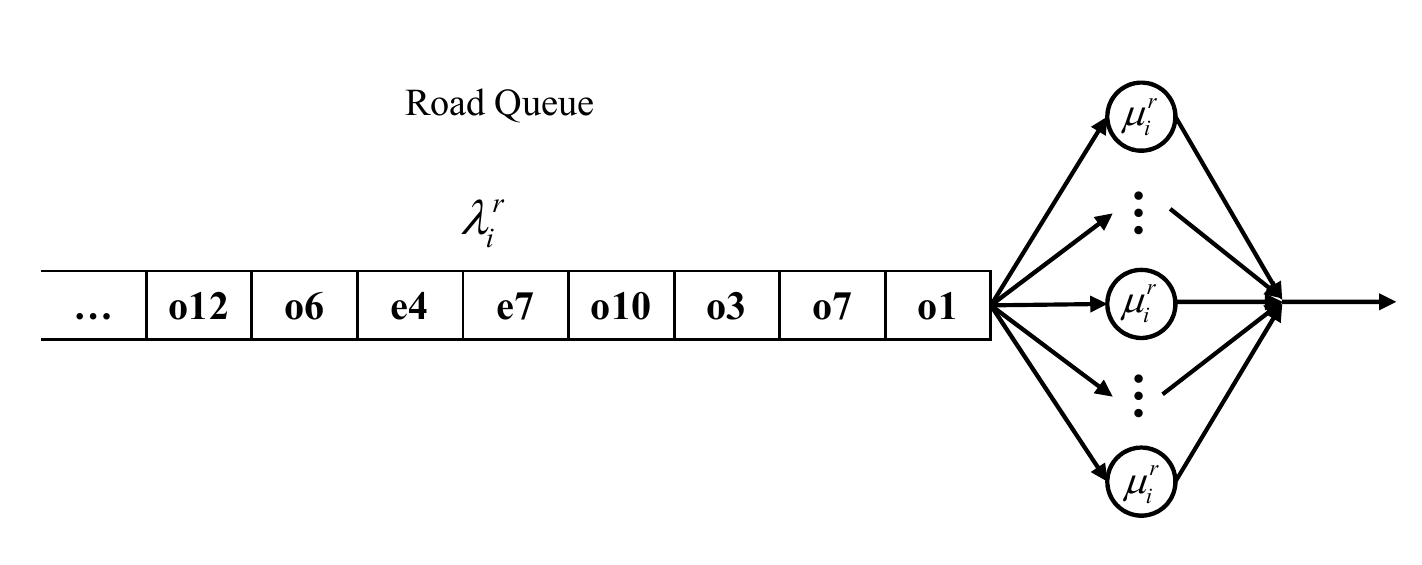}
	\caption{The queue for vehicle traveling through road queue $j$}\label{fig:road_server}
\end{figure}  
\begin{align}\label{eq:road_arrival}
\lambda_{i}^{r}\coloneqq D_{i}^{TTS}+D_{i}^{ATS}+\left(1-p_i^{p,ATS}-p_i^{p,TTS}\right) F_{i,in}
\end{align}  

Except for modal split between ATS and TTS, the both taxi queues interacts with each other in the ways of vehicle flow split and merges within each spatial unit, as shown in Fig.\ref{fig:subnetwork}. Once vehicle flows enter one spatial unit, it will split based on service types and vehicle status. For those vehicles who can pick up new passengers in this spatial unit, $F_{i,in} p_i^{p,ATS}$, they will form external vehicle arrival for ATS queue along with newly online vehicles $\lambda_{i}^{v,ATS}$ and yield a departure flow from ATS queue with rate of $D_i^{ATS}$. Similar split is applied for TTS queue and yields a departure flow for TTS queue with rate of $D_i^{TTS}$. The departure flows from both taxi queues (i.e. occupied vehicles with new pickups) will queue at road queue, along with remaining vehicle arrivals of incoming flow who do not pickup any new passengers. Note that the vehicle sequence of road queue shown in Fig.\ref{fig:subnetwork} is just an example. Following the properties of $M/M/1$ and $M/M/c$, we can also derive following equations:     
\begin{figure}[!htbp]
	\centering
	\includegraphics[width=1.2 \columnwidth]{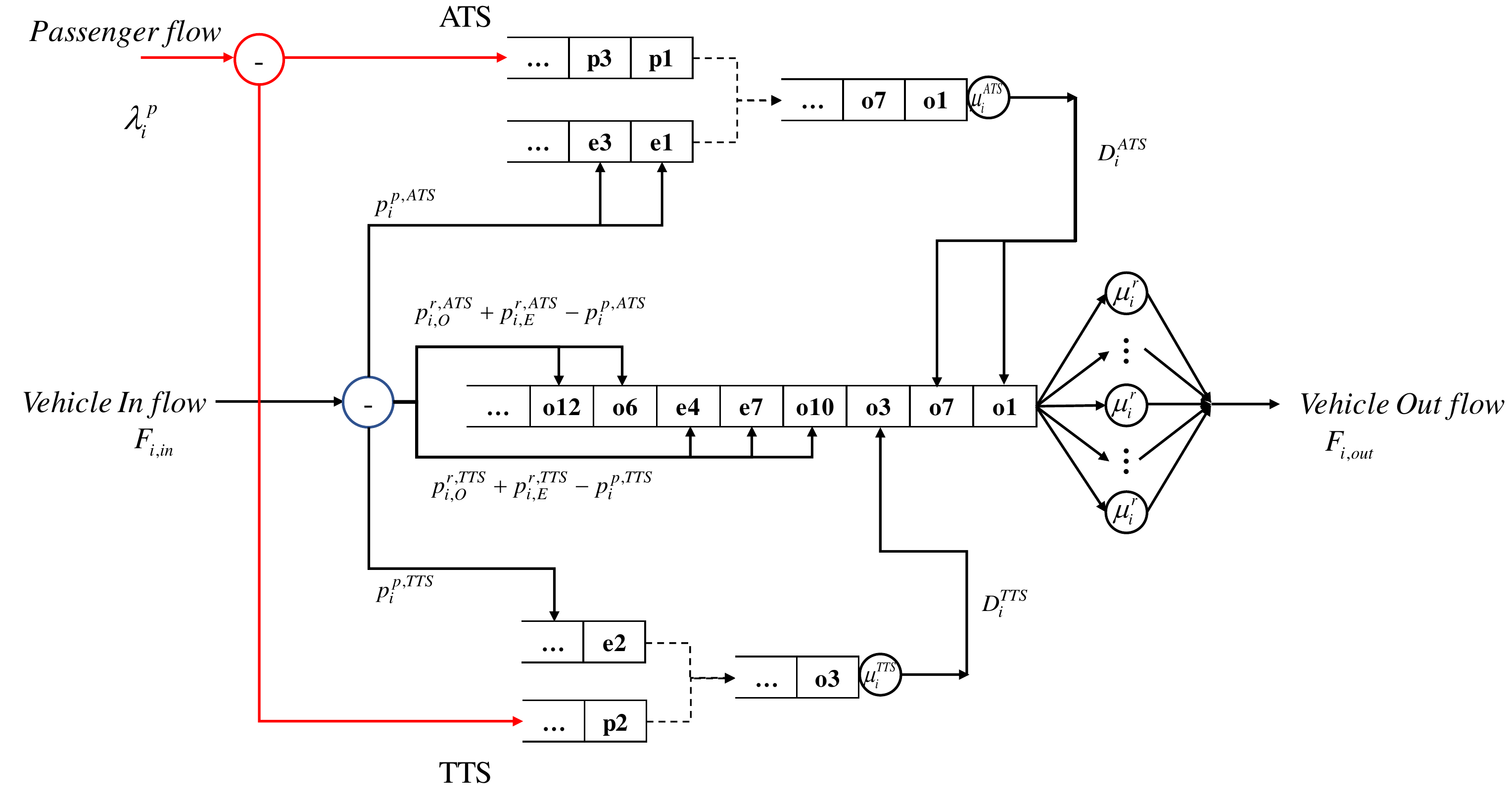}
	\caption{The subnetwork consisting of two taxi queues and one road queue at spatial unit $i$}\label{fig:subnetwork}
\end{figure}

\begin{gather}
\lambda_{i}^{pv,ATS}=D_i^{ATS}\\
\lambda_{i}^{pv,TTS}=D_i^{TTS}\\
F_{i,out}=\lambda_{i}^{r}\coloneqq D_{i}^{TTS}+D_{i}^{ATS}+\left(1-p_i^{p,ATS}-p_i^{p,TTS}\right) F_{i,in}
\end{gather} 

Beyond one unique spatial unit, the incoming and outgoing flows, $F_{i,in}$ and $F_{i,out}$, can be formulated with routing probabilities as follows. 
\begin{gather}
F_{i,in}= \sum_{j\in I}a_{ji} F_{j,out}\left( p_{ji}^{O,ATS}p_{j,O}^{r,ATS} + p_{ji}^{E,ATS}p_{j,E}^{r,ATS} +p_{ji}^{O,TTS}p_{j,O}^{r,TTS}+ p_{ji}^{E,TTS}p_{j,E}^{r,TTS}\right) \label{in_flows}
\end{gather} 



\section{Stationary State of Urban Taxi Queueing Network}
It has been demonstrated that the joint distribution of two individual flows in the synchronization process can never reach the steady state, regardless of arrival and service rates \cite{Harrison_assemble}. Thus, inclusions of synchronization processes make the queueing network complicated and unstable. Therefore, this section starts from the transient and asymptotic behaviors of synchronized flows and discusses the possible approximations for synchronization processes with common queues enabling us to obtain stationary distributions for the proposed queueing network.  

\subsection{Instability of Synchronized Flows and Approximations}
The proposed $SM/M/1$ queue can be presented as a two stage process by introducing a virtual buffer for synchronized pairs of passengers and vehicles, consisting of (1) virtually synchronizing two distinct flows of passengers $\lambda_{P}$ and vehicles $\lambda_{V}$ and immediately queueing up at a virtual buffer; and (2) serving the queued synchronized flow with a service rate of $\mu$, as shown in the left side of Fig.\ref{fig:approximation}. Let $X_{t}^P$ and $X_{t}^V$ be the number of passengers and vehicles in corresponding buffers at time $t$. $S_t=\min(X_{t}^P,X_{t}^V)$ is number of the synchronized pairs of passengers and vehicles at the virtual buffer. Prior studies \cite{Alexander2010,Prabhakar_synchro} have explored the transient and asymptotic behaviors of $S_t$ and proved that the synchronized flow $S_t$ converges to a Poisson process both analytically and numerically. The literature further yields a $M/M/1$ approximation for $SM/M/1$, with the arrival rate $\min(\lambda^{P},\lambda^{V})$ and service rate $\mu>\min(\lambda^{P},\lambda^{V})$ shown in the right side of Fig.\ref{fig:approximation}. The approximation has been validated by simulation demonstrating small differences in system performance metrics between the simulation and approximation, less than 1\% in most cases and no more than 3\% under the condition of heavy traffic (i.e. large $\min(\lambda^{P},\lambda^{V})/\mu$), as well as equal arrival rates of two distinct flows \cite{Alexander2010}. For completeness, we present several key properties of $S_t$. The details in derivations and proofs can be found in \cite{Alexander2010,Prabhakar_synchro}.          
\begin{figure}[!htbp]
	\centering
	\includegraphics[width=0.75 \columnwidth]{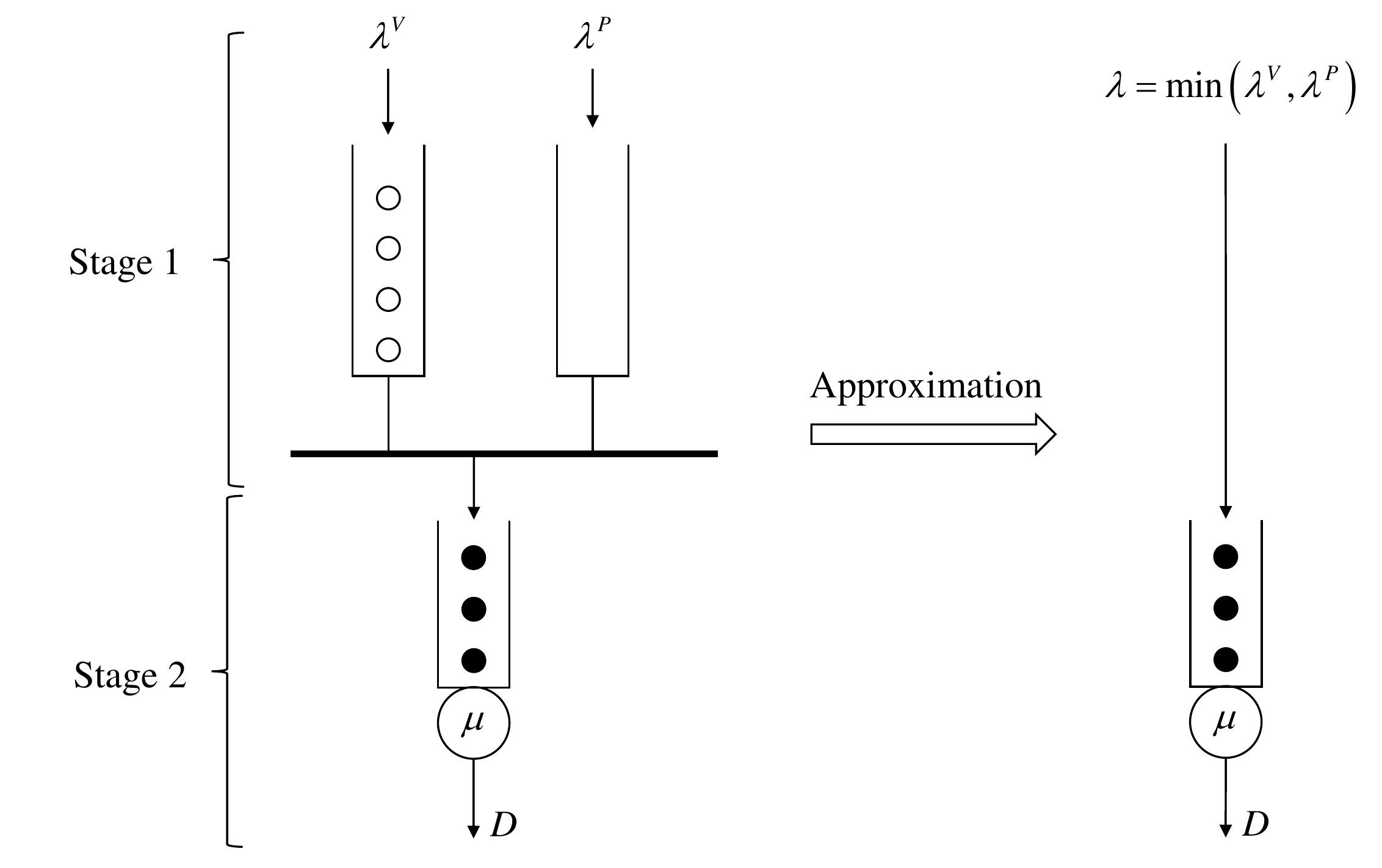}
	\caption{The general case of the $SM/M/1$ approximation with $M/M/1$}\label{fig:approximation}
\end{figure} 

Assume independent Possion arrivals of both passengers and vehicles. Observe that, 
\begin{align*}
P(S_t=s)&=P(S_t=s\ \text{and}\ X_t^{P}\geq X_t^{V}) + P(S_t=s\ \text{and}\ X_t^{P}< X_t^{V})	\\
&=\sum_{x=s}^{\infty}e^{-\lambda_{P}t}\frac{(\lambda^{P}t)^x}{x!}e^{-\lambda^{V}t}\frac{(\lambda^{V}t)^s}{s!}+\sum_{y=s+1}^{\infty}e^{-\lambda_{V}t}\frac{(\lambda^{V}t)^y}{y!}e^{-\lambda^{P}t}\frac{(\lambda^{P}t)^s}{s!}
\end{align*}
. Obviously, this is not a Poisson process. The approximation idea is that as $t\rightarrow\infty$, $S_t$ may approach a Poisson process so that asymptotically we have $M/M/1$-like behavior. Investigating the mean and variance of $S_t$, it is well-known that,      
\begin{itemize}
	\item  The pair process $S_t$ has the following asymptotic properties: 
	\begin{equation*}
	\begin{cases}
	\lim_{t\to\infty} P(S_t=X_t^P)=1 & \text{if}\ \lambda^{P}<\lambda^{V}\\
	\lim_{t\to\infty}P(S_t=X_t^V)=1 & \text{if}\ \lambda^{P}>\lambda^{V}\\
	\lim_{t\to\infty}P(S_t=X_t^V)=\lim_{t\to\infty}P(S_t=X_t^P)=\frac{1}{2} & \text{if}\ \lambda^{P}=\lambda^{V}\\
	\end{cases}
	\end{equation*}
	
	\item 	The long-run averages of the mean and variance of the synchronization process $S_t$ are given by
	\begin{equation*}
	\begin{cases}
	\lim_{t\to\infty}\frac{E[S_t]}{t}=\lim_{t\to\infty}\frac{V[S_t]}{t}=\min(\lambda^{P},\lambda^{V}) & \text{if}\ \lambda^{P}\neq\lambda^{V}\\
	\lim_{t\to\infty}\frac{E[S_t]}{t}=\lambda^{P}\ \text{and} \  \lim_{t\to\infty}\frac{V[S_t]}{t}=\lambda^{P}\left(1-\frac{1}{\pi}\right) & \text{if}\ \lambda^{P}=\lambda^{V}\\
	\end{cases}
	\end{equation*} 
\end{itemize} 

\subsection{Network Flow Balance} 
It is straightforward that the queueing network can be presented by Jackson network consisting of $|I|$ $M/M/1$ queues for ATS passenger-vehicle matching, $|I|$ $M/M/1$ queues for TTS passenger-vehicle matching, and $|I|$ standard $M/M/c$ queues for road subsystems. The $M/M/1$ approximations reduce the two individual flows into a minimum one, which cannot yield a reliable track of demand and supply unbalance. However, recall the asymptotic behaviors of $S_t$, as $t\rightarrow\infty$, the synchronization process can match almost all arrivals in the flow with lower rate. 

For each taxi subsystem $i\in I$,
\begin{align}
\lambda_{i}^{pv,TTS}=D_i^{TTS}&=\min \left(\hat{\lambda}_i^{v,TTS},\lambda_{i}^{p,TTS}\right)=\min \left(\lambda_{i}^{v,TTS}+F_{i,in}p_i^{p,TTS}, p_i^{TTS} \lambda_{i}^{p}\right)\label{eq:se}\\
\lambda_{i}^{pv,ATS}=D_i^{ATS}&=\min \left(\hat{\lambda}_i^{v,ATS},\lambda_{i}^{p,ATS}\right)=\min \left(\lambda_{i}^{v,ATS}+F_{i,in}p_i^{p,ATS}, p_i^{ATS} \lambda_{i}^{p}\right)
\end{align}
For each road subsystem $i \in I$,
\begin{align}
F_{i,out}=\lambda_{i}^{r}&= D_{i}^{TTS}+D_{i}^{ATS}+\left(1-p_i^{p,ATS}-p_i^{p,TTS}\right) F_{i,in}\nonumber\\
&=\min \left(\lambda_{i}^{v,TTS}+F_{i,in}p_i^{p,TTS}, p_i^{TTS} \lambda_{i}^{p}\right)+\min \left(\lambda_{i}^{v,ATS}+F_{i,in}p_i^{p,ATS}, p_i^{ATS} \lambda_{i}^{p}\right)\nonumber \\
&+\left(1-p_i^{p,ATS}-p_i^{p,TTS}\right) F_{i,in} \label{flow_out}
\end{align}

Substitute equation \ref{flow_out} into \ref{in_flows},
\begin{align}
F_{i,in}&= \sum_{j\in I}a_{ji} ( \min (\lambda_{j}^{v,TTS}+F_{j,in}p_j^{p,TTS}, p_j^{TTS} \lambda_{j}^{p})+\min (\lambda_{j}^{v,ATS}+F_{j,in}p_j^{p,ATS}, p_j^{ATS} \lambda_{j}^{p}) \nonumber\\
&+(1-p_j^{p,ATS}-p_j^{p,TTS}) F_{j,in} )( p_{ji}^{O,ATS}p_{j,O}^{r,ATS} + p_{ji}^{E,ATS}p_{j,E}^{r,ATS} +p_{ji}^{O,TTS}p_{j,O}^{r,TTS}+ p_{ji}^{E,TTS}p_{j,E}^{r,TTS}) \label{eq:system_single}
\end{align}

The above equation \ref{eq:system_single} is available for all $|I|$ spatial units, which forms a system of $|I|$ equations. Solving such equation system leads to incoming flows, as well as effective arrival rate for both taxi and road queues. However, we introduce following inequalities to the equation system and convert the problem into a linear programming one, due to the existence of min sets. These inequalities are strictly consistent with min sets in equation \ref{eq:system_single}.     
\begin{align}
&\lambda_{i}^{pv,TTS}\geq \lambda_{i}^{v,TTS}+F_{i,in}p_i^{p,TTS} \label{eq:10}\\
& \lambda_{i}^{pv,TTS}\geq p_i^{TTS} \lambda_{i}^{p}\\
&\lambda_{i}^{pv,TTS}\leq \lambda_{i}^{v,TTS}+F_{i,in}p_i^{p,TTS}\\
& \lambda_{i}^{pv,TTS}\leq p_i^{TTS} \lambda_{i}^{p}\\
&\lambda_{i}^{pv,ATS}\geq \lambda_{i}^{v,ATS}+F_{i,in}p_i^{p,ATS}\\
&\lambda_{i}^{pv,ATS}\geq p_i^{ATS} \lambda_{i}^{p} \\
&\lambda_{i}^{pv,ATS}\leq \lambda_{i}^{v,ATS}+F_{i,in}p_i^{p,ATS}\\
&\lambda_{i}^{pv,ATS}\leq p_i^{ATS} \lambda_{i}^{p} \label{eq:13}
\end{align} 

Therefore, we can convert the effective arrival rate computations into a linear programming problem that can be solved in polynomial time and maintain same solutions:
\begin{align}\label{eq:linear}
&\min_{\{\lambda_{i}^{pv,TTS},\lambda_{i}^{pv,ATS}, F_{i,in}\}_{i\in I}} \sum_{i\in I} \left(\lambda_{i}^{pv,TTS}+\lambda_{i}^{pv,ATS}\right) \\
\text{Subject to}& \nonumber\\
&\text{equations}\ \ref{eq:10} \text{to} \ref{eq:13}\nonumber\text{, for every}\ i\in I\nonumber\\
&\text{equations}\ \ref{eq:system_single}\nonumber\text{, for every}\ i\in I\nonumber\\
&  \lambda_{i}^{pv,TTS}\geq 0,\lambda_{i}^{pv,ATS}\geq 0, F_{i,in}\geq 0 \text{, for every}\ i\in I\nonumber
\end{align}

\subsection{Stationary State Distribution of Queueing Network}
Recall the subnetwork in Fig.\ref{fig:subnetwork}, the taxi system in each spatial unit behaves as an independent system of one or multiple (depending on the scale of spatial units) set of two parallel $M/M/1$ queues and one $M/M/c$ queue. We can prove the existence of a steady-state distribution for the subnetwork and derive.   

\begin{theorem}
	If we have $\lambda_i^{pv,TTS}<\mu_{i}^{TTS}$, $\lambda_i^{pv,ATS}<\mu_{i}^{ATS}$, and $\lambda_{i}^r<c\mu_{i}^{r}$. Further, for the state $X=\left\{x_i^{TTS}, x_i^{ATS}, x_i^{O,TTS}, x_i^{E,TTS}, x_i^{O,ATS}, x_i^{E,ATS}\right\}$, the steady state probability is given by:
	\begin{equation}\label{eq:part_state}
	\pi(X)=\begin{cases}
	\left(\frac{\lambda_{i}^{pv,TTS}}{\mu_{i}^{TTS}}\right)^{x_i^{TTS}} \left(\frac{\lambda_{i}^{pv,ATS}}{\mu_{i}^{ATS}}\right)^{x_i^{ATS}} \frac{1}{x_i!}\left(\frac{\lambda_{i}^r}{\mu_{i}^{r}} \right)^{x_i} \pi(\phi) & \text{if}\ 0\leq x_i<c_i \\
	\left(\frac{\lambda_{i}^{pv,TTS}}{\mu_{i}^{TTS}}\right)^{x_i^{TTS}}\left(\frac{\lambda_{i}^{pv,ATS}}{\mu_{i}^{ATS}}\right)^{x_i^{ATS}} \frac{1}{c_i^{x_i-c_i}c_i!}\left(\frac{\lambda_{i}^r}{\mu_{i}^r} \right)^{x_i} \pi(\phi) & \text{if}\ x_i\geq c_i
	\end{cases}
	\end{equation}
	where, $x_i=x_i^{O,TTS}+x_i^{E,TTS}+x_i^{O,ATS}+x_i^{E,ATS}$,\\ $\pi(\phi)=\left(1-\frac{\lambda_{i}^{pv,TTS}}{\mu_{i}^{TTS}}\right)\left(1-\frac{\lambda_{i}^{pv,ATS}}{\mu_{i}^{ATS}}\right)\left(\frac{\mu_{i}}{(c_i-1)!(c_i\mu_{i}^r-\lambda_{i}^r)}\left( \frac{\lambda_{i}^r}{\mu_{i}^r}\right)^{c_i}+\sum_{n=0}^{c_i-1}\frac{1}{n!}\left(\frac{\lambda_{i}^r}{\mu_{i}^r}^n\right)\right)^{-1}$
\end{theorem}

\begin{proof}
	Following the theorem 1.13 in \cite{kelly}, we prove the stationary state distribution for the subnetwork as follows: Let $X(t)$ be a stationary Markov process with transition rates $q(m,n)$, where, $m,n$ are two system states. If we can find a collection of numbers $q^{\prime}(m,n)$, such that $q^{\prime}(m)=q(m)$ and a collection of positive numbers $\pi(m)$, summing to unity, such that $\pi(m)q(m,n)=\pi(n) q^{\prime}(n,m)$, then $q^{\prime}(n,m)$ are the transition rates of the reversed process $X(\tau-t)$ and $\pi(m)$ is the equilibrium distribution of both processes.
	
	First, given the state $m\vcentcolon=\left\{x_i^{TTS}, x_i^{ATS}, x_i^{O,TTS}, x_i^{E,TTS}, x_i^{O,ATS}, x_i^{E,ATS}\right\}$, we can enumerate the system states and define the rates of reversed process: (1)The one TTS (similar for ATS) arrival at taxi queue $i$ yields the state $n\vcentcolon=\left\{x_i^{TTS}+1, x_i^{ATS}, x_i^{O,TTS}, x_i^{E,TTS}, x_i^{O,ATS}, x_i^{E,ATS}\right\}$;   
	$$q(m,n)=\lambda_{i}^{pv,TTS}\ ,  q^{\prime}(n,m)=\mu_{i}^{TTS}.$$
	
	(2) One occupied vehicle departing from TTS $i$ (similar for ATS) and arriving at corresponding road queue yield the state $n\vcentcolon=\left\{x_i^{TTS}-1, x_i^{ATS}, x_i^{O,TTS}+1, x_i^{E,TTS}, x_i^{O,ATS}, x_i^{E,ATS}\right\}$;
	$$q(m,n)=\mu_{i}^{TTS}\ , \ q^{\prime}(n,m)= (x_i+1)\lambda_{i}^{pv,TTS} \mu_{i}^r/\lambda_{i}^r\ \text{if}\ 0\leq x_i<c_i\ \text{or}\ c_i\lambda_{i}^{pv,TTS}\mu_{i}^r/\lambda_{i}^r\ \text{if}\ x_i\geq c_i.$$ 
	
	(3) One vehicle departing from one road queue $i$ and arriving at another road queue $j$ yields the state $n\vcentcolon=\left\{x_i-1, x_j+1\right\}$;
	\begin{align*}
	&q(m,n)=x_i\mu_{i}^r p_{ij}\ ,q^{\prime}(n,m)=\lambda_{i}^r p_{ij}(x_j+1) \mu_{j}^r/\lambda_{j}^r\    \text{if}\ 0\leq x_i<c_i\ \text{and}\ 0\leq x_j< c_j\ \text{or,}\\
	&q(m,n)=x_i\mu_{i}^r p_{ij}\ ,q^{\prime}(n,m)=\lambda_{i}^r p_{ij}c_j \mu_{j}^r/\lambda_{j}^r\    \text{if}\ 0\leq x_i<c_i\ \text{and}\ x_j\geq c_j \ \text{or,}\\ 
	&q(m,n)= c_i\mu_{i}^r p_{ij}\ ,q^{\prime}(n,m)=\lambda_{i}^r p_{ij} (x_j+1) \mu_{j}^r/\lambda_{j}^r\  \text{if}\ x_i\geq c_i\ \text{and}\ 0\leq x_j< c_j\ \text{or,}\\
	&q(m,n)= c_i\mu_{i}^r p_{ij}\ ,q^{\prime}(n,m)=\lambda_{i}^R p_{ij}c_j \mu_{j}^r/\lambda_{j}^r\ \text{if}\ x_i\geq c_i\ \text{and}\ x_j\geq c_j.
	\end{align*} 
	
	(4) One vehicle departing from one road queue $j$ and arriving at one TTS (similar for ATS) queue $i$ yields the state $n\vcentcolon=\left\{x_i^{TTS}+1, x_i^{ATS}, x_j-1\right\}$;  
	\begin{align*}
	&q(m,n)= x_j\mu_{j}^r p_{ji}\ ,q^{\prime}(n,m)=\lambda_{j}^r p_{ji} \mu_{i}^{TTS}/\lambda_{i}^{pv,TTS}\    \text{if}\ 0\leq x_j<c_j\  \text{or, }\\ 
	& q(m,n)= c_j\mu_{j}^r p_{ji}\ ,q^{\prime}(n,m)=\lambda_{j}^r p_{ji} \mu_{i}^{TTS}/\lambda_{i}^{pv,TTS}\ \text{if}\ x_j\geq c_j.
	\end{align*} 
	
	(5) One vehicle getting destination and departing the system immediately after road queue $i$ yields the state $n\vcentcolon=\left\{x_i^{TTS}, x_i^{ATS}, x_i-1\right\}$;
	\begin{align*}
	&q(m,n)= x_i\mu_{i}^r p_{i0}\ ,q^{\prime}(n,m)=\lambda_{i}^r p_{i0}\    \text{if}\ 0\leq x_i<c_i\  \text{or,}\\ &q(m,n)= c_i\mu_{i}^r\ ,q^{\prime}(n,m)=\lambda_{i}^r p_{i0}\ \text{if}\ x_i\geq c_i.
	\end{align*} 
	With the above 5 state transitions, it can now be easily checked that for any two states $i,j$, the equation \ref{eq:part_state} is satisfied. 
\end{proof}

\begin{theorem}
	If for every taxi or road subsystem $i\in I$, we have $\lambda_i^{pv,TTS}<\mu_{i}^{TTS}$, $\lambda_i^{pv,ATS}<\mu_{i}^{ATS}$, and $\lambda_{i}^r<c_i\mu_{i}^r$. Further for the state\\ $X=\left\{x_i^{TTS},x_i^{ATS}, x_i^{O,TTS},x_i^{E,TTS},x_i^{O,ATS},x_i^{E,ATS}\right\}_{i=1}^{|I|}$, the steady state probability is given by:
	\begin{equation}
	\pi(X)=\begin{cases}
	\prod_{i\in I}\left(\frac{\lambda_{i}^{pv,TTS}}{\mu_{i}^{TTS}}\right)^{x_i^{TTS}} \left(\frac{\lambda_{i}^{pv,ATS}}{\mu_{i}^{ATS}}\right)^{x_i^{ATS}}\prod_{i\in I} \frac{1}{x_i!}\left(\frac{\lambda_{i}^r}{\mu_{i}^r} \right)^{x_i} \pi(\phi) & \text{if}\ 0\leq x_i<c_i \\
	\prod_{i\in I}\left(\frac{\lambda_{i}^{pv,TTS}}{\mu_{i}^{TTS}}\right)^{x_i^{TTS}}\left(\frac{\lambda_{i}^{pv,ATS}}{\mu_{i}^{ATS}}\right)^{x_i^{ATS}}\prod_{i\in I} \frac{1}{c_i^{x_i-c_i}c_i!}\left(\frac{\lambda_{i}^r}{\mu_{i}^r} \right)^{x_i} \pi(\phi) & \text{if}\ x_i\geq c_i
	\end{cases}
	\end{equation}
	where, $x_i=x_i^{O,TTS}+x_i^{E,TTS}+x_i^{O,ATS}+x_i^{E,ATS}$,\\ $\pi(\phi)=\prod_{i\in I}\left(1-\frac{\lambda_{i}^{pv,TTS}}{\mu_{i}^{TTS}}\right)\left(1-\frac{\lambda_{i}^{pv,ATS}}{\mu_{i}^{ATS}}\right)\prod_{i\in I}\left(\frac{\mu_{i}^r}{(c_i-1)!(c_i\mu_{i}^r-\lambda_{i}^r)}\left( \frac{\lambda_{i}^r}{\mu_{i}^r}\right)^{c_i}+\sum_{n=0}^{c_i-1}\frac{1}{n!}\left(\frac{\lambda_{i}^r}{\mu_{i}^r}^n\right)\right)^{-1}$
\end{theorem}
The proof for theorem 1 is developed for the subsystem of queueing network, consisting of state transitions between taxi and road queues within one spatial unit. Then it can be easier to extend the proof to the whole queuing network, since the routing process over network is based on a fixed probability matrix or a Bernoulli splitting process. It is straightforward that the stationary state distribution of the proposed queueing network is the product of the stationary state distribution of the subsystem. This is also one classical proof in the literature, thus not presented here. 

\subsection{Performance Metrics} 
The $M/M/1$ and $M/M/c$ queues have well-defined performance metrics under stationary distribution. For example, we have the queue server utilization rate $\rho_i^{ATS}, \rho_i^{TTS}$, the expected number of vehicles at one queue $L_i^{ATS}, L_i^{TTS}$, expected number of waiting passenger-vehicle pairs $L_i^{q,ATS}, L_i^{q,TTS}$, expected sojourn time $W_i^{ATS}, W_i^{TTS}$, and expected waiting time in queue $W_i^{q,ATS}, W_i^{q,TTS}$ as follows:
\begin{align*}
&\rho_i^{ATS}=\frac{\lambda_i^{pv,ATS}}{\mu_i^{ATS}}\ , \rho_i^{TTS}=\frac{\lambda_i^{pv,TTS}}{\mu_i^{TTS}}\\
&L_i^{ATS}=\frac{\lambda_i^{pv,ATS}}{\mu_i^{ATS}-\lambda_i^{pv,ATS}} \ ,L_i^{TTS}=\frac{\lambda_i^{pv,TTS}}{\mu_i^{TTS}-\lambda_i^{pv,TTS}} \\
&L_i^{q,ATS}=\frac{(\lambda_i^{pv,ATS})^2}{\mu_i^{ATS}(\mu_i^{ATS}-\lambda_i^{pv,ATS})} \ , L_i^{q,TTS}=\frac{(\lambda_i^{pv,TTS})^2}{\mu_i^{TTS}(\mu_i^{TTS}-\lambda_i^{pv,TTS})} \\
&W_i^{ATS}=\frac{1}{\mu_i^{ATS}-\lambda_i^{pv,ATS}}\ , W_i^{TTS}=\frac{1}{\mu_i^{TTS}-\lambda_i^{pv,TTS}}  \\ &W_i^{q,ATS}=\frac{\lambda_i^{pv,ATS}}{\mu_i^{ATS}(\mu_i^{ATS}-\lambda_i^{pv,ATS})} \ , W_i^{q,TTS}=\frac{\lambda_i^{pv,TTS}}{\mu_i^{TTS}(\mu_i^{TTS}-\lambda_i^{pv,TTS})}
\end{align*}  
For each road queue $i\in I$, given $\lambda_{i}^r$ and $\mu_{i}^r$, we can also derive similar system performance metrics as follows:  

\begin{align*} 
&\rho_i^r=\frac{\lambda_{i}^r}{c_i\mu_{i}^r} \ ,
L_i^{r}=c_i \rho_i^r + \left(\frac{(c_i \rho_i^r)^{c_i}\rho_i^r}{{c_i}!(1-\rho_i^r)^2}\right)p_i(0) \ ,
L_i^{q,r}=\left(\frac{(c_i \rho_i^r)^{c_i}\rho_i^r}{{c_i}!(1-\rho_i^r)^2}\right)p_i(0) \\
&W_i^{r}=\frac{1}{\mu_{i}^r}+ \left(\frac{(c_i \rho_i^r)^{c_i}\rho_i^r}{c_i!c_i\mu_{i}^r(1-\rho_i^r)^2}\right)p_i(0) \ ,
W_i^{q,r}=\left(\frac{(c_i \rho_i^r)^{c_i}\rho_i^k}{c_i!c_i\mu_{i}^r(1-\rho_i^r)^2}\right)p_i(0)
\end{align*}
where, $\rho_i^r$ is the utilization rate of road queue, $L_i^{r}$ is the expected queue length of the system, $L_i^{q,r}$ is the expected queue length waiting for service; $W_i^{r}$ is the expected sojourn time; and $W_i^{q,r}$ is the expected waiting time before service begins; and $p_i(0)$ is the probability of empty queues, derived from $\left(\frac{r^{c_k}}{c_k!(1-\rho)+\sum_{n=0}^{c_k-1}\frac{r^n}{n!}}\right)^{-1}$.

Within the proposed queueing network consisting of taxi queues $M/M/1$ and road queues $M/M/c$, we have average number of vehicles in network $L_I$, total average load on network $\gamma_I$, and average delay throughout the network $W_I$.  
\begin{align*}
&L_I=\sum_{i\in I}L_i^{r}+\sum_{i\in I}(L_i^{ATS}+L_i^{TTS}) \  ,
\gamma_I=\sum_{i\in I} \left(\hat{\lambda}_{i}^{v,TTS}+\hat{\lambda}_{i}^{v,ATS}\right) \   ,
W_I=\frac{L_I}{\gamma_I}
\end{align*}

\section{Case Study}
\subsection{The Case of New York City}
In this section, we apply our proposed queueing network into the competitive taxi market of New York City (NYC), where has one of the largest TTS (fleet size of more than 13,000 yellow taxicabs) and ATS market (weekly active Uber drivers more than 45,000) in North America, as of April 2017. The case study is developed with multiple datasets: (a) the ride records of both yellow taxicabs and for-hire vehicle (Uber) shared by the NYC Taxi \& Limousine Commission (ref:\url{http://www.nyc.gov/html/tlc/html/about/trip_record_data.shtml}), containing a time-stamped trip record of 6-tuples $(O, D, t, tt, d, F)$, which denote the geolocation of origin and destination, pickup time, in-vehicle travel time, trip distance, and total fares respectively; (b) the Uber trajectory and operation data directly crawled from Uber platform with high-frequency (every 5s) empty vehicle trajectory information, as well as surge pricing and estimated waiting time at 470 specific locations \cite{zhang_pricing}; (c) the city traffic flow data with 15-minute short count link volume and speed data in one specific week of each year, shared by New York Department of Transportation (ref:\url{https://www.dot.ny.gov/divisions/engineering/technical-services/highway-data-services/hdsb}); and (d) the Google Maps Directions API for shortest route planning between the pair of locations (ref:\url{https://developers.google.com/maps/documentation/directions/intro}).     

This study also introduces the idea of a `homogeneous region,' which is defined based on the Poisson assumption on the passenger and vehicle arrival process. We perform extensive hypothesis tests of the Poisson assumption under different spatial scales (e.g. Borough\footnote{5 in total and each is with an average area of 60.4 mi$^2$, see \url{https://en.wikipedia.org/wiki/Boroughs_of_New_York_City}}, community districts\footnote{71 in total if we include regions of airports and parks, and each is with an average of $4.3 $mi$^2$, see \url{https://www1.nyc.gov/site/planning/community/community-portal.page} }, zip code tabulation area [ZCTA] \footnote{214 in total, and each is with an average of 1.4 mi$^2$, see \url{https://www.census.gov/geo/reference/zctas.html}}, and census tract \footnote{2165 in total, and each is with an average of 0.14 mi$^2$, see \url{https://www.census.gov/geo/reference/gtc/gtc_ct.html}}), as well as arrival count interval (varying from 1 minute to 1 hour) and study period (time-of-the-day and day-of-the-week). The Kolmogorov Smirnov (KS) test is adapted for discrete distribution\cite{correct_ks} and is used to test whether the passenger or vehicle arrivals can be assumed to be Poisson distributed. In addition, three additional $\chi^2$ distribution based tests~\cite{poisson_test} (e.g. Anscombe, Likelihood, Conditional) are applied to test whether the arrivals follow a Poisson distribution. The key test results for passengers and vehicle arrivals during peak hours are summarized in Fig. \ref{arrivals} and \ref{vehicle_arrivals}. Note that we skip several results due to space limits and share full results as appendix at the end of this manuscript. To sum up, we have the following findings for case study development: a) the study period should be limited to one hour peak (from 6 to 7pm) or off peak (from 10 to 11 am) of every Mondays to Thursdays, which leads to more spatial units holding Poisson arrival assumptions; and b) we should aggregate trip records at community districts (71 in total, $\sim$4.3 mi$^2$ on average per community district) and 1-minute count interval, which has higher probability of being in line with our Poisson arrival flow assumptions. 
\begin{figure}[!htbp]
	\centerline{\subfigure[borough 1-hour]{\includegraphics[width=0.335 \columnwidth]{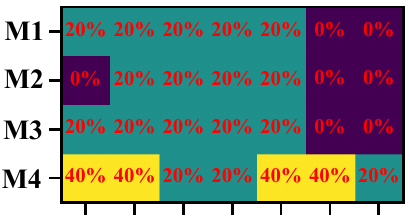}}	
		\subfigure[ZCTA 1-hour]{ \includegraphics[width=0.3 \columnwidth]{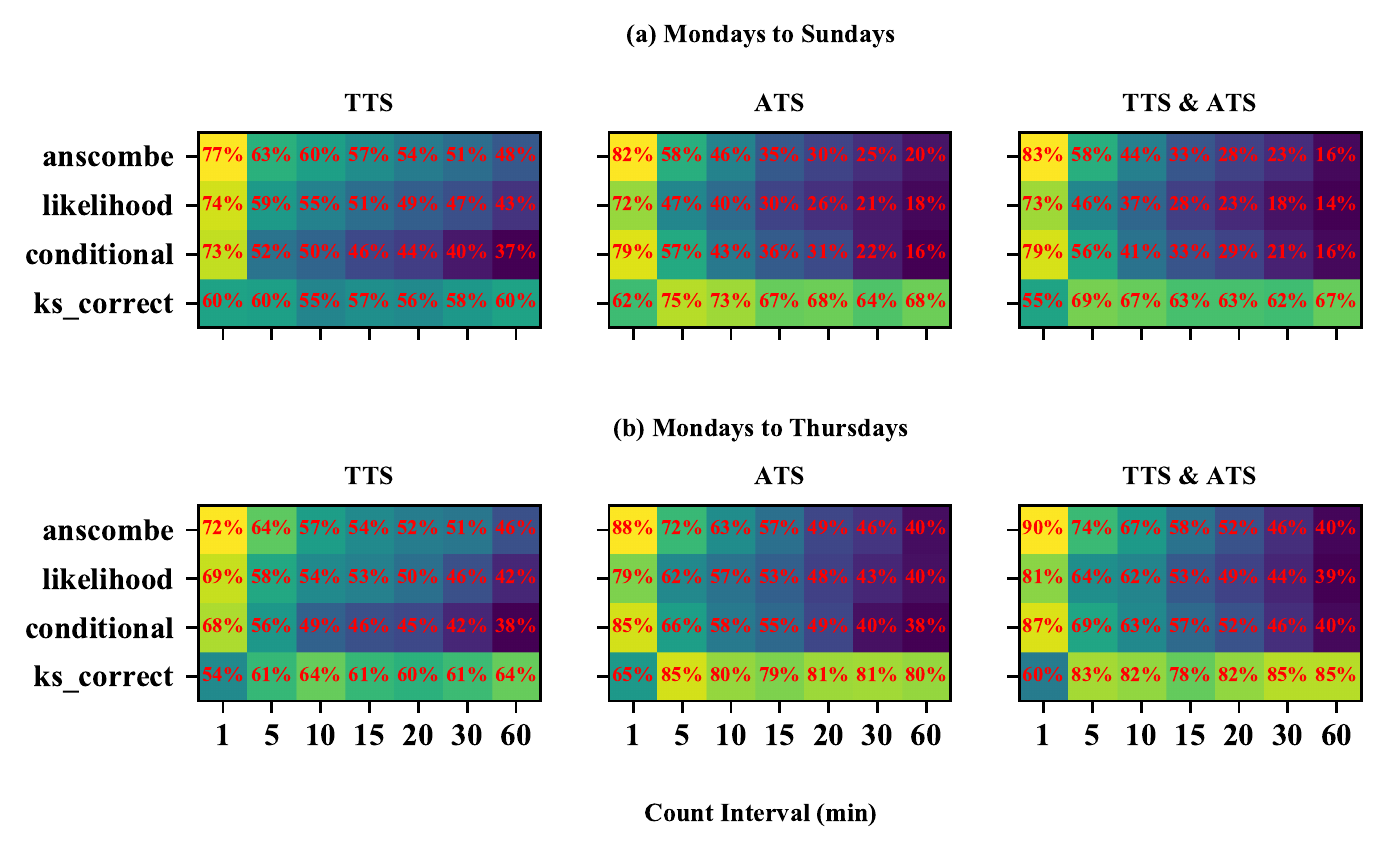}}
		\subfigure[census tract 1-hour]{ \includegraphics[width=0.3 \columnwidth]{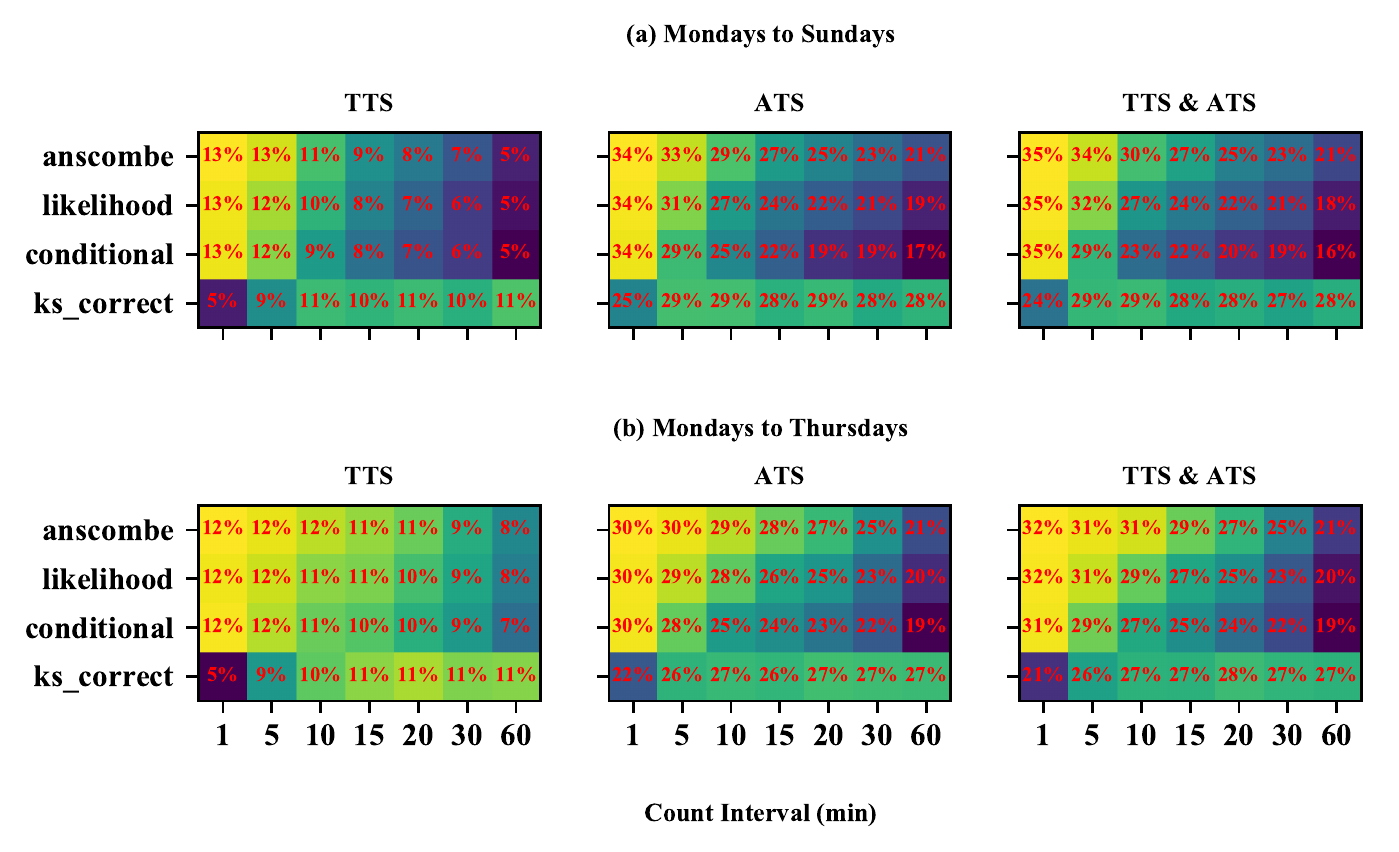}}
	}
	\centerline{ \subfigure[community districts 3-hour]{\includegraphics[width=0.335 \columnwidth]{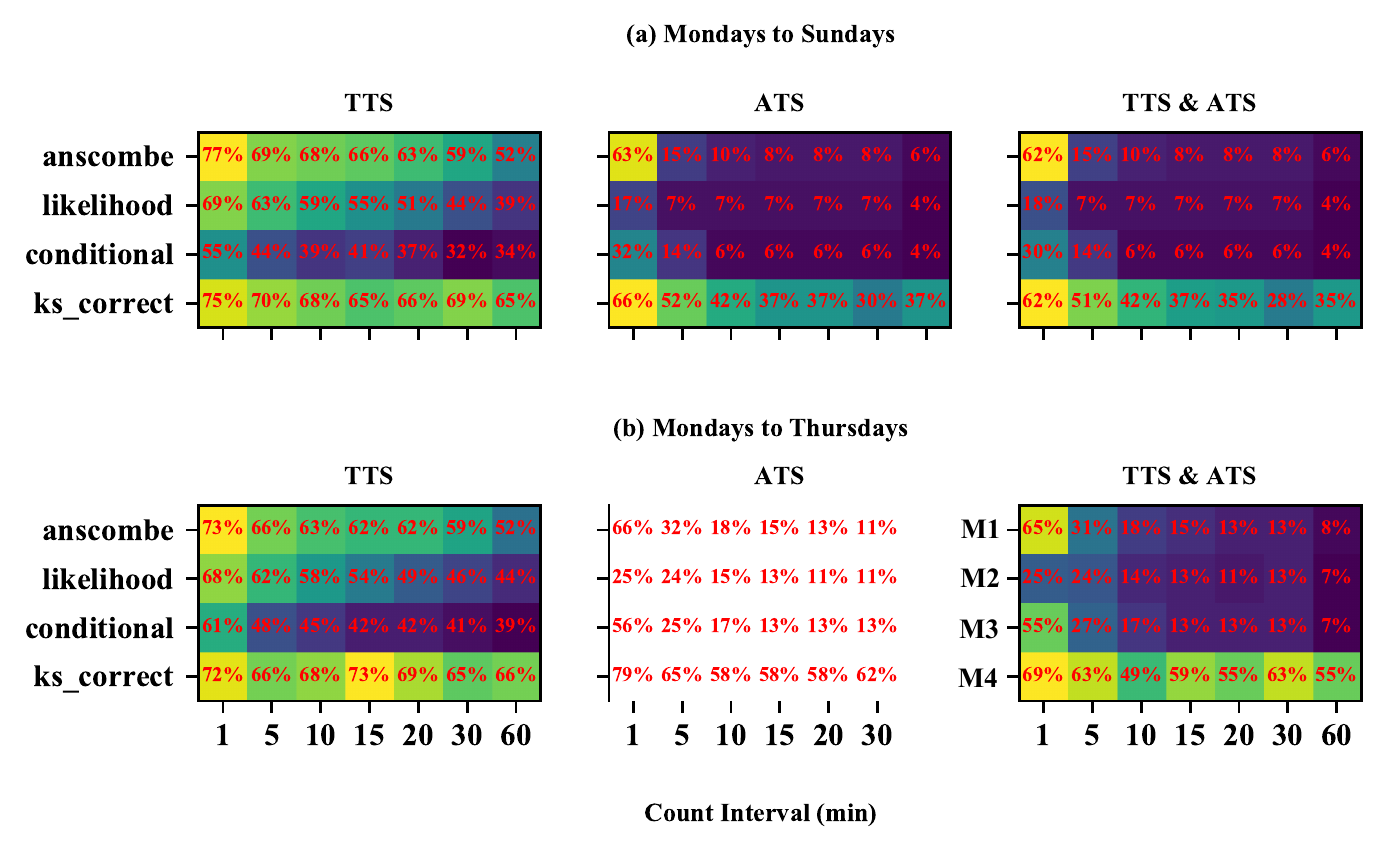}}
		\subfigure[community districts all days]{ \includegraphics[width=0.3 \columnwidth]{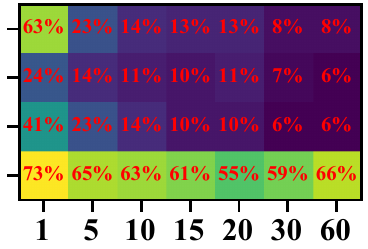} }
		\subfigure[community districts 1-hour]{\includegraphics[width=0.3 \columnwidth]{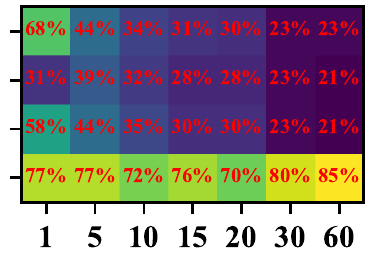}}
	}
	\caption{The percentage of community districts with significant hypothesis testing results under various scenarios: (a) borough and 1-hour peak of Mondays to Thursdays; (b) zcta and 1-hour peak of Mondays to Thursdays; (c) census tract and 1-hour peak of Mondays to Thursdays; (d) community districts and 3-hour peak of Mondays to Thursdays; (e) community districts and 1-hour peak of all days; and (f) community districts and 1-hour peak of Mondays to Thursdays. Moreover, x-axis indicates the arrival count interval (minutes), and y-axis denotes test methods (M1: Anscombe, M2: Likelihood, M3: Conditional, and M4: adapted KS).}
	\label{arrivals}
\end{figure}
\begin{figure}[!htbp]
	\centerline{ \subfigure[TTS vehicle arrivals]{\includegraphics[width=0.465 \columnwidth]{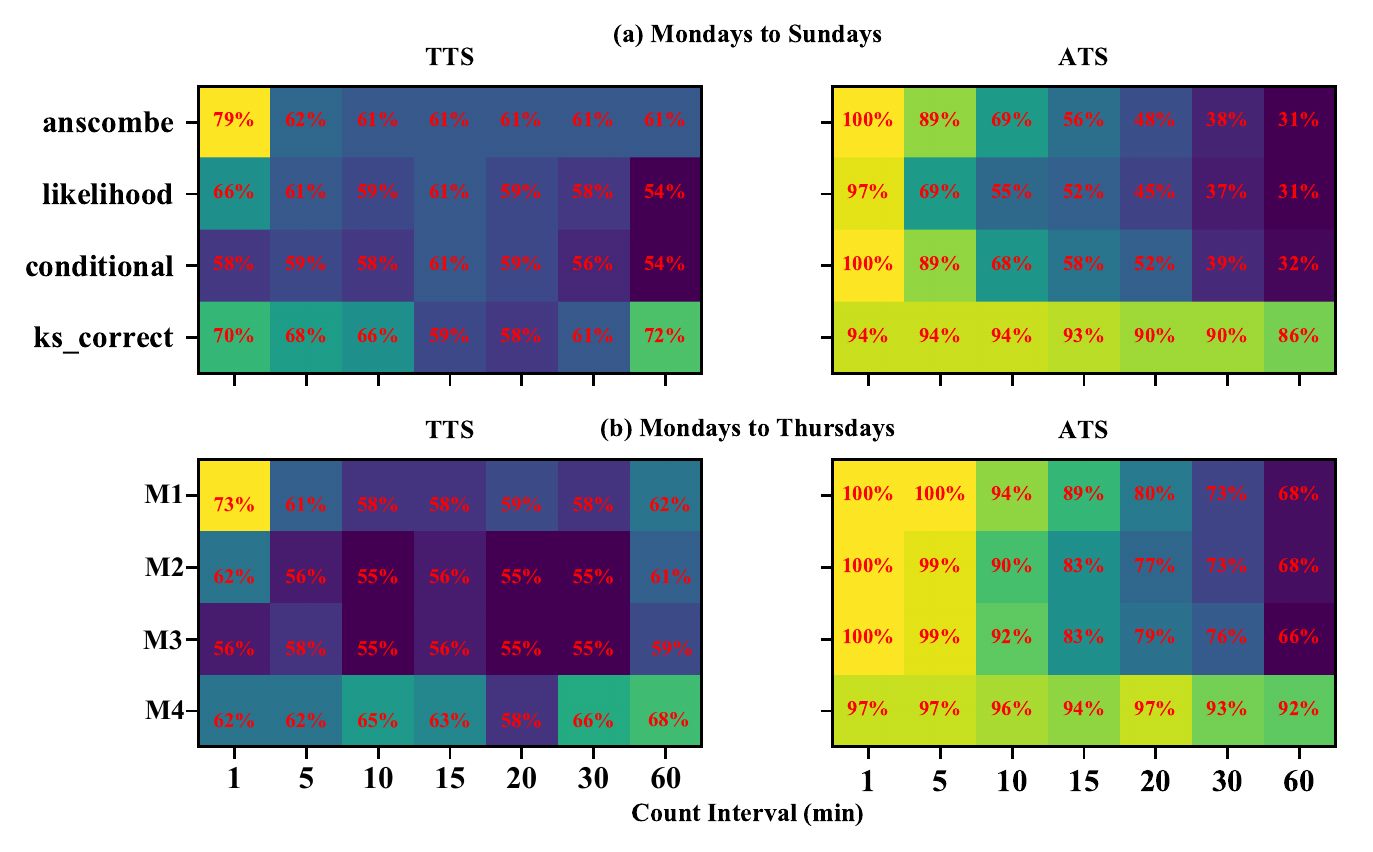}}
		\subfigure[ATS vehicle arrivals]{\includegraphics[width=0.44 \columnwidth]{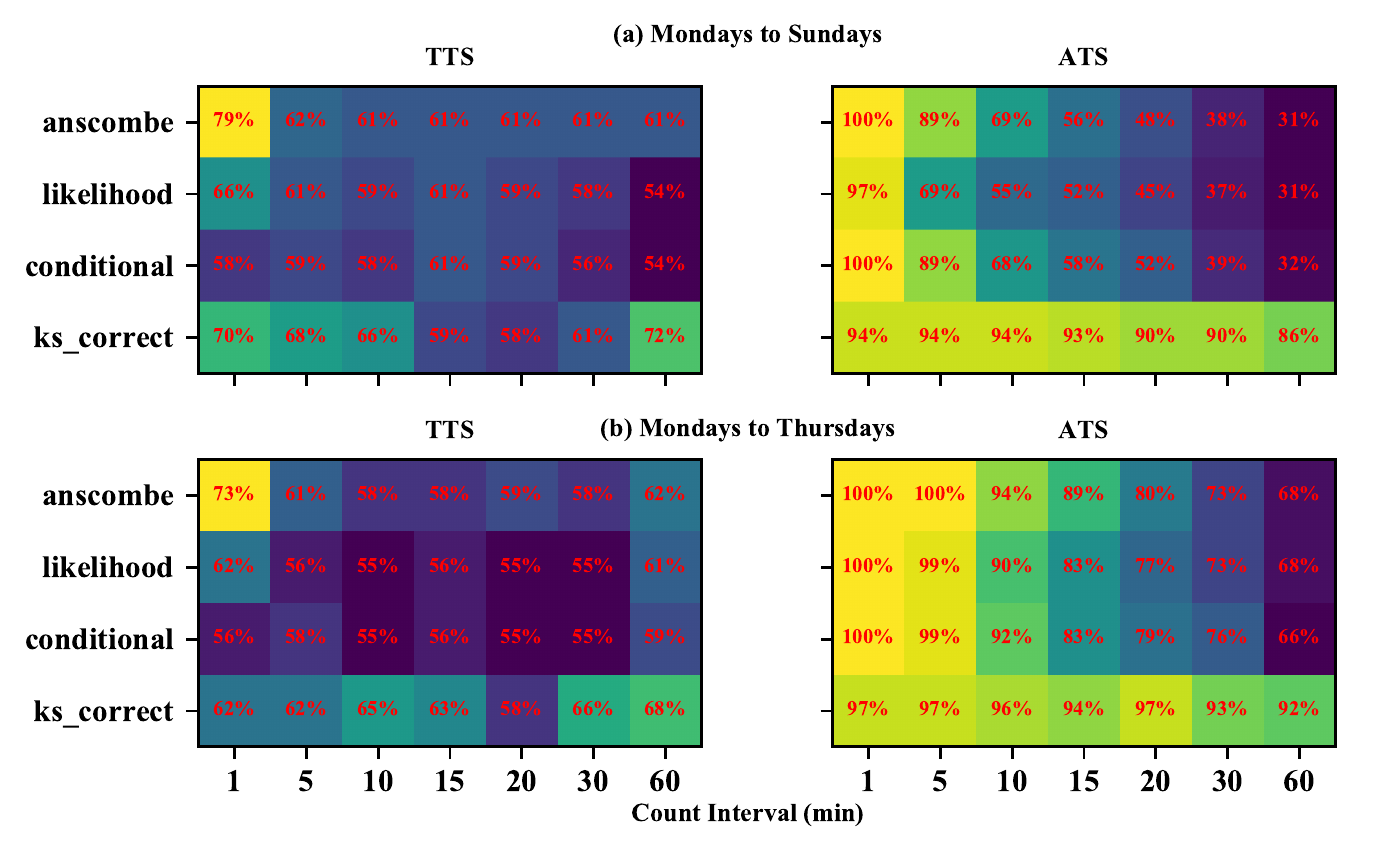}    }}
	\caption{The percentage of community districts with significant hypothesis testing results. X-axis indicates the arrival count interval (minutes), and y-axis denotes test methods (M1: Anscombe, M2: Likelihood, M3: Conditional, and M4: adapted KS). }
	\label{vehicle_arrivals}
\end{figure}

\subsection{Queue Inputs}
Under proposed spatiotemporal aggregation scales, we further investigate the passenger and vehicle arrival rates $\lambda_{i}^{p}, \lambda_{i}^{v,ATS}, \lambda_{i}^{v,TTS}$. Fig. \ref{input_arrivals} (a) to (c) exhibits the p values from hypothesis testing in each community district. The red color indicates small p values less than 0.05, which reject the null hypothesis of Poisson arrivals at confidence level of 95\%. It is straightforward that most community districts can be assumed to have Poisson passenger and vehicle arrivals. Fig. \ref{input_arrivals} (d) to (f) show corresponding arrival rates. Downtown Manhattan areas have relatively higher passenger arrival rates of more than 100 passengers per minute. In contrast, remote areas have much fewer arrivals, revealing a significantly imbalanced distribution of passengers. On the other hand, either ATS or TTS vehicle arrival rates are less than 5 vehicles per minute and do not have many variance across space, since the vehicle arrival only counts new onlines and excludes those who have been in system and move from other community districts.
\begin{figure}[!htbp]
	\centerline{\includegraphics[width=0.7\columnwidth]{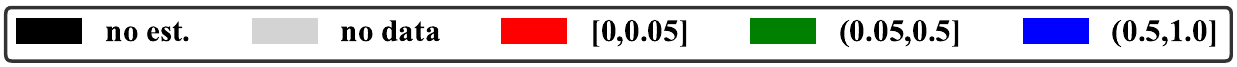}}
	\centerline{\subfigure[passenger arrivals]{\includegraphics[width=0.275 \columnwidth]{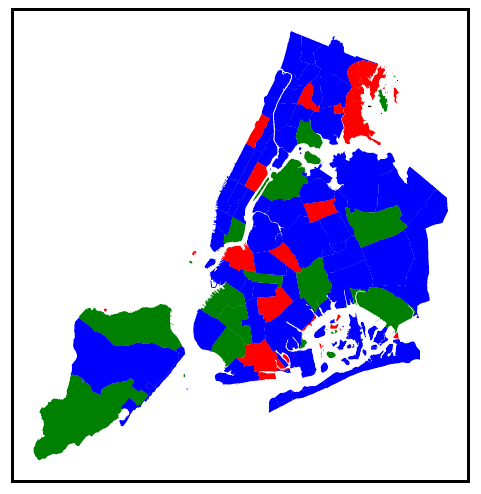}}	
		\subfigure[TTS vehicle new arrivals]{ \includegraphics[width=0.295 \columnwidth]{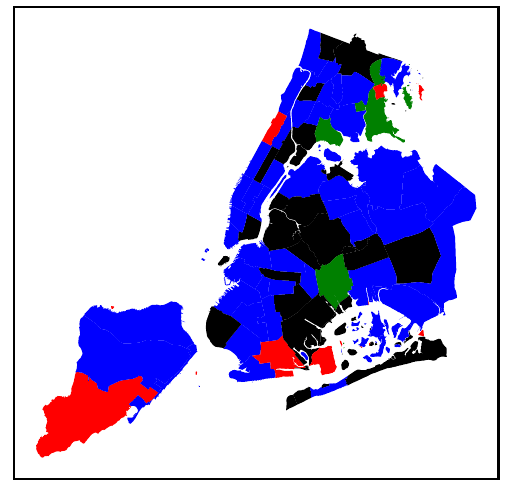}}
		\subfigure[ATS vehicle new arrivals]{ \includegraphics[width=0.295 \columnwidth]{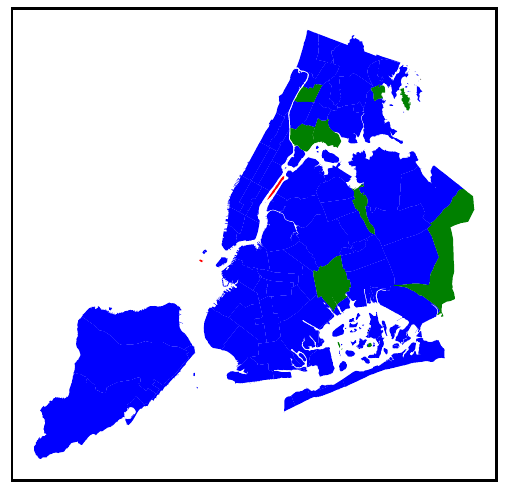}}
	}
	\centering
	 \subfigure[$\lambda_{i}^{p}$]{\includegraphics[width=0.51 \columnwidth]{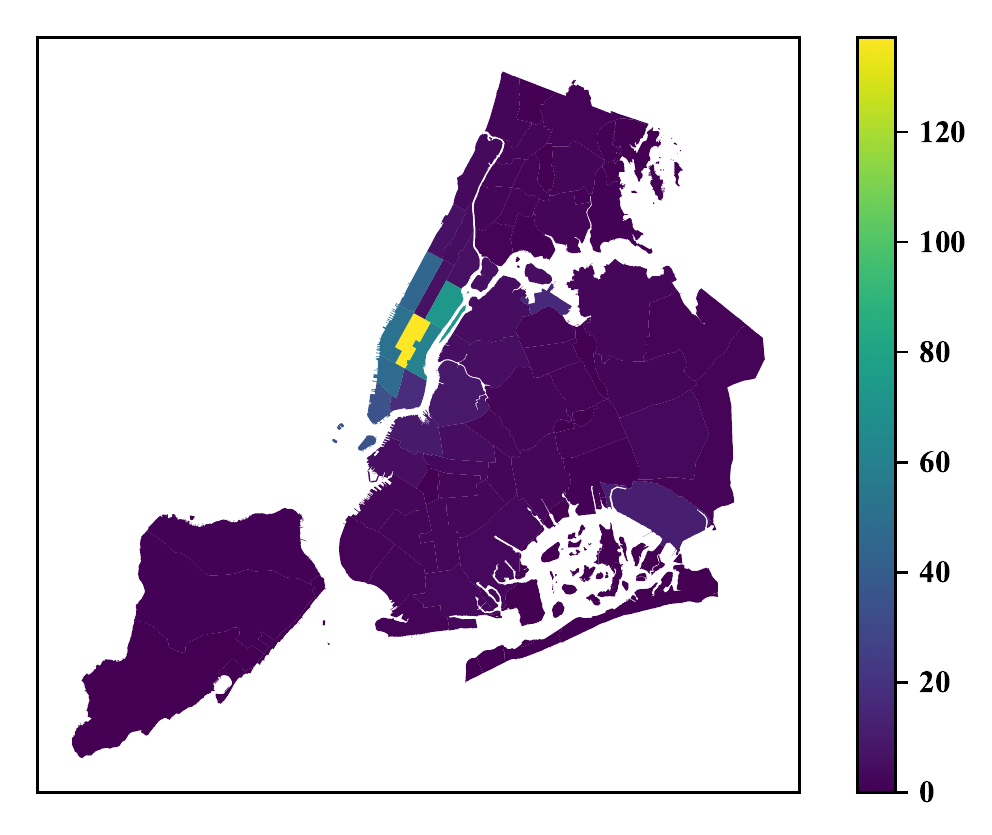}}
	 \centerline{
		\subfigure[$\lambda_{i}^{v,TTS}$]{ \includegraphics[width=0.49 \columnwidth]{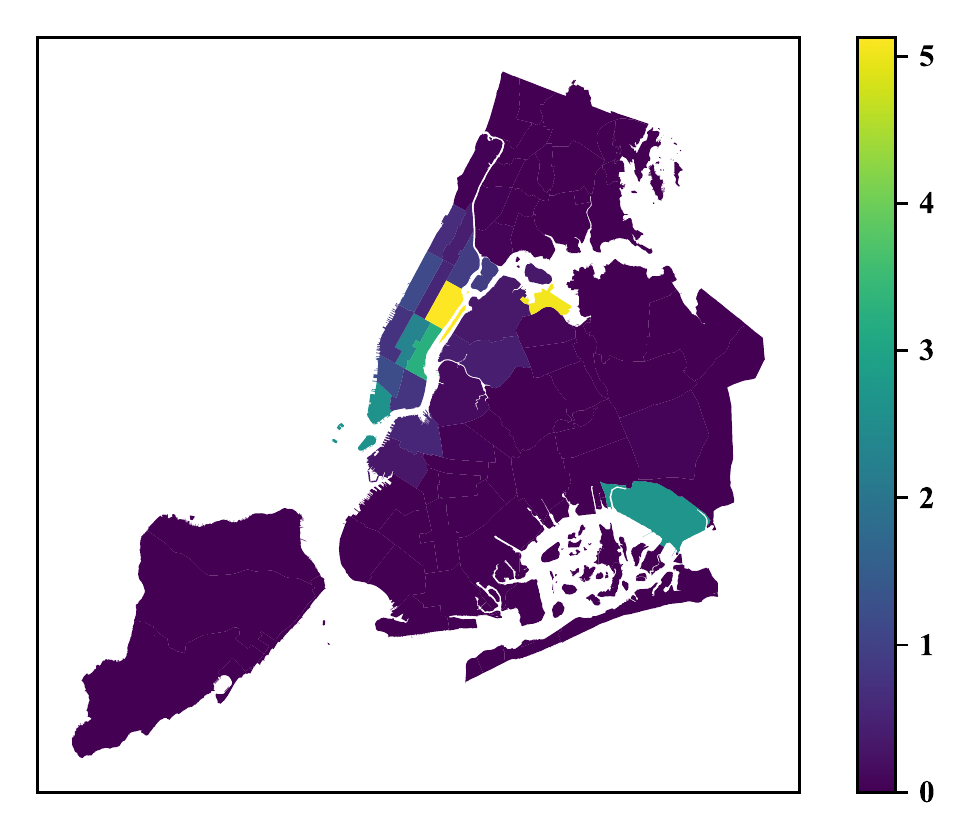} }
		\subfigure[$\lambda_{i}^{v,ATS}$]{\includegraphics[width=0.5 \columnwidth]{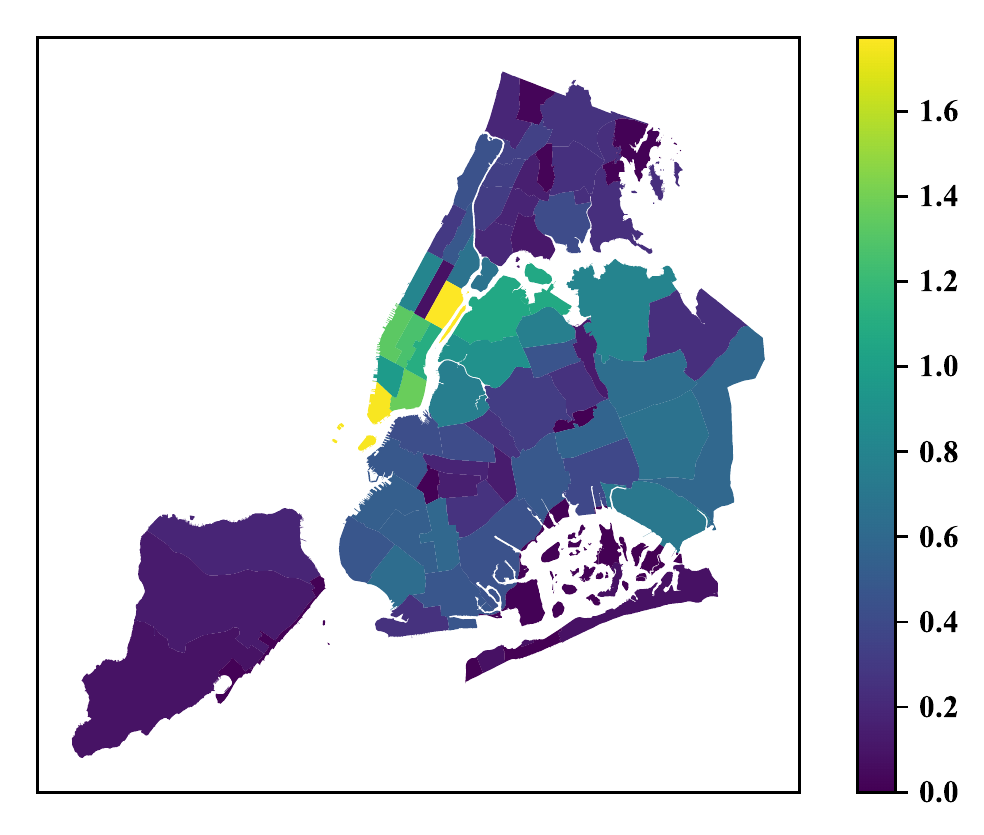}}}
	\caption{The hypothesis test p value (subplots a to c) and arrival rate (subplots d to f) of both passenger and new vehicles during peak hours}
	\label{input_arrivals}
\end{figure}

Except for external passenger and vehicle arrivals, we also examine the fixed probabilities of modal split $p_i^{ATS}$, pickup $p_i^{p,ATS} \text{or} p_i^{p,TTS}$, and system-exit $p_{i0}^{E,ATS}\text{or} p_{i0}^{O,TTS}$. Fig. \ref{fixed_split} (a) to (e) present mean value of corresponding probabilities across minutes. Fig. \ref{fixed_split} (f) to (j) summarize variance value of corresponding probabilities across minutes. The dark color indicates almost zero variance, which provides strong empirical evidence of fixed probabilities. We can take the mean value of the corresponding probabilities as mode inputs. Regarding the empirical observations on vehicle passenger searching time, $t_i^{ATS}$ can be directly counted from trajectory dataset. However, the measurement on $t_i^{TTS}$ is relatively more complicated, by introducing shortest path planning data from the Google API. We segment TTS rides with Google shortest path then enumerate defined passenger search time or travel time. Also, the vehicle routing matrix is derived based on similar measurement method as that for $t_i^{ATS}$ and $t_i^{TTS}$. The last input of interest is the number of servers at road queues, $c_i$. Here, the service rate can also be defined as the inverse of expected travel time under free flow in each road subsystem. We employ the free flow speed and together with measured trip distance by datasets and estimated trip distance from segmented shortest paths, obtain expected travel time under free flow. The number of servers are mainly generated based on the critical taxi accumulation in the MFD-like shapes, as shown in Fig.\ref{mfd}. Each borough tends to have consistent critical accumulations and has small differences from others.       

\begin{figure}[!htbp]
	\centerline{\subfigure[mean of $p_i^{ATS}$]{\includegraphics[width=0.33 \columnwidth]{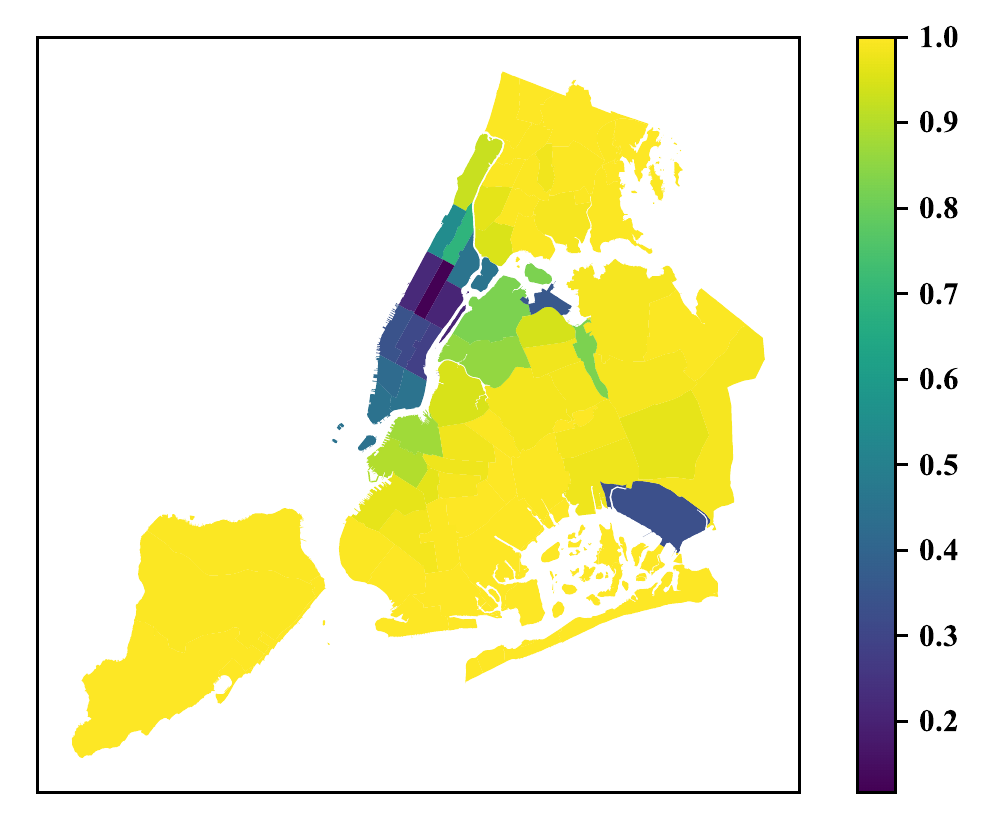}}	
		\subfigure[mean of $p_i^{p,ATS}$]{ \includegraphics[width=0.33 \columnwidth]{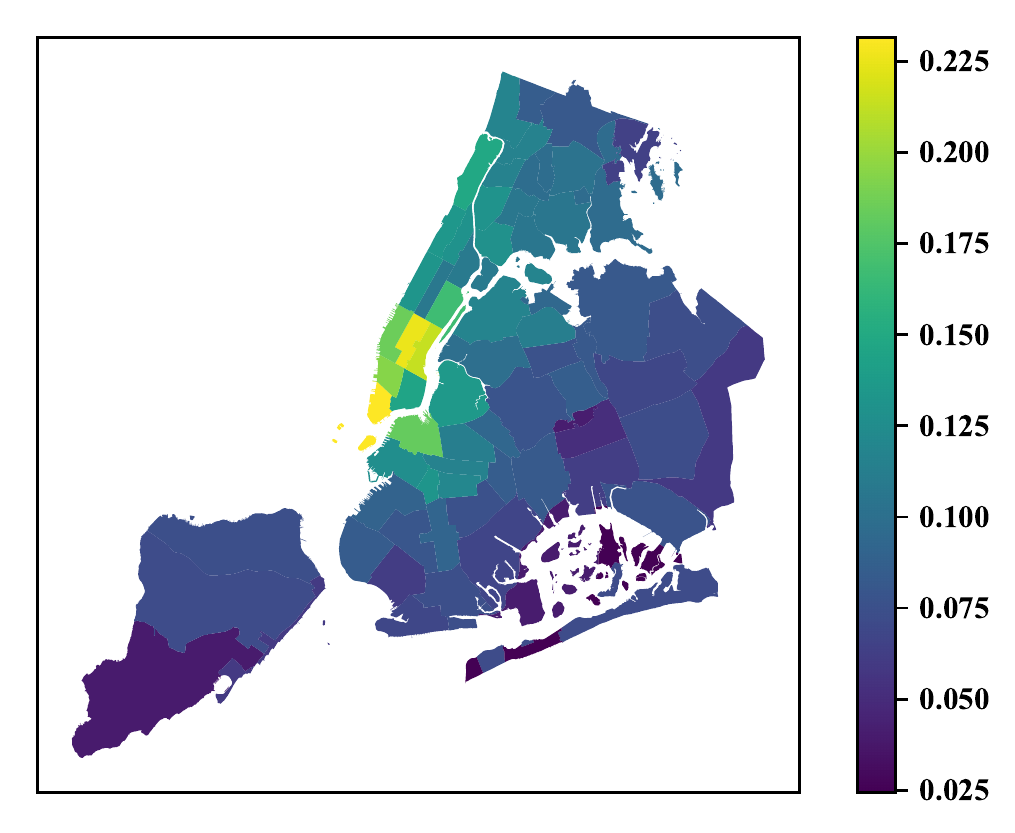}}
		\subfigure[mean of $p_{i0}^{E,ATS}$]{ \includegraphics[width=0.33 \columnwidth]{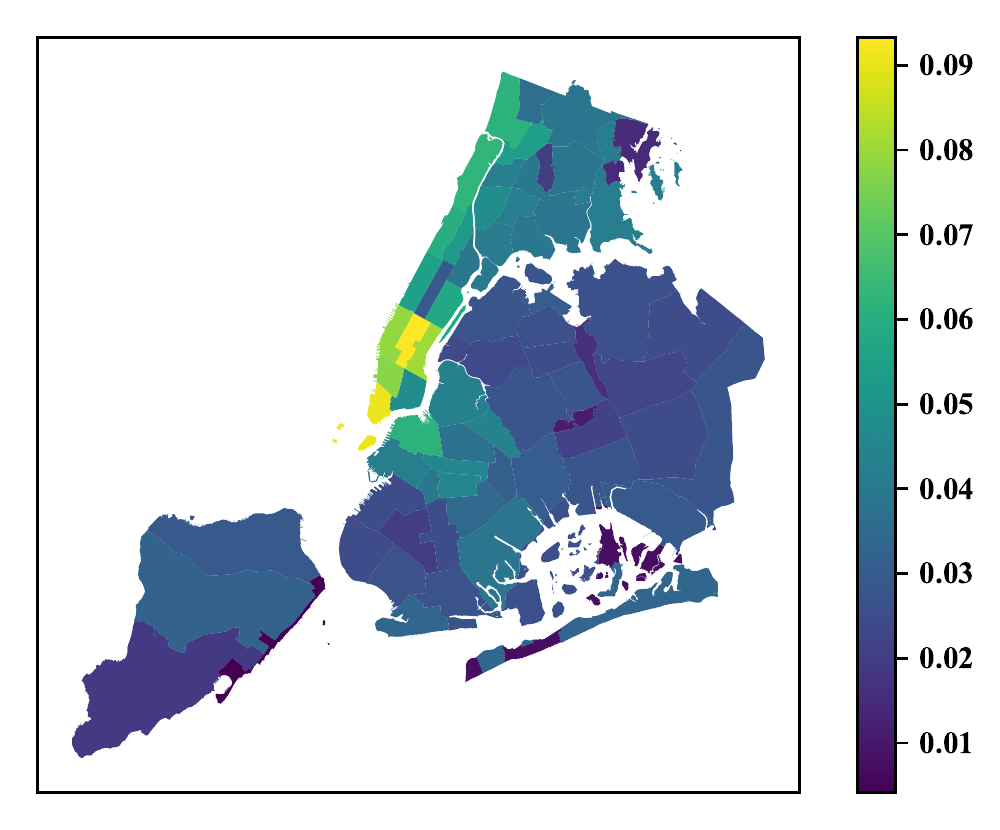}}	
}	\centerline{
		\subfigure[mean of $p_i^{p,TTS}$]{ \includegraphics[width=0.33 \columnwidth]{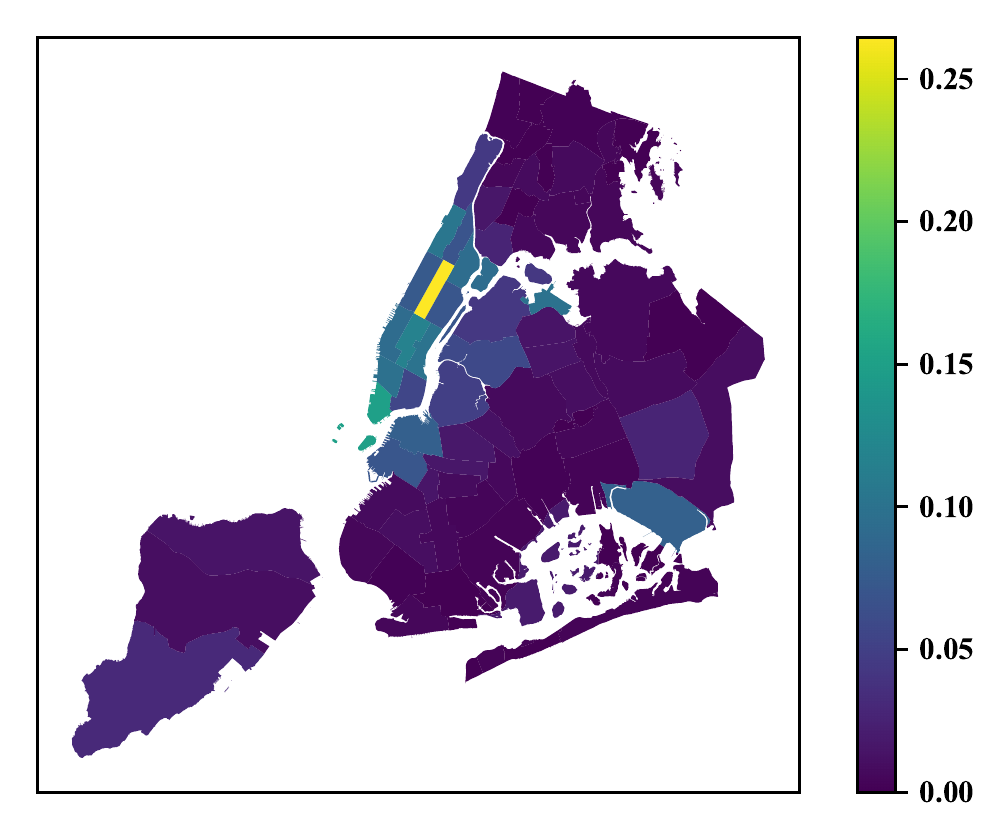}}
		\subfigure[mean of $p_{i0}^{O,TTS}$]{ \includegraphics[width=0.33 \columnwidth]{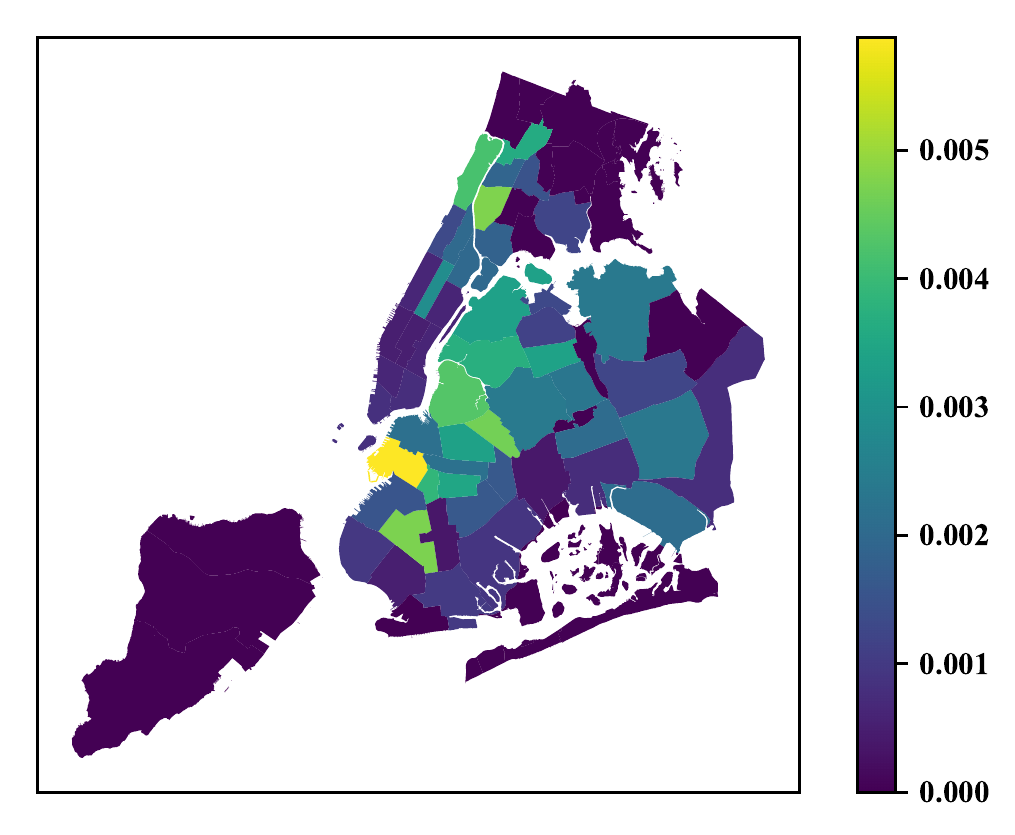}}		
	 \subfigure[variance of $p_i^{ATS}$]{ \includegraphics[width=0.33 \columnwidth]{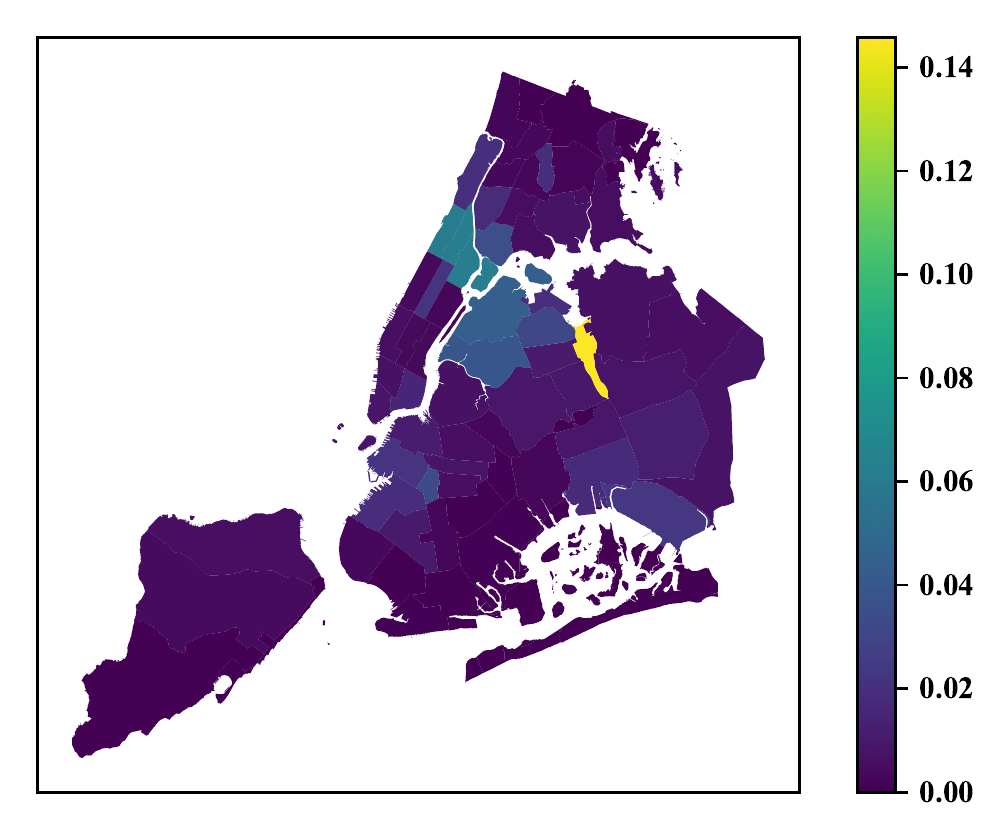}}
	}\centerline{
		\subfigure[variance of $p_i^{p,ATS}$]{ \includegraphics[width=0.33 \columnwidth]{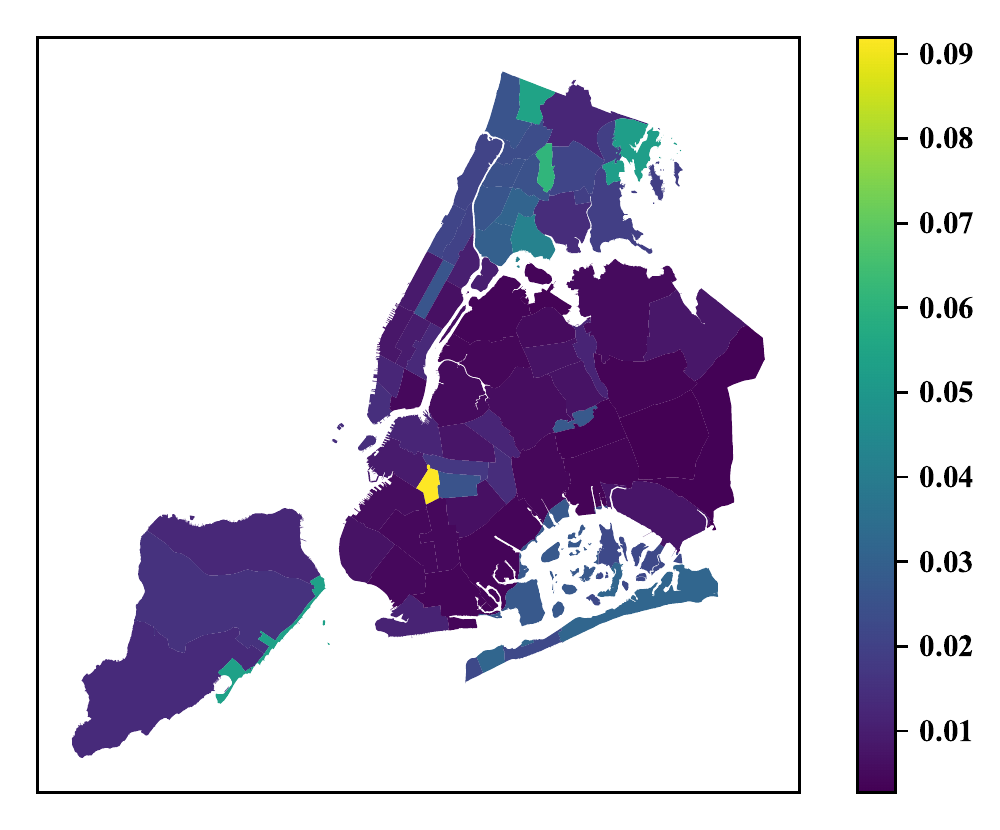}}
		\subfigure[variance of $p_{i0}^{E,ATS}$]{ \includegraphics[width=0.33 \columnwidth]{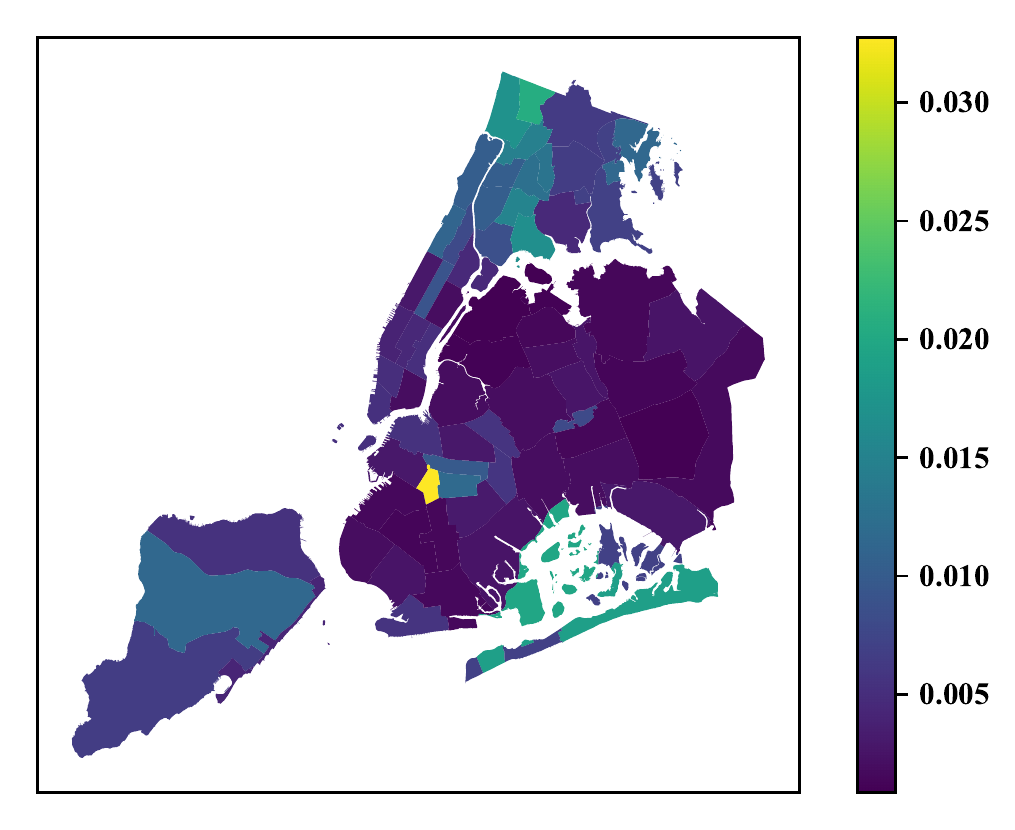}} 
		\subfigure[variance of $p_i^{p,TTS}$]{ \includegraphics[width=0.33 \columnwidth]{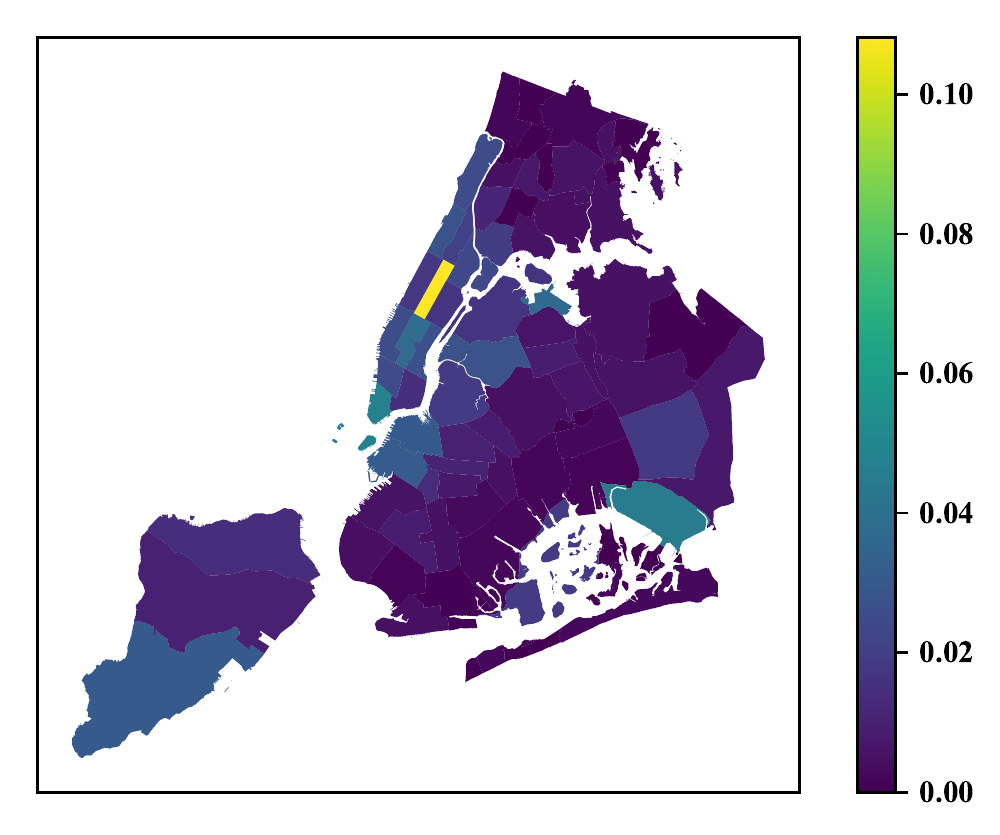}}}
	\centering
		\subfigure[variance of $p_{i0}^{O,TTS}$]{ \includegraphics[width=0.33 \columnwidth]{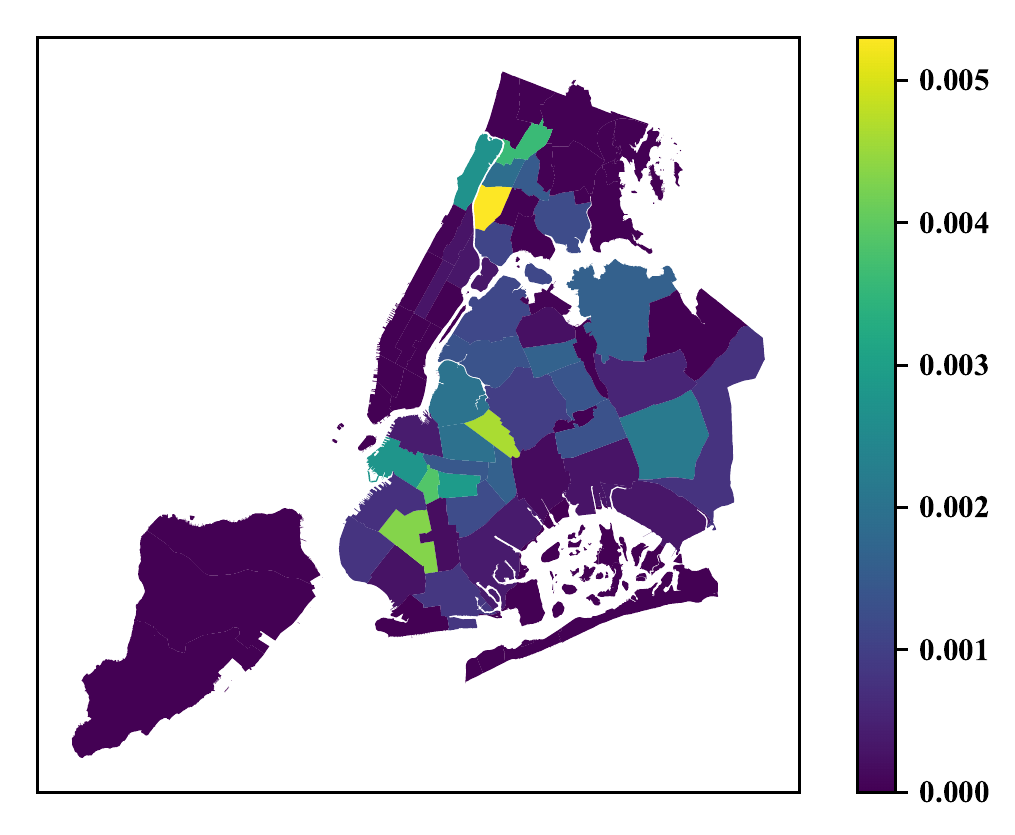}}

	\caption{The mean and variance of modal split, pickup, and system-exiting probabilities during peak hours}
	\label{fixed_split}
\end{figure}
\begin{figure}[!htbp]
	\centering
	\includegraphics[width=0.6 \columnwidth]{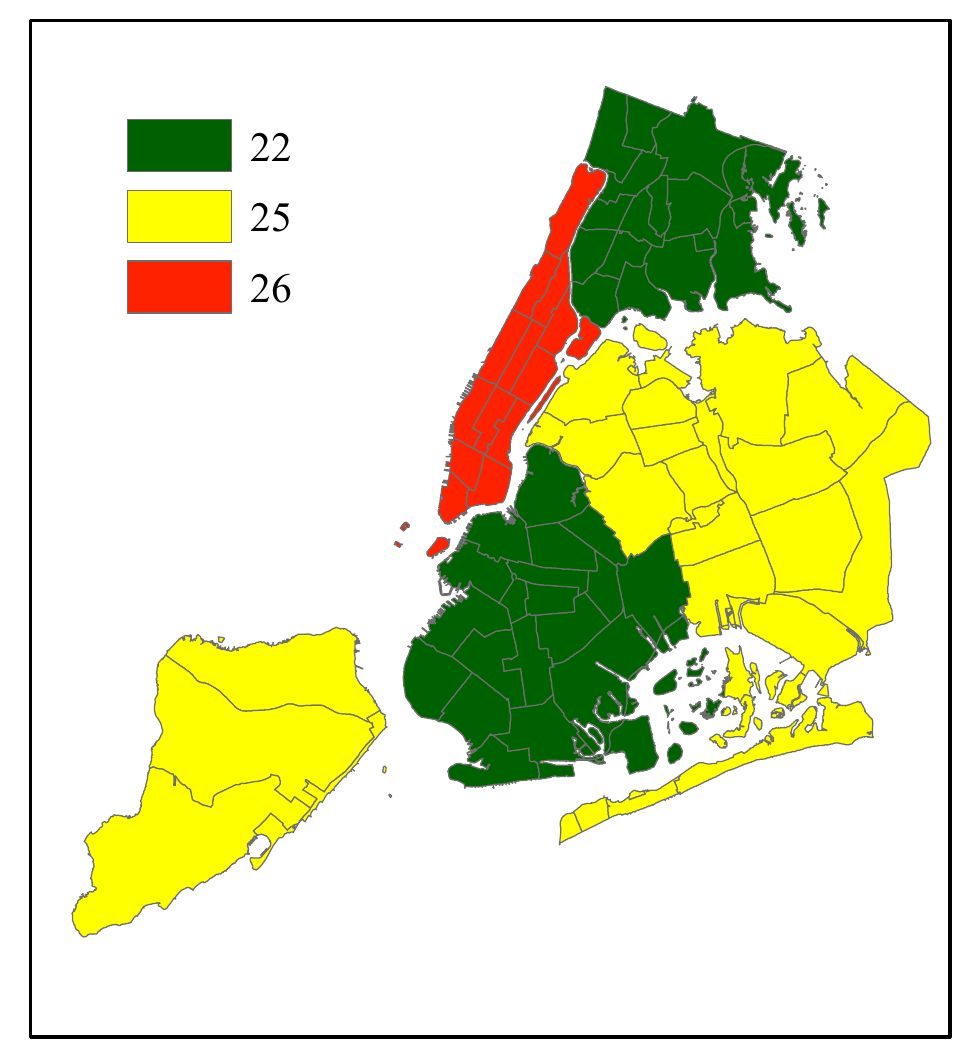}
	\caption{The number of servers of $M/M/c$ road queues}
	\label{mfd}
\end{figure}

\subsection{Model Evaluation}
In this section, we examine the model performance based on proposed settings and solutions of the linear programming problem. First, we compare the estimated $\lambda_i^{pv,*}$ that denotes the paired flow arrival rate in the synchronization process, to the observed passenger pickup rates in reality. Fig. \ref{paired_flow} shows almost same patterns between peak and off peak hours, but reveals many differences between ATS and TTS. The ATS system presents much lower absolute percentage errors (i.e. $<$5\%) in almost every spatial unit. In contrast, the proposed modeling structure has reliable outputs for ``hot'' areas of TTS system, which attract more than 90\% of TTS activities. Such significant differences may arise from spatial distribution of both services. Since the modeling structures involves vehicle movement over road network and routing probabilities, which directs majority of vehicles to ``hot'' areas and leads to unreliable estimations for remaining areas. In addition, the low percentage error also provides strong empirical evidence of $SM/M/1$ approximation with $M/M/1$.   

\begin{figure}[!htbp]
	\centerline{\subfigure[TTS - peak hours]{\includegraphics[width=0.485 \columnwidth]{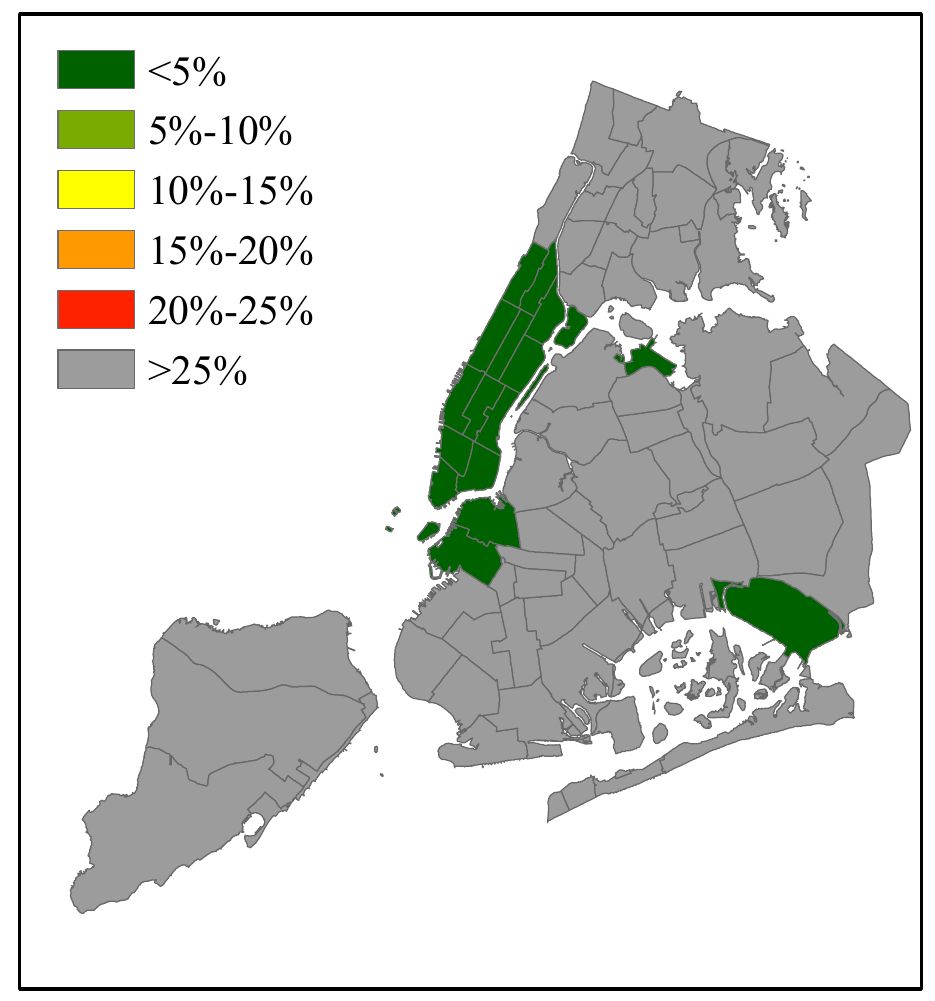}}
		\subfigure[ATS - peak hours]{	\includegraphics[width=0.48 \columnwidth]{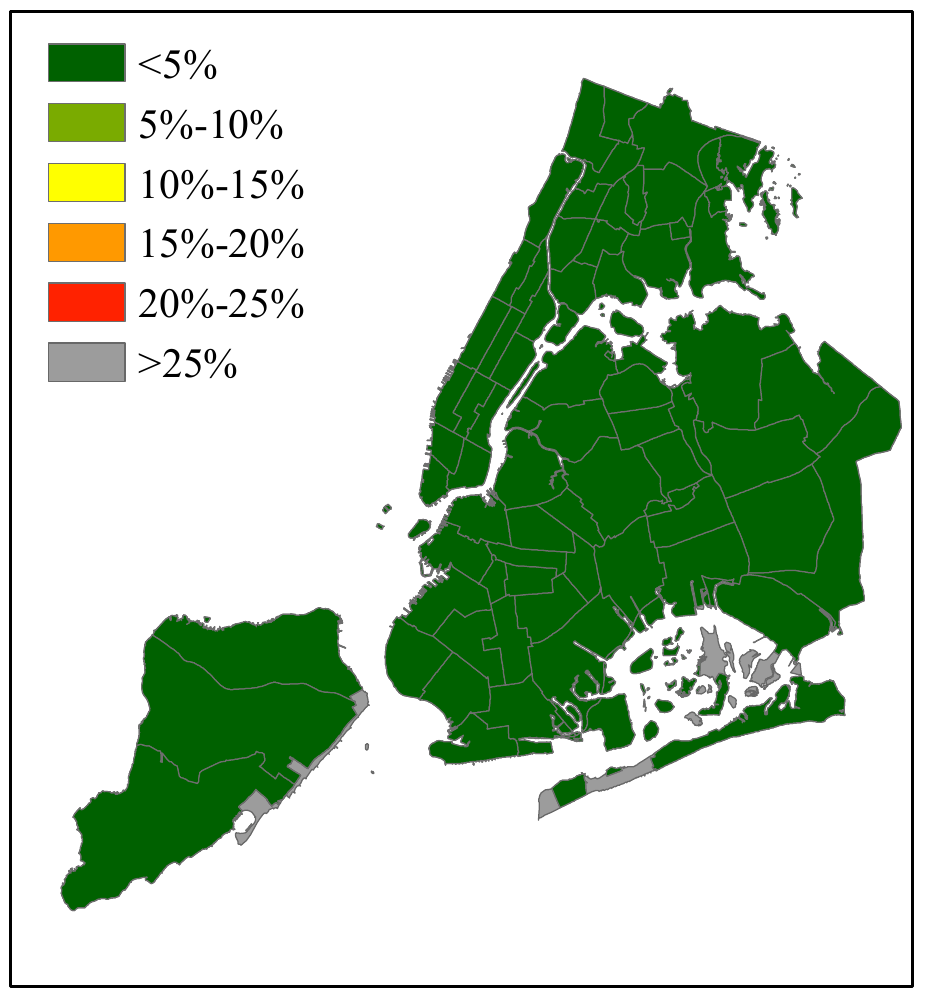} }
	} \centerline{
		\subfigure[TTS - off peak hours]{\includegraphics[width=0.48 \columnwidth]{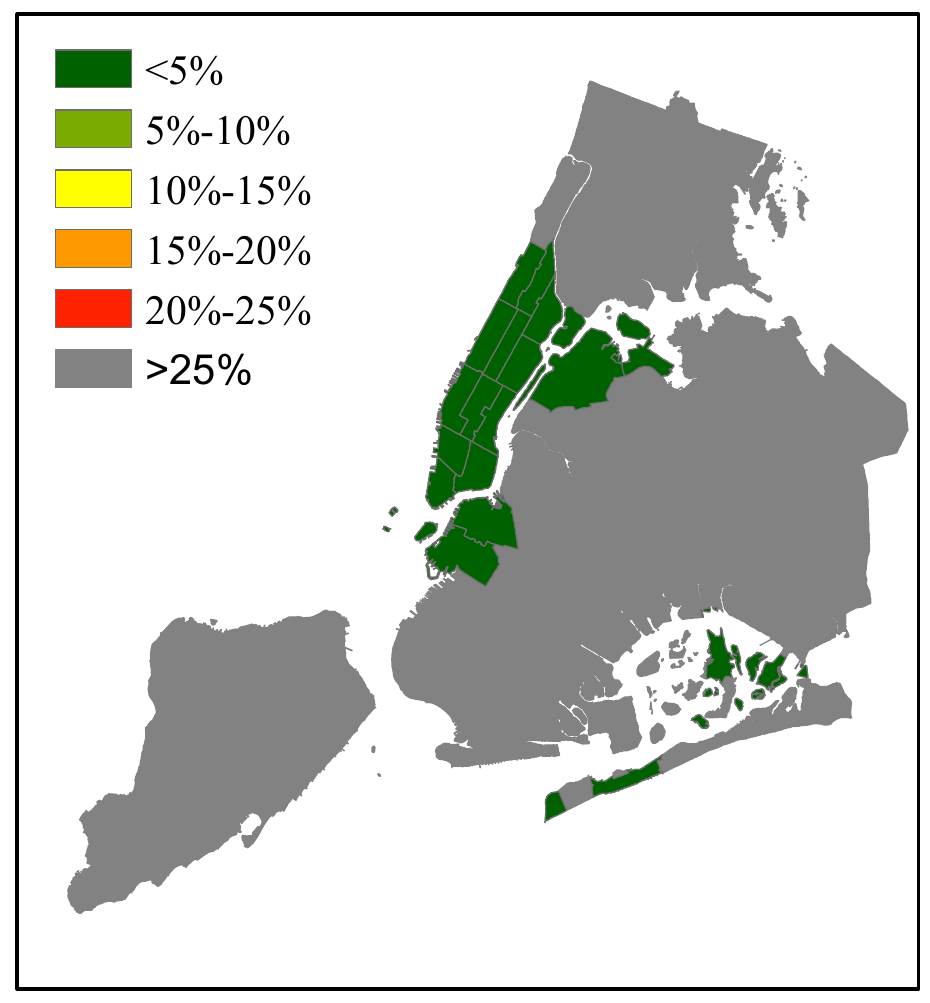}}
		\subfigure[ATS - off peak hours]{	\includegraphics[width=0.485 \columnwidth]{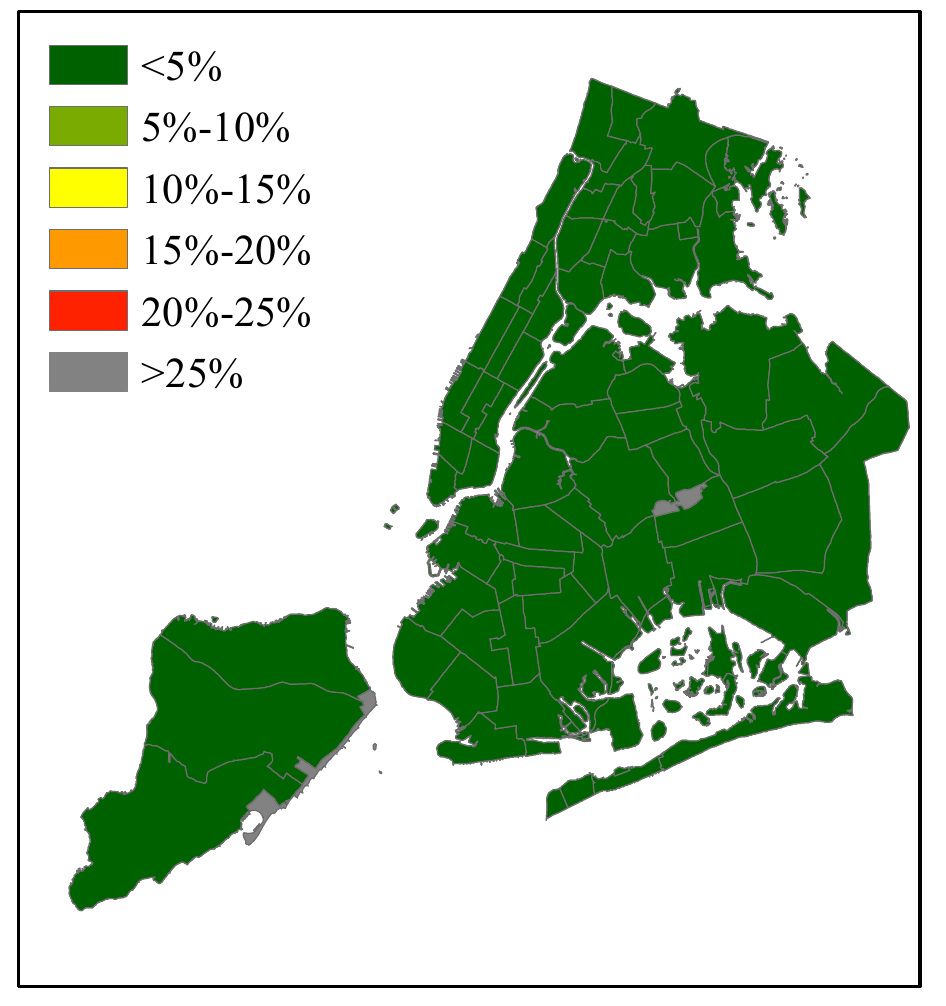} }}
	\caption{The absolute percentage errors between estimated $\lambda_{i}^{pv,*}$ and observed passenger pickup flows}
	\label{paired_flow}
\end{figure}

Moreover, we examine the taxi system performance, in terms of sojourn time (or total time from arrival to departure in one specific queue system), and evaluate the accuracy for both ATS and TTS, summarized in Fig.\ref{sojourn_perform}. Note that, multiple black spatial units in Fig. \ref{sojourn_perform} represent no accuracy measurement, since there are very limited empirical observations, resulting in unreliable measurements. Overall, the model fits the ATS system better and relatively worse for TTS system outside Manhattan, which is in line with findings in $\lambda_{i}^{pv,*}$. For most spatial units, the proposed modeling structures can yield accurate measurements on ATS vehicle-passenger matching performance with less than 10\% relative errors. However, similar to $\lambda_{i}^{pv,*}$ estimation, the outer Manhattan areas (gray areas in Fig.\ref{sojourn_perform}) still have worse outputs of TTS system performance. The reason is attributed to the limited number of TTS trips outside Manhattan and the unbalanced distribution of these trips.   

\begin{figure}[!htbp]
	\centerline{\subfigure[TTS - peak hours]{\includegraphics[width=0.49 \columnwidth]{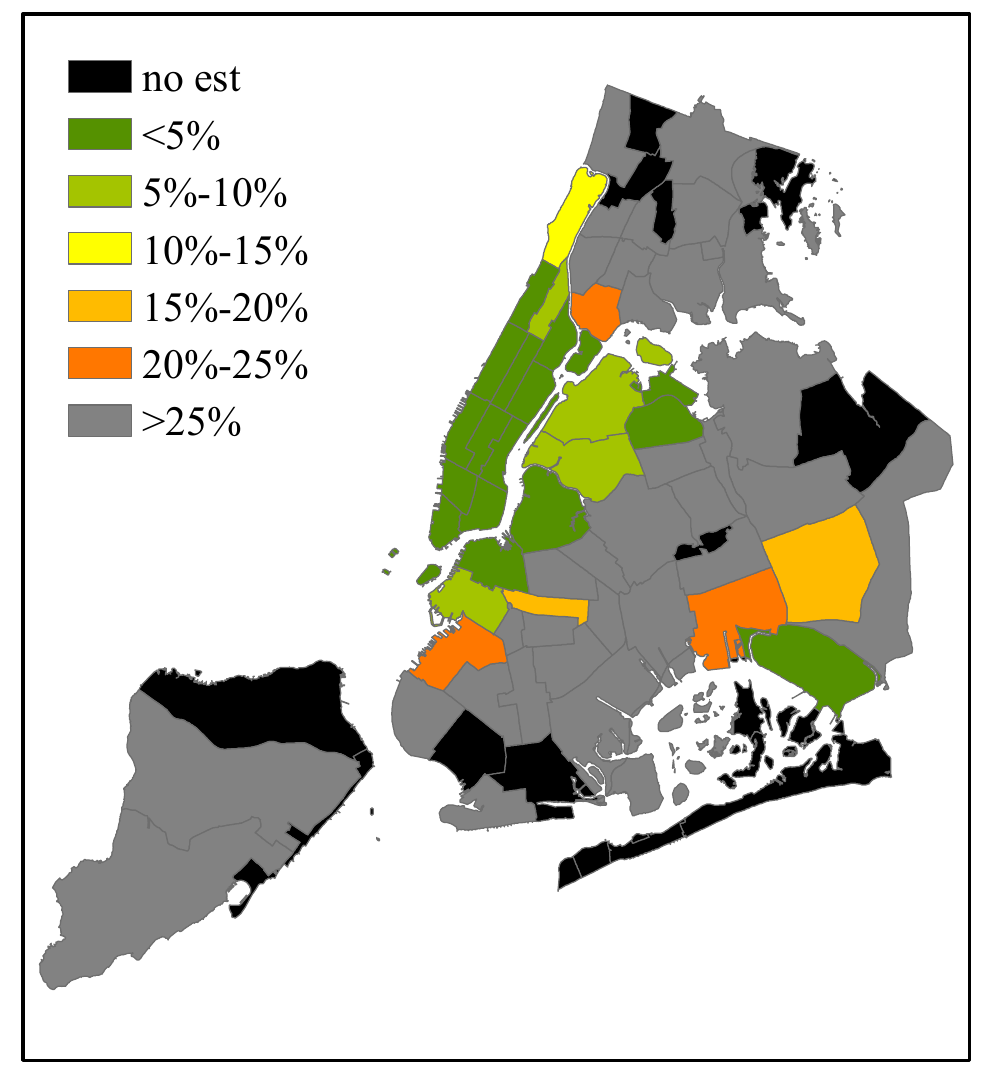}}
		\subfigure[ATS - peak hours]{	\includegraphics[width=0.503 \columnwidth]{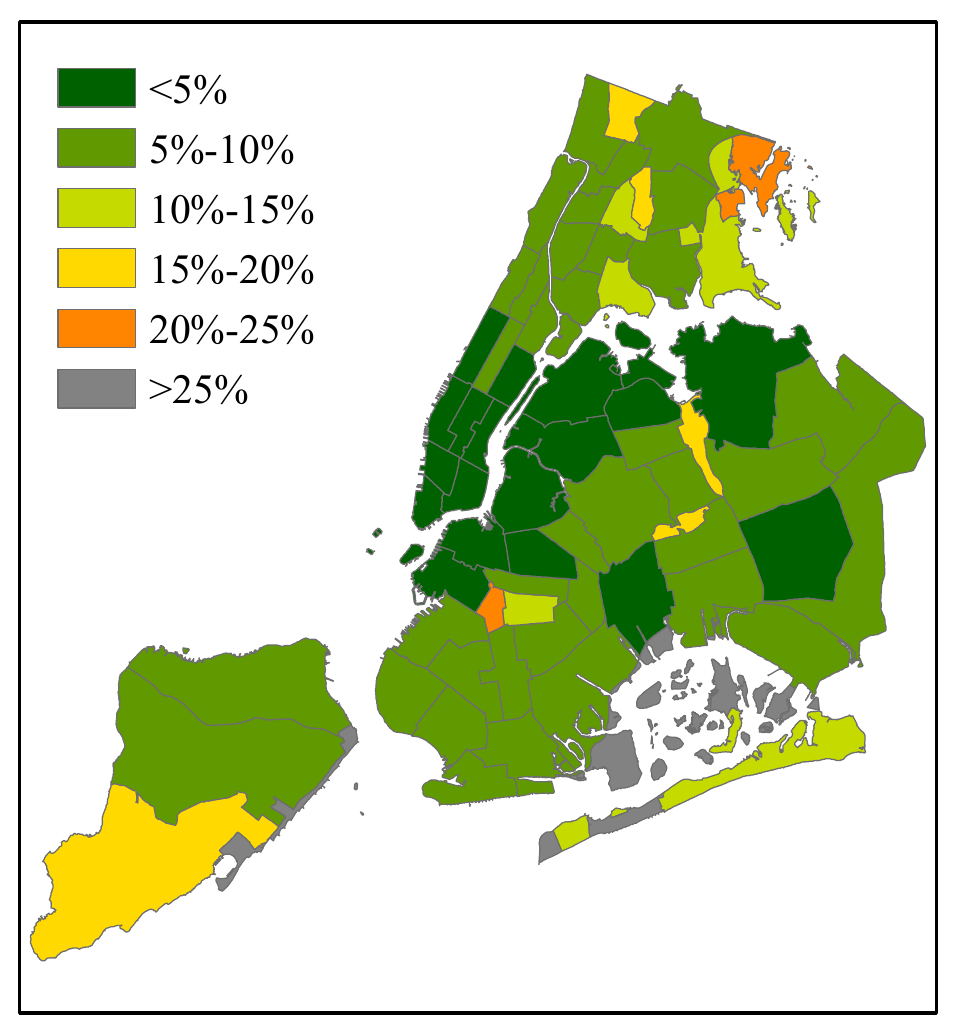}}}
	\centerline{
		\subfigure[TTS - off peak hours]{\includegraphics[width=0.49 \columnwidth]{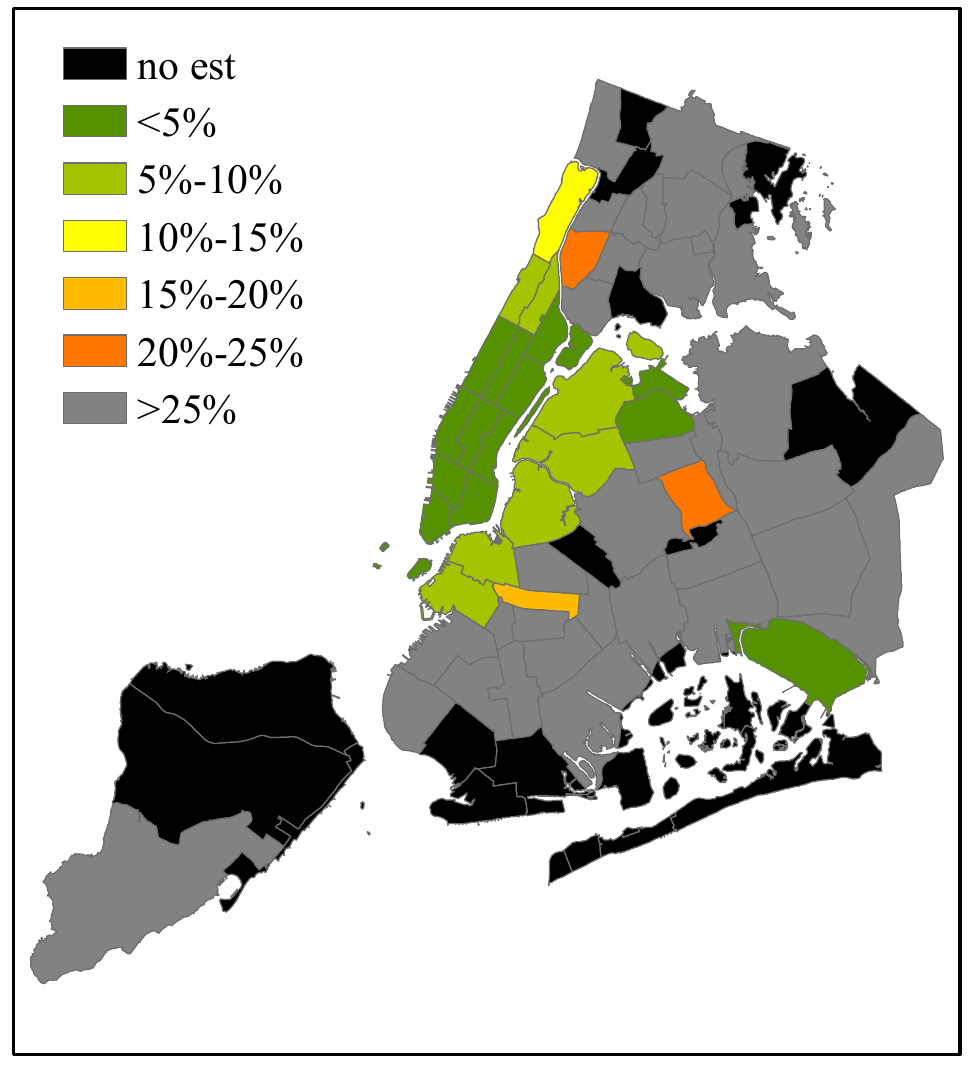}}
		\subfigure[ATS - off peak hours]{	\includegraphics[width=0.505 \columnwidth]{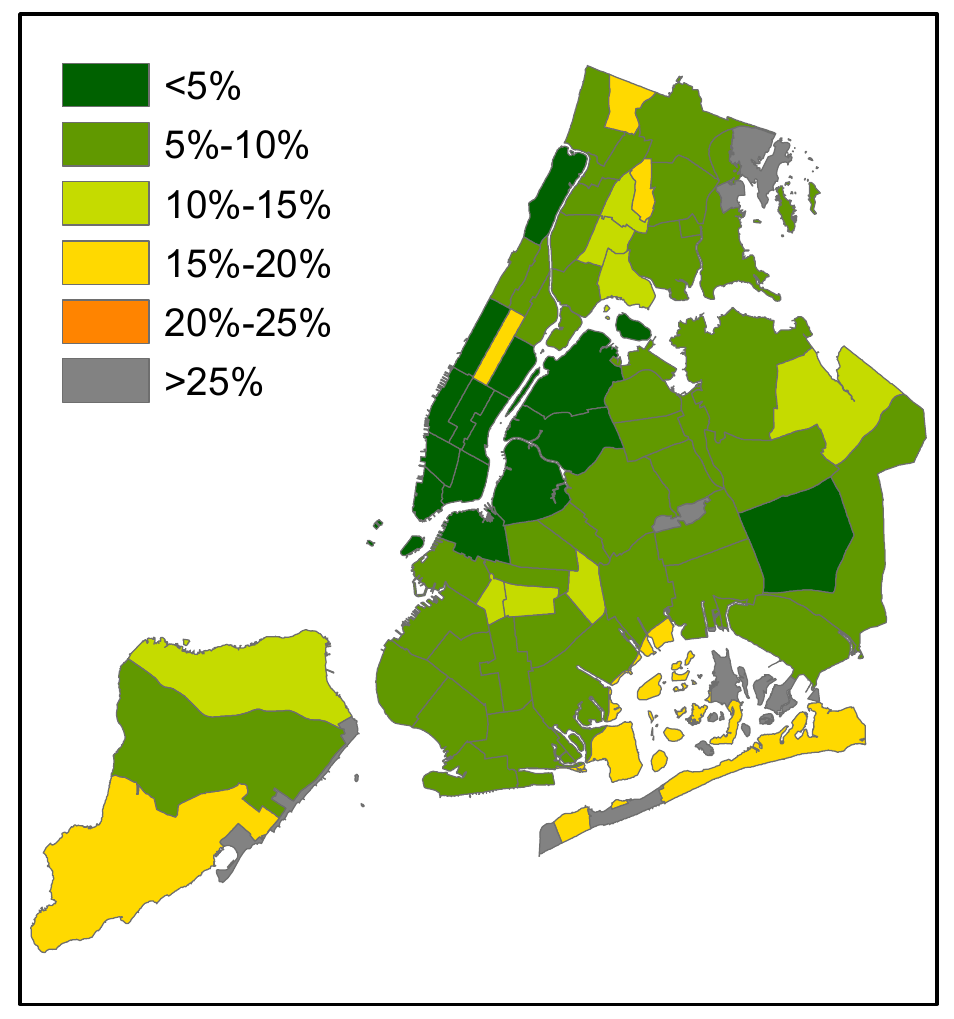} }}
	\caption{The mean absolute percentage errors between expected sojourn time and observed one at taxi queues}
	\label{sojourn_perform}
\end{figure}

Other performance metrics can be similarly examined at either taxi or road queues. We will not present all results due to the space limitations. Overall, the proposed queueing network model can describe vehicle-passenger matching process well and yields reliable estimations on both the expected and distribution of matching performance (e.g. sojourn time). At appropriate spatial and temporal aggregation (such as homogeneous spatial units and periods in this study), vehicles and passengers behave identically without much fewer variants, which is also our assumptions for the queueing network. However, the queueing network performance relies on the parameter calibrations that requires high-resolution dataset of individual movements. Our datasets are now are not of sufficient resolution to allow the calibration of all parameters (for example, TTS in outer Manhattan areas).

\section{Conclusions and Future Works}
The study develops a queueing network approach to describe the complex interactions between the ATS and TTS systems within an unified taxi market, as well as between the taxi- and urban road- system. We first introduce the queueing network structures in which not only the queue node itself can capture the dynamics of taxi passenger and vehicle behaviors but also the node connections can allow the flow exchanges accounting for network externalities. Specifically, we propose (1) the synchronized process $SM/M/1$ for both ATS and TTS passenger-vehicle matching behaviors; (2) the multi-server $M/M/c$ queue for the urban road system; and (3) the state-dependent service rate of $SM/M/1$. Moreover, we provide an approximation of the proposed non-stationary queueing network with a Jackson network and investigate the stationary state distributions. Finally, we fully utilize our rich dataset of TTS and ATTS movements from NYC to test the modeling approach. 

Overall, the application of the proposed modeling structure is far beyond what we have presented in this paper. One main characteristic of ATS is the dynamic pricing and thus the drivers' and passengers' incentives. The differentiated controls over both space and time make the problem more interesting but also challenging which are not fully addressed in this study. The model in this paper emphasizes the macroscopic interactions between urban road and taxi systems, but does not capture the behavioral dynamics of the individuals and how they respond to the taxi market (aka feedback in the system). In future studies, the proposed queueing network will be generalized to include the full dynamics of taxi markets and individual behaviors, thus allowing us to have in-depth insights into system control. Further, sensitivity analysis of the model should also be conducted to understand the stability of the outputs to minor changes in taxi supply and demand.                  

\section*{Acknowledgements}
We would like to thank Mr. Zengxiang Lei for his discussions and helps during paper revision.

\bibliography{refs}
\bibliographystyle{IEEEtran}

\newpage
\appendix
\renewcommand\thefigure{\thesection.\arabic{figure}}  
\renewcommand{\theequation}{\thesection.\arabic{equation}}
\setcounter{equation}{0}    
\setcounter{figure}{0}    
\section*{Appendices}
\section{Hypothesis Testing of Poisson Assumption}
The study proposes $M/M/1$ queue for vehicle-passenger matching of each taxi service type, which requires strong assumption of Poisson arrivals. However, we have not seen any empirical studies for the assumption of taxi flows, regardless of passenger and external vehicle arrivals. With massive taxi trip records, this study employs multiple hypothesis testing methods and performs extensive statistical tests. From the test results, we can expect following experiment setups in one large-scale taxi system, given strong empirical evidence: 
\begin{itemize}
	\item[-] Homogeneous Spatial Units are the spatial divisions of one large city or taxi system. In one specific homogeneous spatial unit, both traditional street-hailing (TTS) and emerging app-based (ATS) taxi services have Poisson passenger and vehicle arrivals;
	\item[-] Arrival Count Intervals are the time interval to count how many arrivals generate. The set of arrival counts should be from Poisson distribution;
	\item[-] Hours and Days are of interest. The taxi activities are significantly influenced by time-of-the-day and day-of-the-week. Certain hours and days will be more likely to have Poisson arrivals, other than all days and hours.  
\end{itemize}

The null hypothesis is that the observed taxi arrival events under given spatial division and time interval follow Poisson distribution. The alternative hypothesis is that the observed taxi arrival events under given spatial division and time interval do not follow Poisson distribution. The first hypothesis testing method is adapted from Kolmogorov-Smirnov test, considering the application for discrete events\cite{correct_ks}. The correction primarily consider sample size, as equation\ref{eq:ks_correct}. Once $D_n$ is greater than one critical value from Kolmogorov-Smirnov distribution, it rejects the null hypothesis and the observed arrival events can not be assumes as Poisson distribution at confidence level of 95\%.    
\begin{equation}\label{eq:ks_correct}
D_n=max_{x\in J}\sqrt{n}|H(x)-F_{n}(x)|-\frac{1}{\sqrt{n}}
\end{equation}   
where, $n$ is sample size; $J$ is levels of arrival count ranging from 0 to maximum count; $H(x)$ is hypothesized cumulative density function of one Poisson distribution with empirical mean estimated from random half of samples; and $F_{n}(x)$ is empirical cumulative density function for remaining half of samples.

Additionally, we introduce three statistical hypothesis testing method for homogeneous Poisson \cite{poisson_test}. The null hypothesis is that taxi arrival events are from Poisson distribution with same rate across time in one given spatial division. The alternative hypothesis is that taxi arrival events are from Poisson distribution but with different rates across time in one given spatial division. Such statistical testing are applied directly on a sequence of arrival counts $\{c_1,c_2,\cdots,c_i,\cdots,c_n\}$ and are generally based on $\chi^2$ distribution. The differences in three methods are in statistics computation. The Anscombe method first transforms original arrival counts then compute statistics with squared errors, as shown in equation \ref{eq:anscombe}. The likelihood ratio statistics are computed from ratio to mean value, as shown in equation \ref{eq:likelihood}. And the conditional $\chi^2$ statistics are measured with ratio of squared error to mean value, as shown in equation \ref{eq:conditional}. If computed statistics are greater than critical $\chi^2_{n-1}$, it rejects null hypothesis at confidence level of 95\% and the taxi arrival events are from Poisson distribution with time-dependent rates.     
\begin{equation}\label{eq:anscombe}
\begin{array}{l}
y_i=\sqrt{c_i+\frac{3}{8}}  \\
T_{anscombe}=4\sum\left(y_i-\bar{y}\right)^2
\end{array}	
\end{equation}

\begin{equation}\label{eq:likelihood}
T_{likelihood}=2\sum_{i=1}^n c_i ln\left(\frac{c_i}{\bar{c}}\right)
\end{equation}
\begin{equation}\label{eq:conditional}
T_{anscombe}=\sum_{i=1}^n \frac{\left(c_i-\bar{c}\right)^2}{\bar{c}}
\end{equation}
where, $n$ is sample size; $c_i$ is arrival count in a specific time interval $i$ and spatial division; and $\bar{c}$ is mean values of all arrival counts.

The potential aggregation scales, as well as study periods of interest, are summarized as follows:      
\begin{itemize}
	\item[-] The potential spatial scales are mainly based on four administrative divisions in NYC, including Borough (Figure \ref{spatial_division} (a), 5 in total, $\sim$60.4 mi$^2$ on average per Borough), Community Districts (Figure \ref{spatial_division} (b), 71 in total, $\sim$4.3 mi$^2$ on average per community district), Zip Code Tabulation Area [ZCTA] (Figure \ref{spatial_division} (c), 214 in total, $\sim$1.4 mi$^2$ on average per zip code tabulation area), and Census Tracts (Figure \ref{spatial_division}\begin{figure}[!h]
		\centerline{
			\subfigure[Borough]{\includegraphics[width=0.5 \columnwidth]{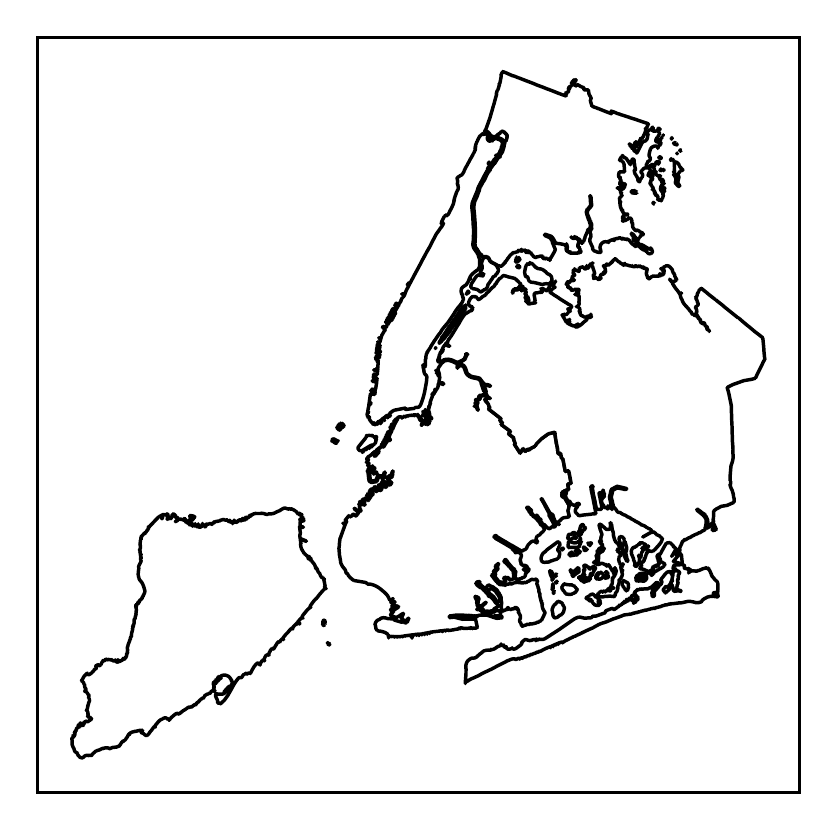}}
			\subfigure[Community Districts]{\includegraphics[width=0.5 \columnwidth]{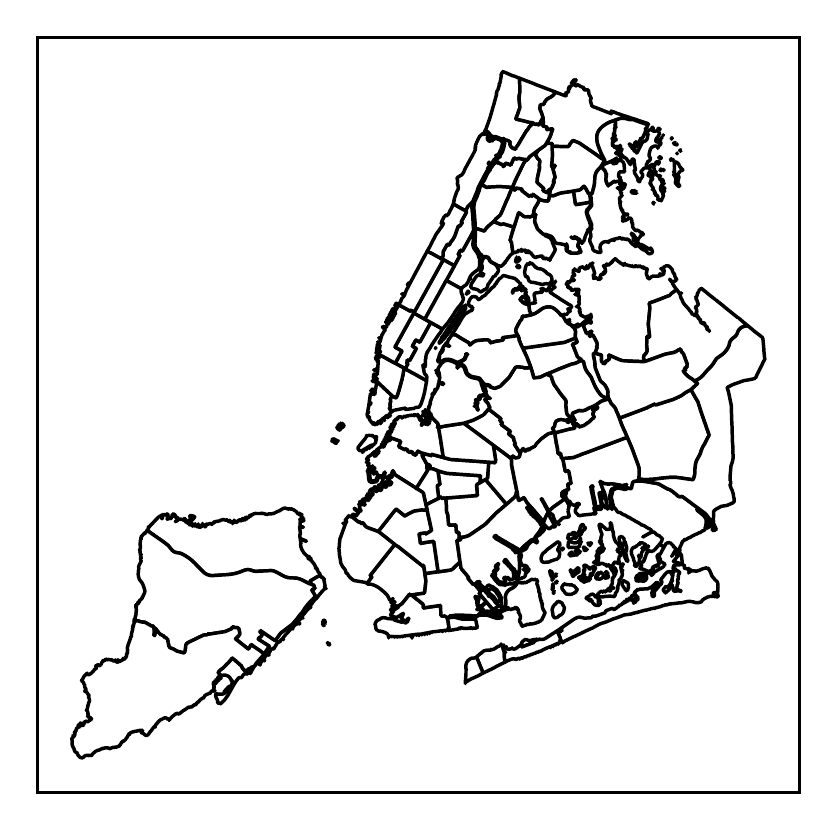}}
		}
		\centerline{
			\subfigure[Zip Code Tabulation Area]{\includegraphics[width=0.5 \columnwidth]{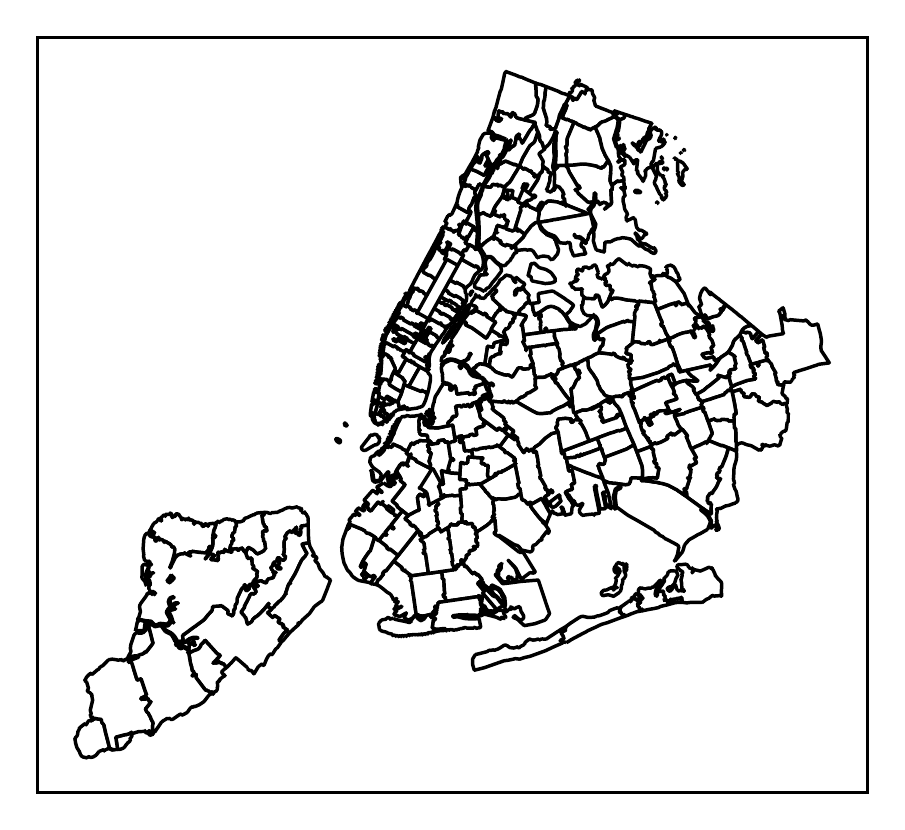}}
			\subfigure[Census Tracts]{\includegraphics[width=0.465 \columnwidth]{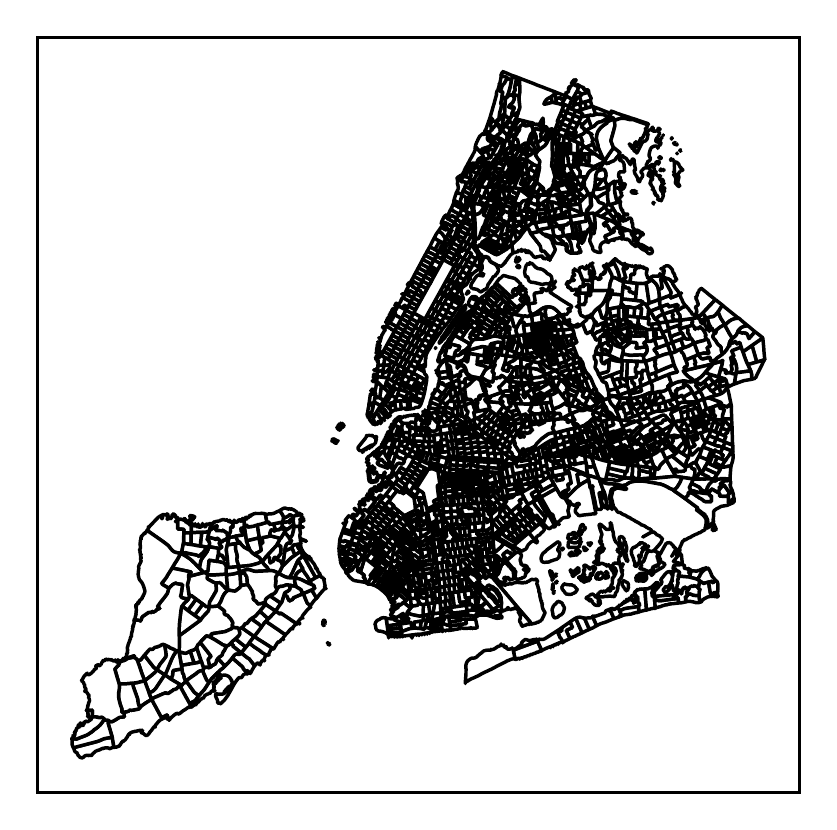}}
		}
		\caption{Potential spatial division in New York City}
		\label{spatial_division}
	\end{figure} (d), 2165 in total, $\sim$0.14 mi$^2$ on average per census tract). Note that we use administrative divisions rather than grid based spatial scale, since most socioeconomic variables are only available at administrative divisions and modeling at administrative divisions will be much easier to measure socioeconomic impacts on taxi activities in future studies.
	\item[-] We test seven different count intervals, that are 1-min, 5-min, 10-min, 15-min, 20-min, 30-min, and 60-min. In other words, we count TTS, ATS, or both TTS and TTS pickups (or vehicle arrivals) every count interval then test whether the flow can be described with Poisson distribution.
	\item[-] In addition, we also test homogeneous period selection, considering time-of-day and day-of-week effects. Regarding the peak (or off peak) hour, we include three difference cases, including 1-hour period (peak: 6pm to 7pm, or off peak: 10am to 11am), 2-hour period (peak: 5pm to 7pm, or off peak: 9am to 11am), and 3-hour period (peak: 5pm to 8pm, or off peak: 9am to 12pm). Moreover, we classify the weekdays from Mondays to Thursdays, compared to the all seven days case.     
\end{itemize}

\section{Spatiotemporal Aggregation}

Figure \ref{arrival-bo_peak} to \ref{arrival-census} show the percentage of zones not rejecting Poisson distribution with 4 methods, 7 count intervals, and day of the week. It is apparent that the smaller count interval generally leads to more zones not rejecting Poisson assumptions, across almost all plots. Although several hypothesis tests by corrected KS reveals similar percentages from 1-min to 60-min count interval, the homogeneous Poisson tests generally reject the null hypothesis of arrival counts are from one single homogeneous Poisson distribution if we have a larger count interval. Thus, we select 1-minute as the count interval in this study. 

Figure \ref{arrival-bo_peak}\begin{figure}[h]
	\centering
	\includegraphics[width=1.2 \columnwidth]{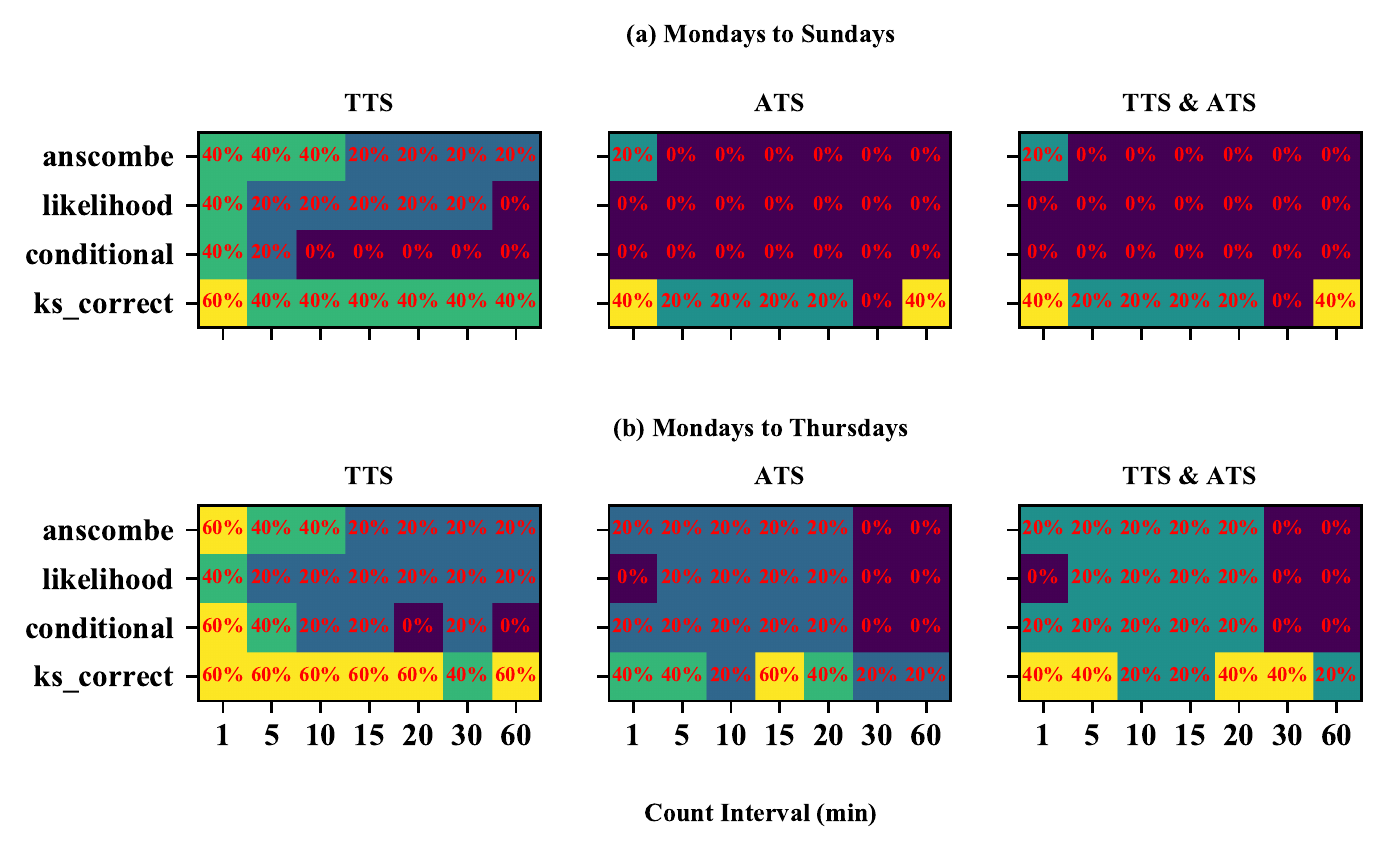}
	\caption{Hypothesis test results for passenger pickups at Boroughs in one-hour peak}
	\label{arrival-bo_peak}
\end{figure} \begin{figure}[h]
	\centering
	\includegraphics[width=1.2 \columnwidth]{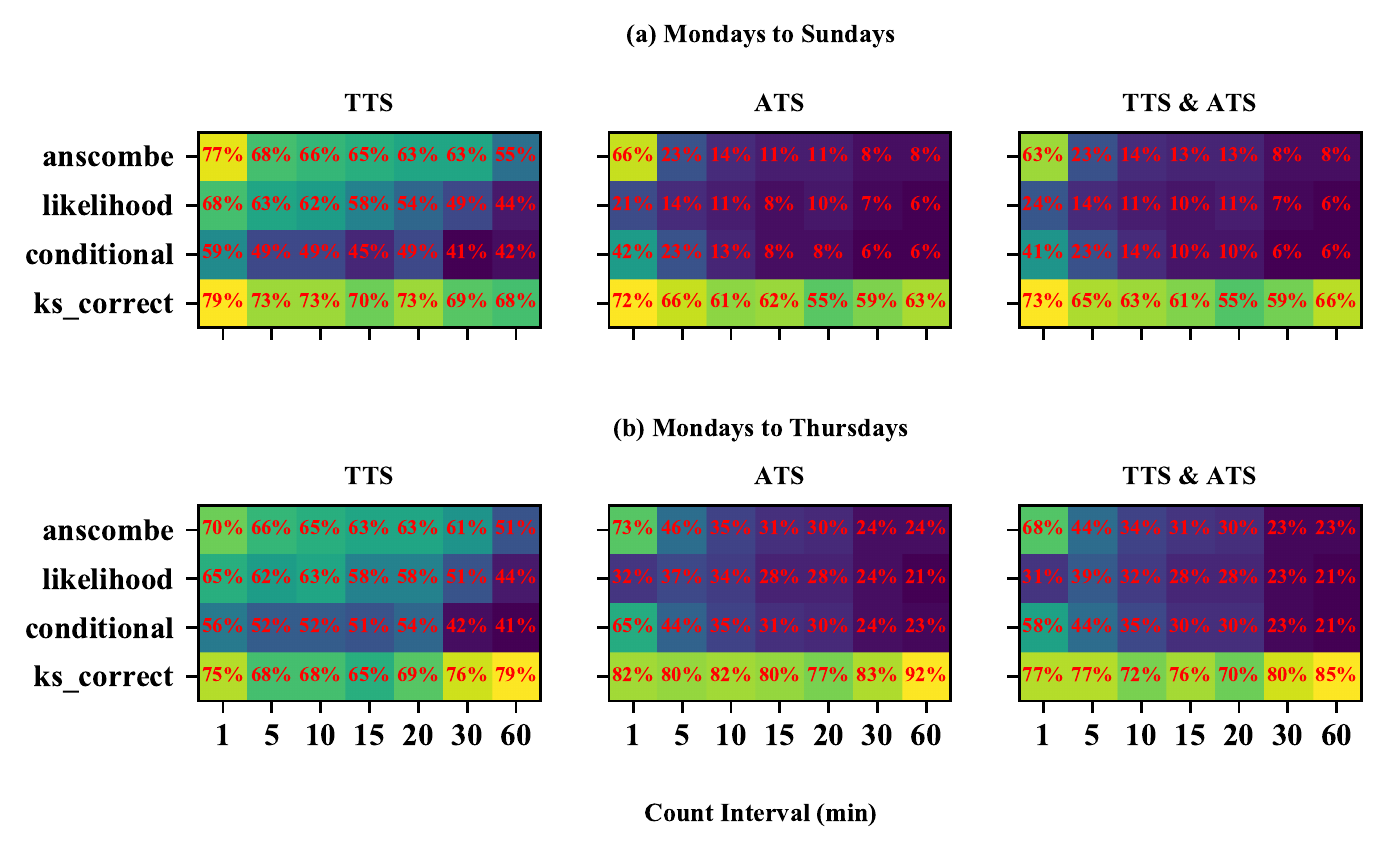}
	\caption{Hypothesis test results for passenger pickups at Community Districts in one-hour peak}
	\label{arrival-cd_peak} \end{figure} to Figure \ref{arrival-census_peak}\begin{figure}[h]
	\centering
	\includegraphics[width=1.2 \columnwidth]{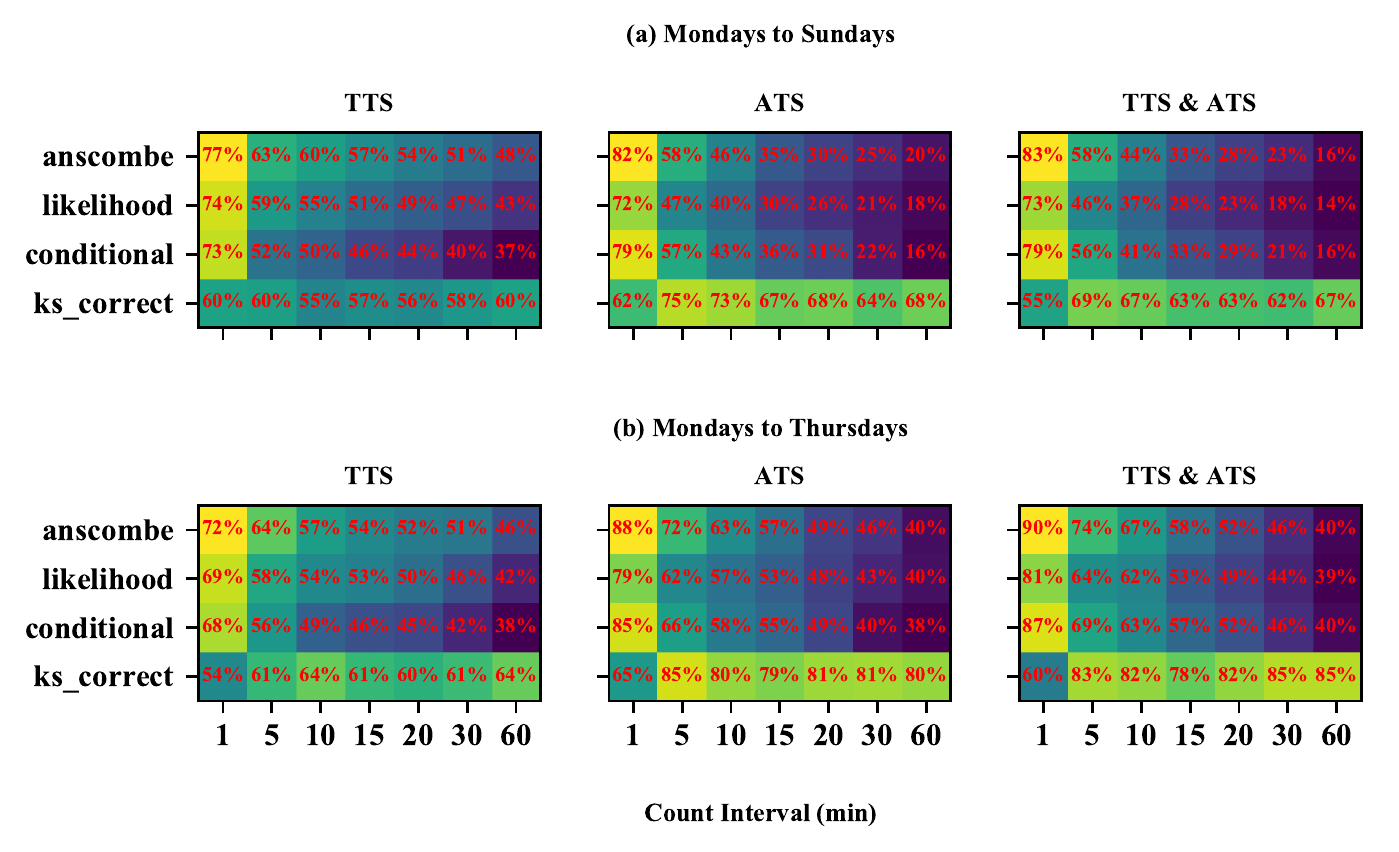}
	\caption{Hypothesis test results for passenger pickups at ZCTA in one-hour peak}
	\label{arrival-zcta_peak}
\end{figure} \begin{figure}[h]
	\centering
	\includegraphics[width=1.2 \columnwidth]{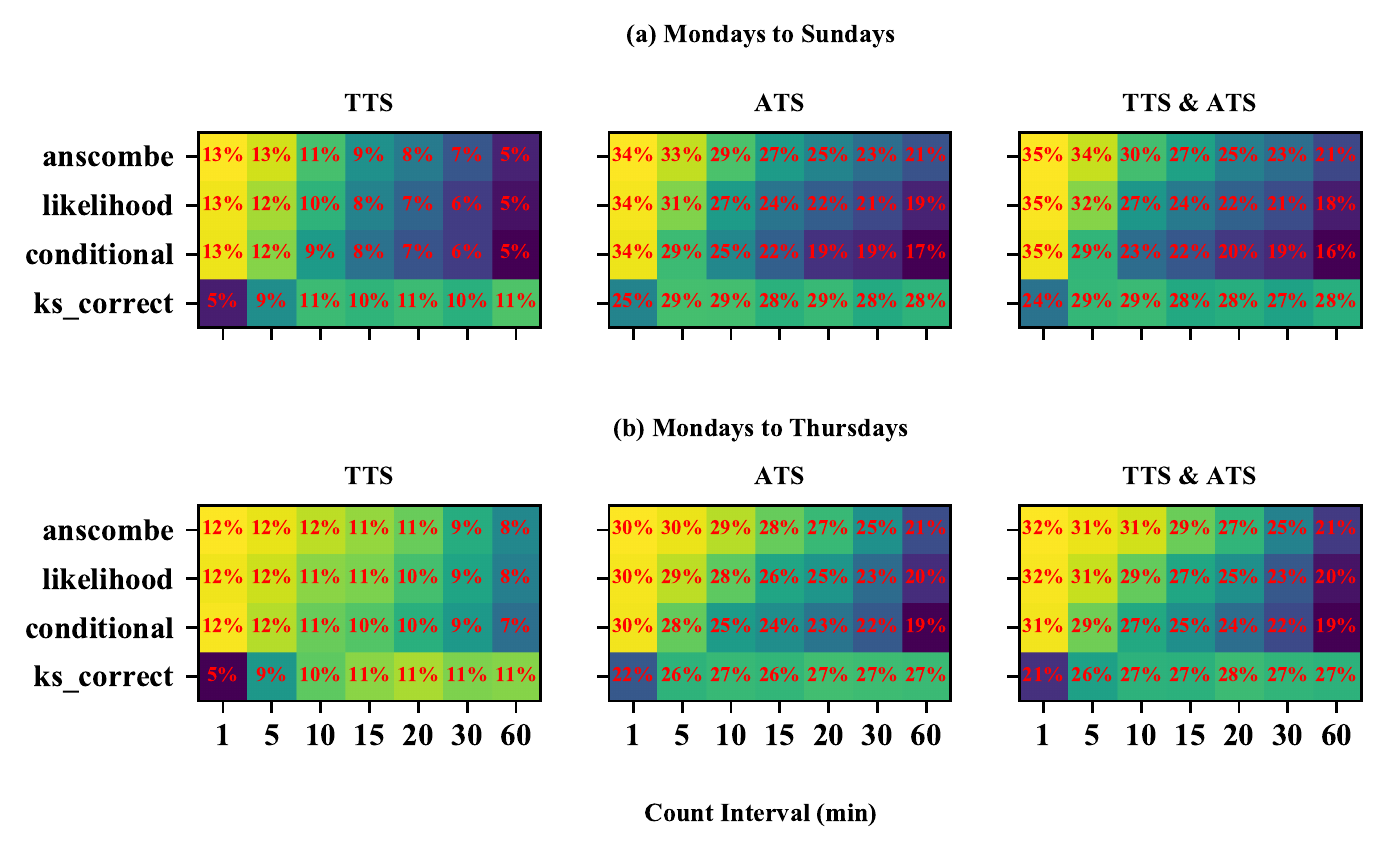}
	\caption{Hypothesis test results for passenger pickups at Census Tracts in one-hour peak}
	\label{arrival-census_peak}
\end{figure} summarize the percentages of zones where passenger pickups can be assumed as Poisson distribution at four levels of spatial scale in peak hour, respectively. Figure \ref{arrival-bo}\begin{figure}[h]
	\centering
	\includegraphics[width=1.2 \columnwidth]{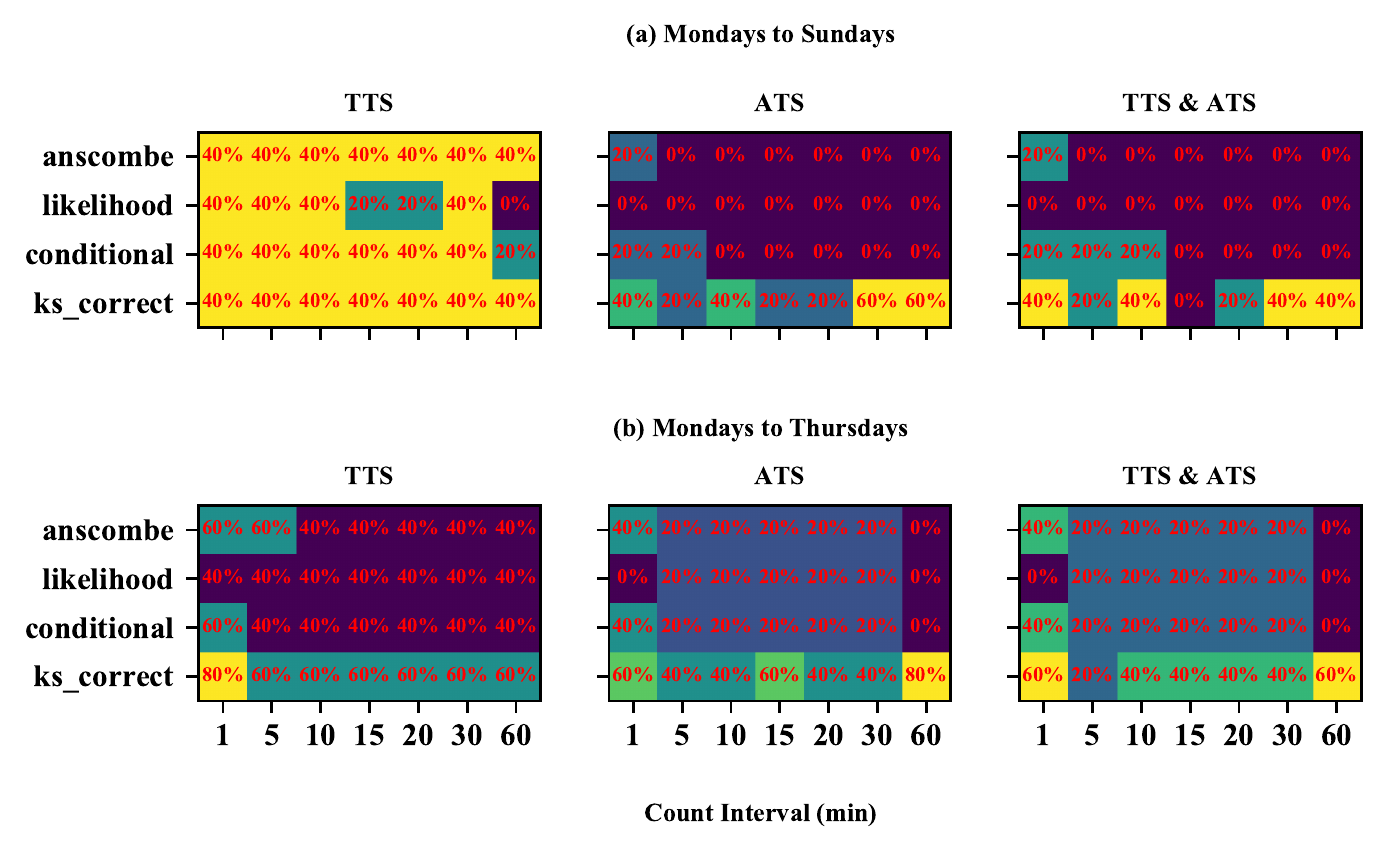}
	\caption{Hypothesis test results for passenger pickups at Boroughs in one-hour off peak}
	\label{arrival-bo}
\end{figure} \begin{figure}[h]
	\centering
	\includegraphics[width=1.2 \columnwidth]{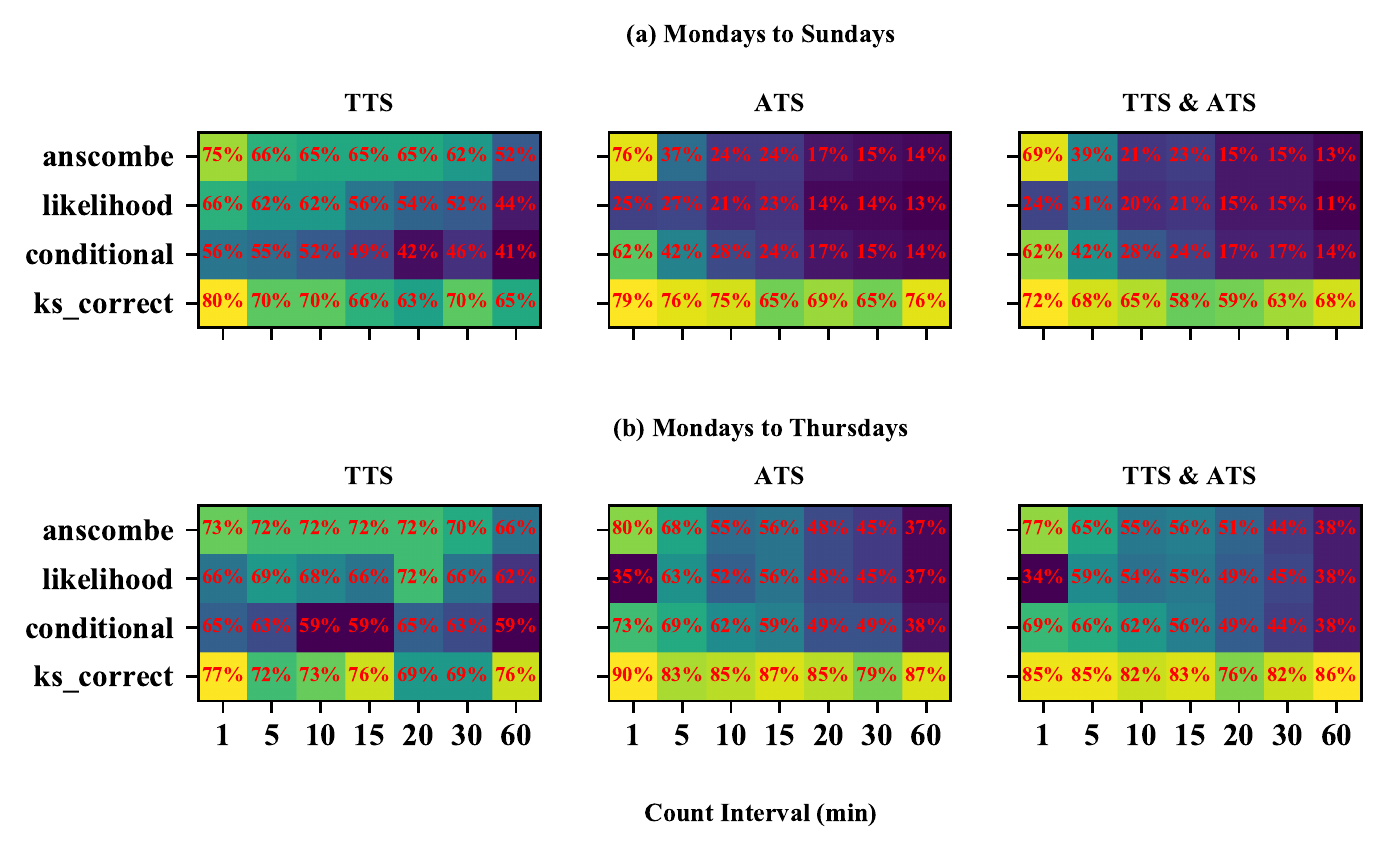}
	\caption{Hypothesis test results for passenger pickups at Community Districts in one-hour off peak}
	\label{arrival-cd} \end{figure} to Figure \ref{arrival-census}\begin{figure}[h]
	\centering
	\includegraphics[width=1.2 \columnwidth]{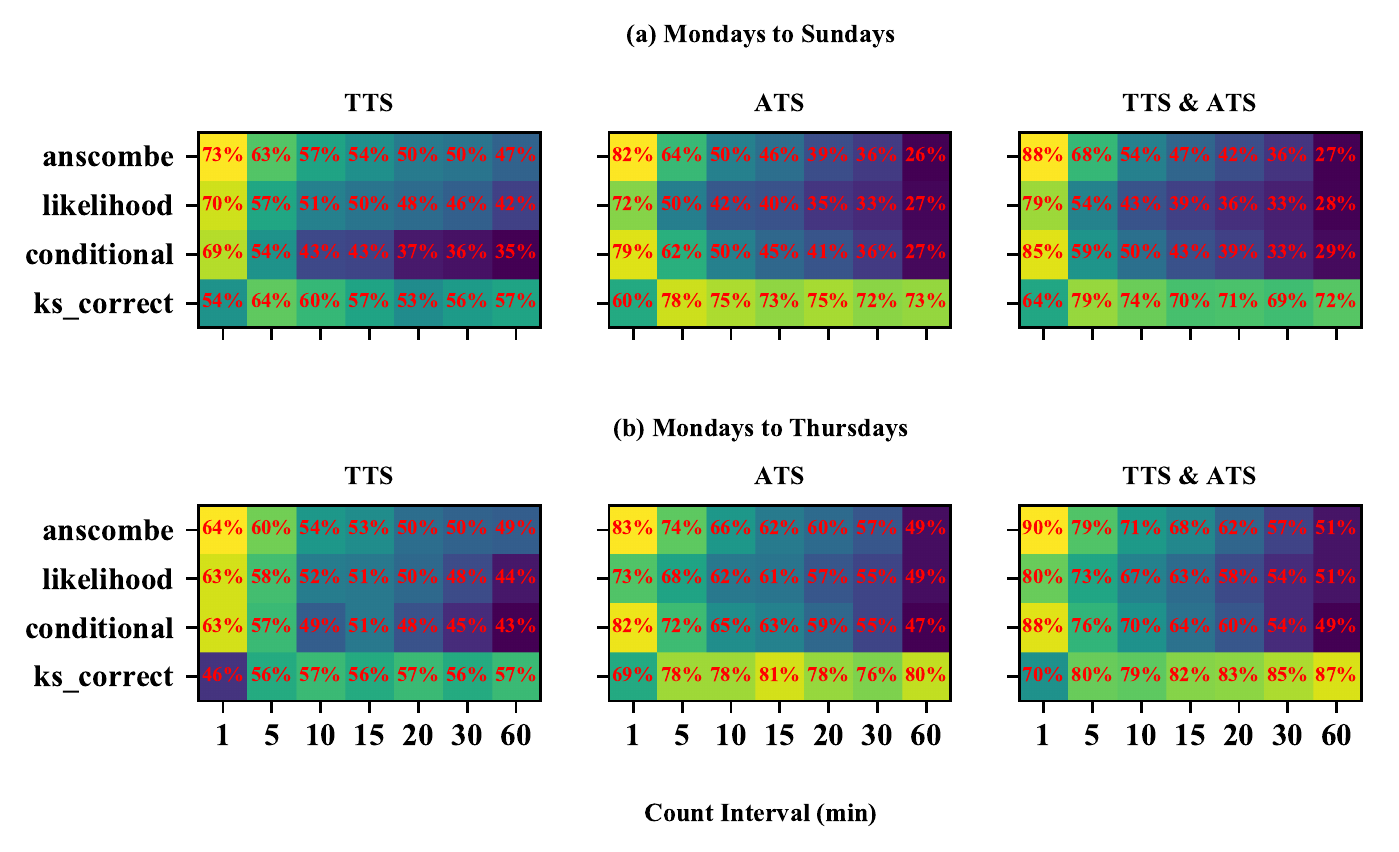}
	\caption{Hypothesis test results for passenger pickups at ZCTA in one-hour off peak}
	\label{arrival-zcta}
\end{figure} \begin{figure}[h]
	\centering
	\includegraphics[width=1.2 \columnwidth]{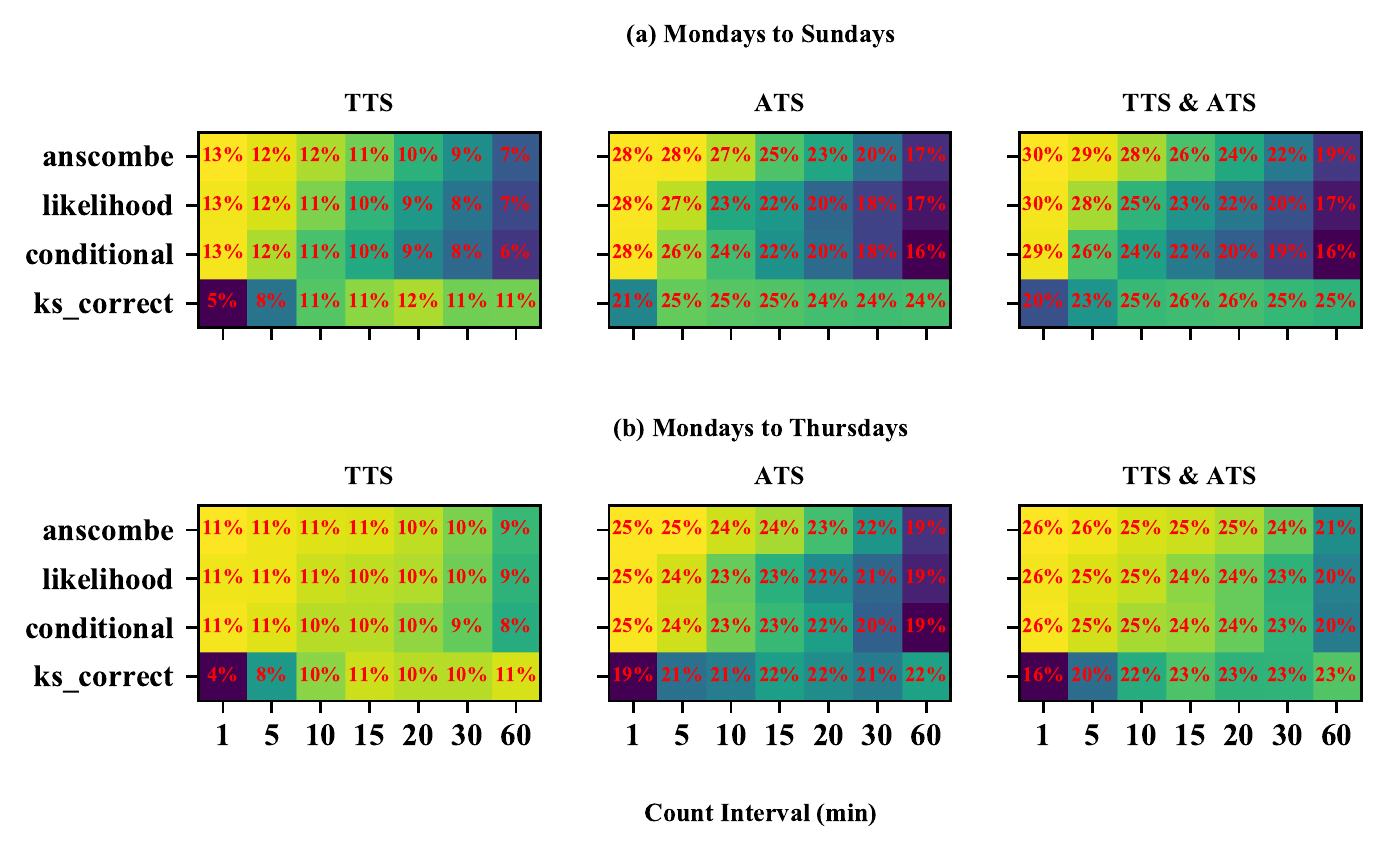}
	\caption{Hypothesis test results for passenger pickups at Census Tracts in one-hour off peak}
	\label{arrival-census}
\end{figure} summarize the percentages of zones where passenger pickups can be assumed as Poisson distribution at four levels of spatial scale in off peak hour, respectively. Within our expectations, the too large or small aggregation does not have perfect output, since large zones have more spatial interactions and heterogeneity and small zones usually do not have any rides. Both community district and ZCTA aggregation have higher percentages. And community district aggregation generally has a slightly higher percentage, regardless of methods, taxi services, count interval and day of the week. One interesting point is that the passenger pickups of TTS and ATS are likely independent considering significant Poisson tests for TTS, ATS, and overall pickups of both TTS and ATS.  

Figure \ref{vehicle-bo_peak}\begin{figure}[h]
	\centering
	\includegraphics[width=1.2 \columnwidth]{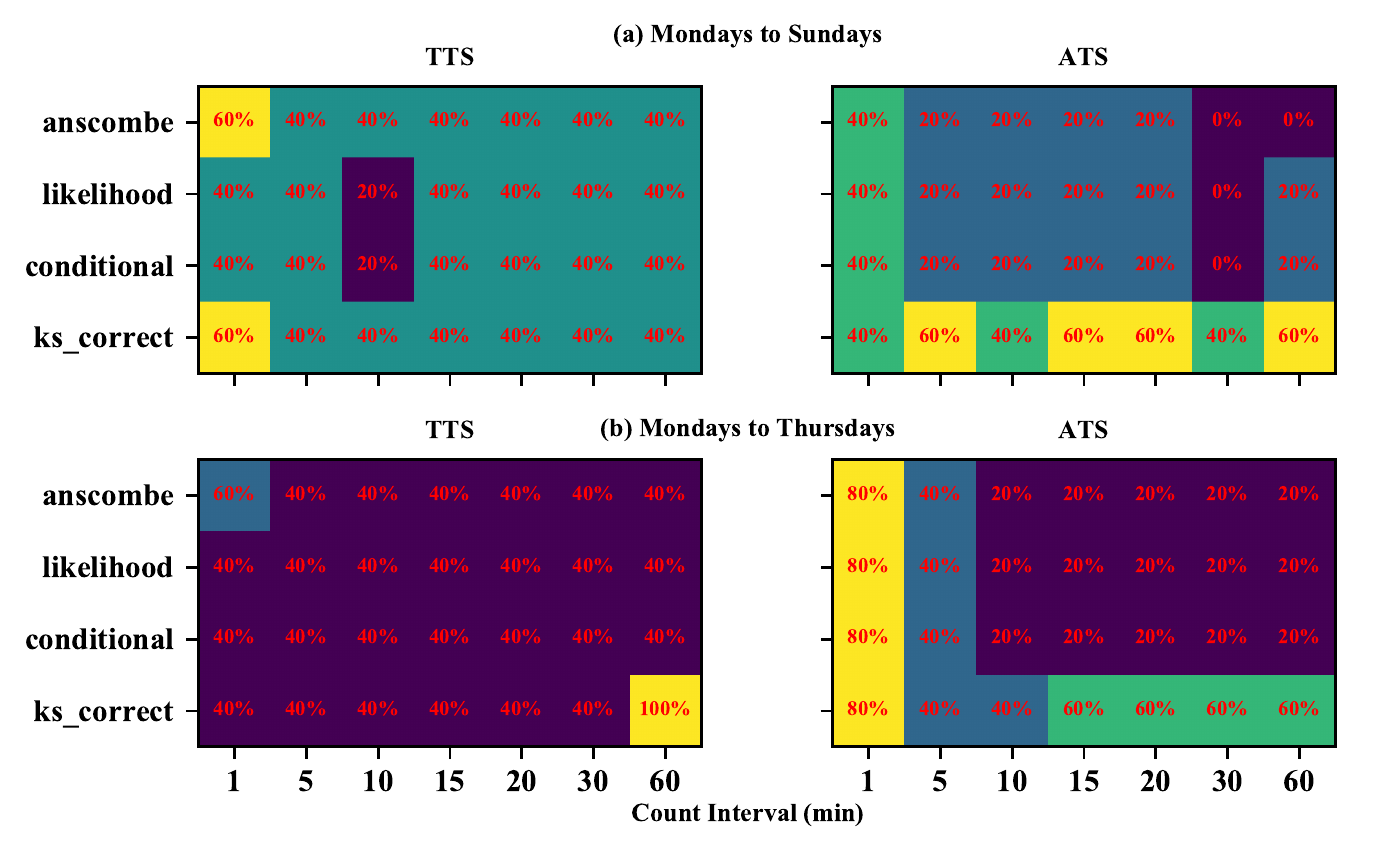}
	\caption{Hypothesis test results for vehicle arrivals at Boroughs in one-hour peak}
	\label{vehicle-bo_peak}
\end{figure} \begin{figure}[h]
	\centering
	\includegraphics[width=1.2 \columnwidth]{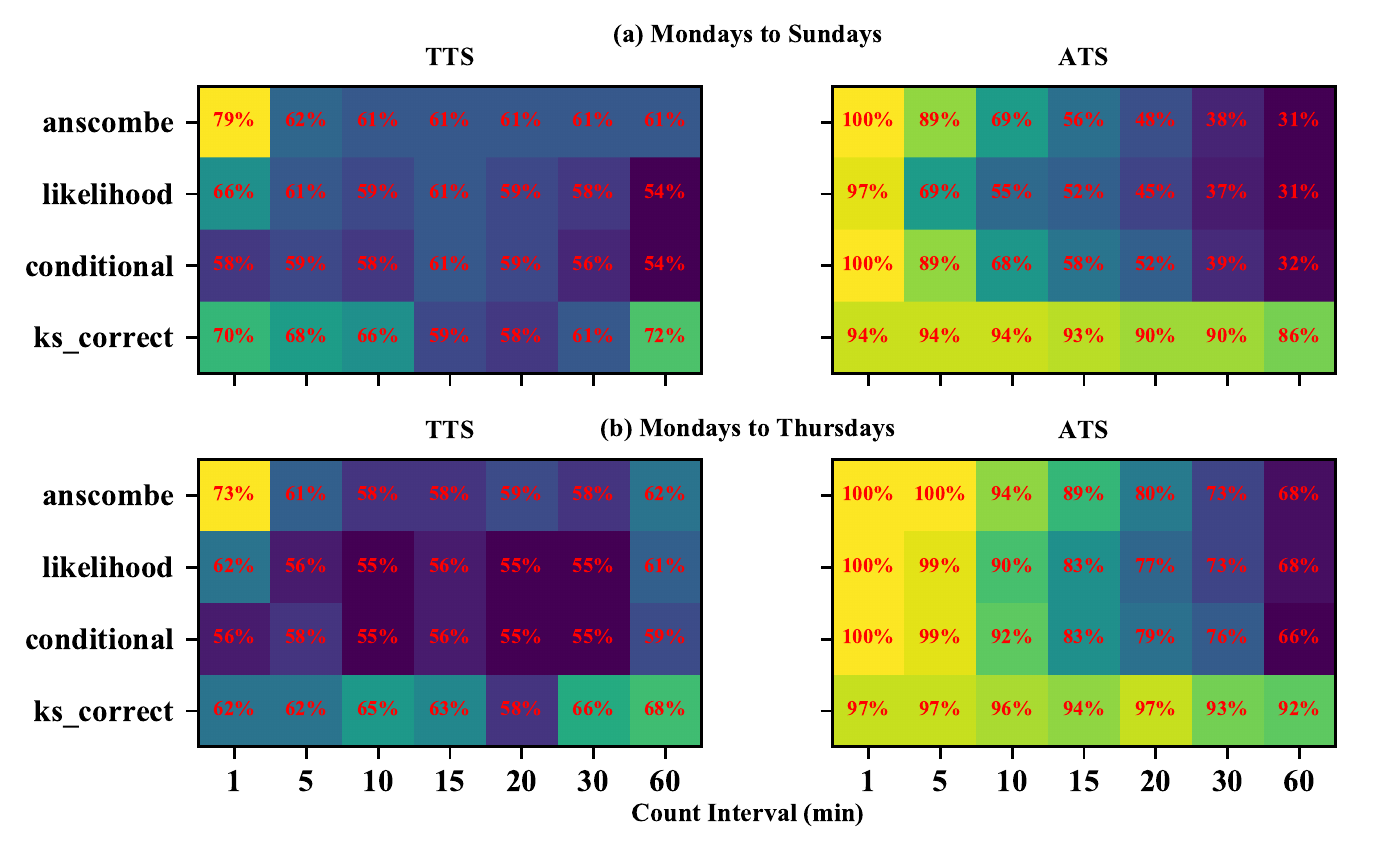}
	\caption{Hypothesis test results for vehicle arrivals at Community Districts in one-hour peak}
	\label{vehicle-cd_peak}
\end{figure} to Figure \ref{vehicle-census_peak} \begin{figure}[h]
	\centering
	\includegraphics[width=1.2 \columnwidth]{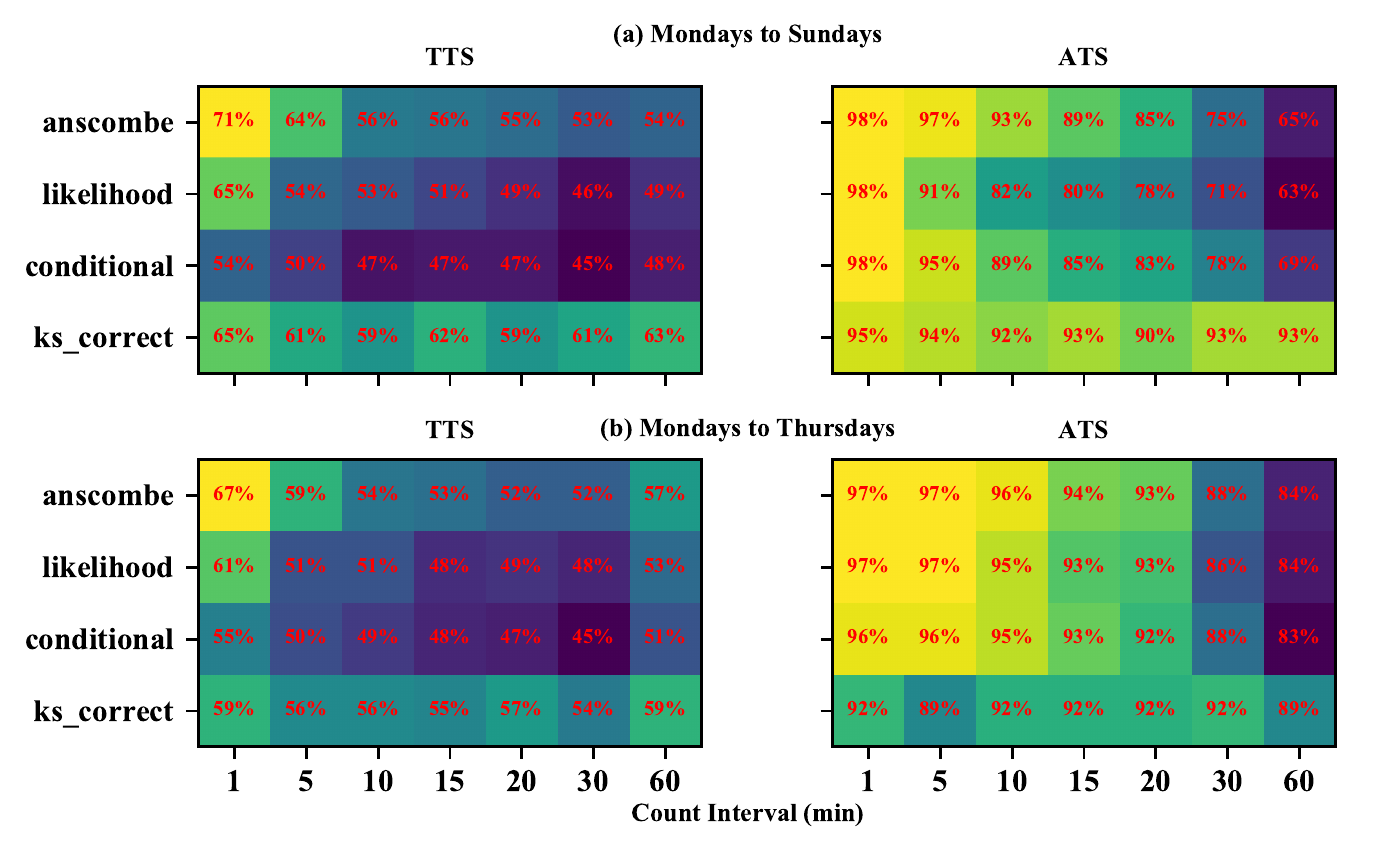}
	\caption{Hypothesis test results for vehicle arrivals at ZCTA in one-hour peak}
	\label{vehicle-zcta_peak}
\end{figure} \begin{figure}[h]
	\centering
	\includegraphics[width=1.2 \columnwidth]{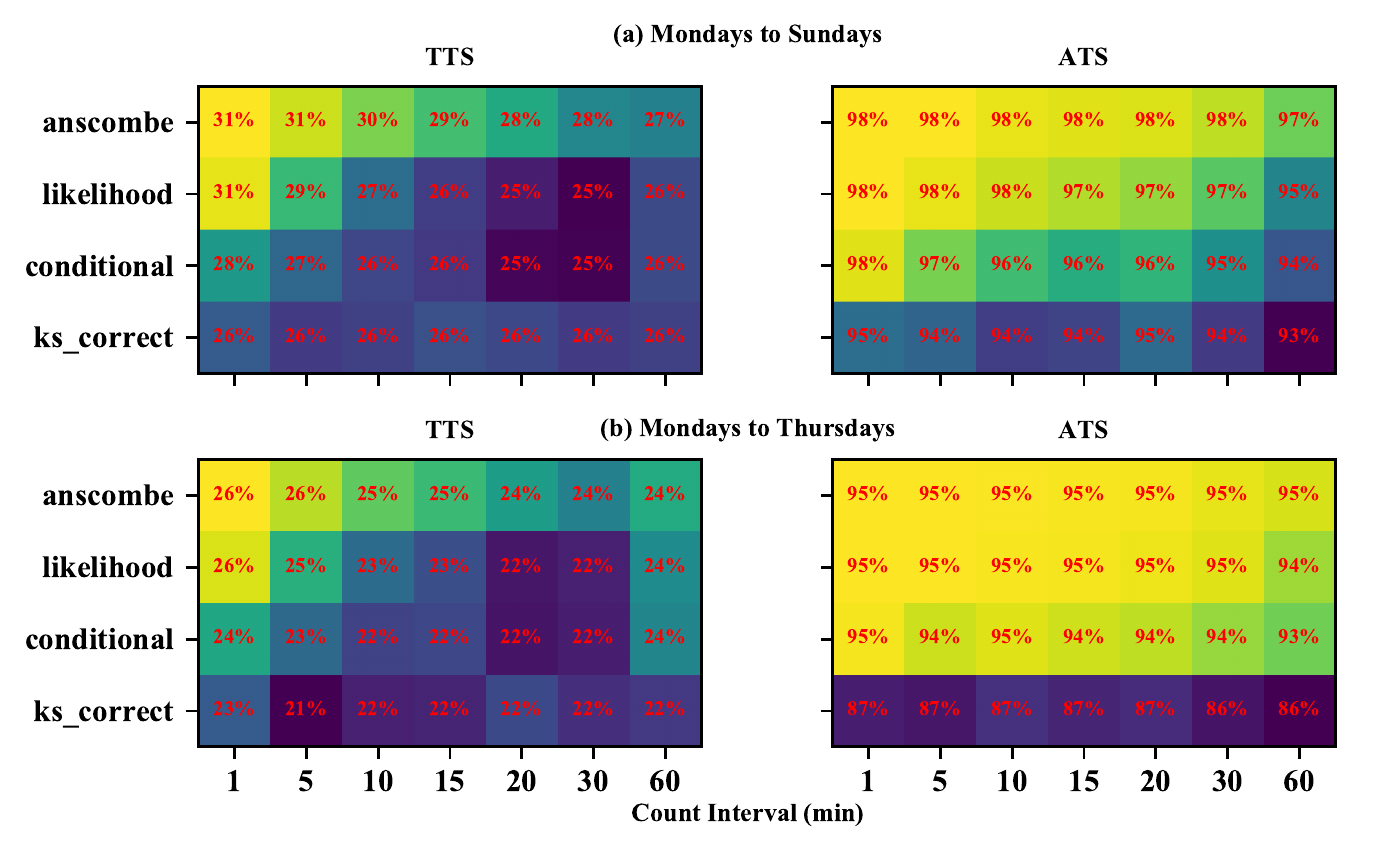}
	\caption{Hypothesis test results for vehicle arrivals at Census Tracts in one-hour peak}
	\label{vehicle-census_peak}
\end{figure} exhibit the percentages of zones where vehicle arrivals (newly online by ATS driver partners or a new shift of TTS driver) can be assumed as Poisson distribution at four levels of spatial scale in peak hour, respectively. Figure \ref{vehicle-bo}\begin{figure}[h]
	\centering
	\includegraphics[width=1.2 \columnwidth]{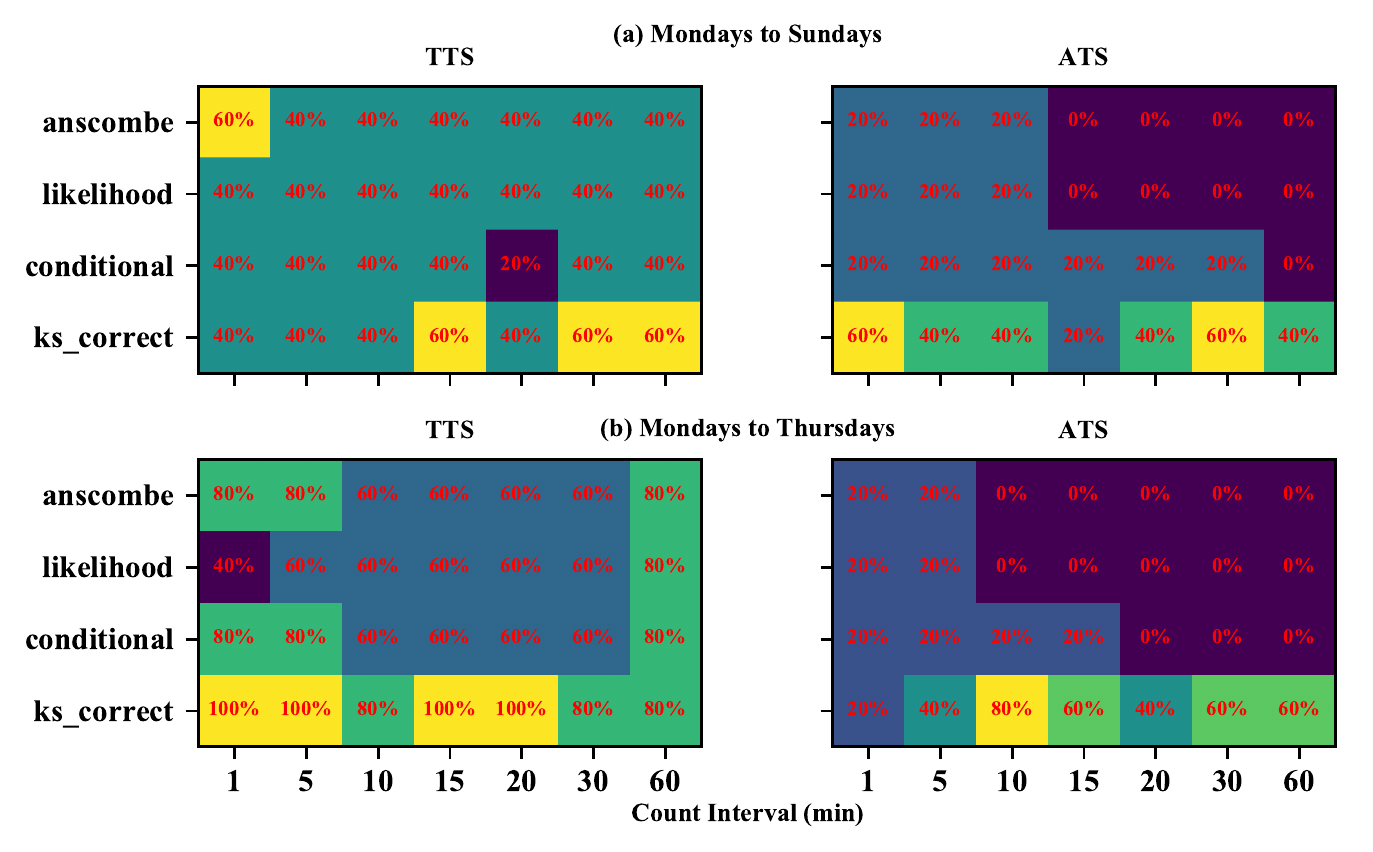}
	\caption{Hypothesis test results for vehicle arrivals at Boroughs in one-hour off peak}
	\label{vehicle-bo}
\end{figure} \begin{figure}[h]
	\centering
	\includegraphics[width=1.2 \columnwidth]{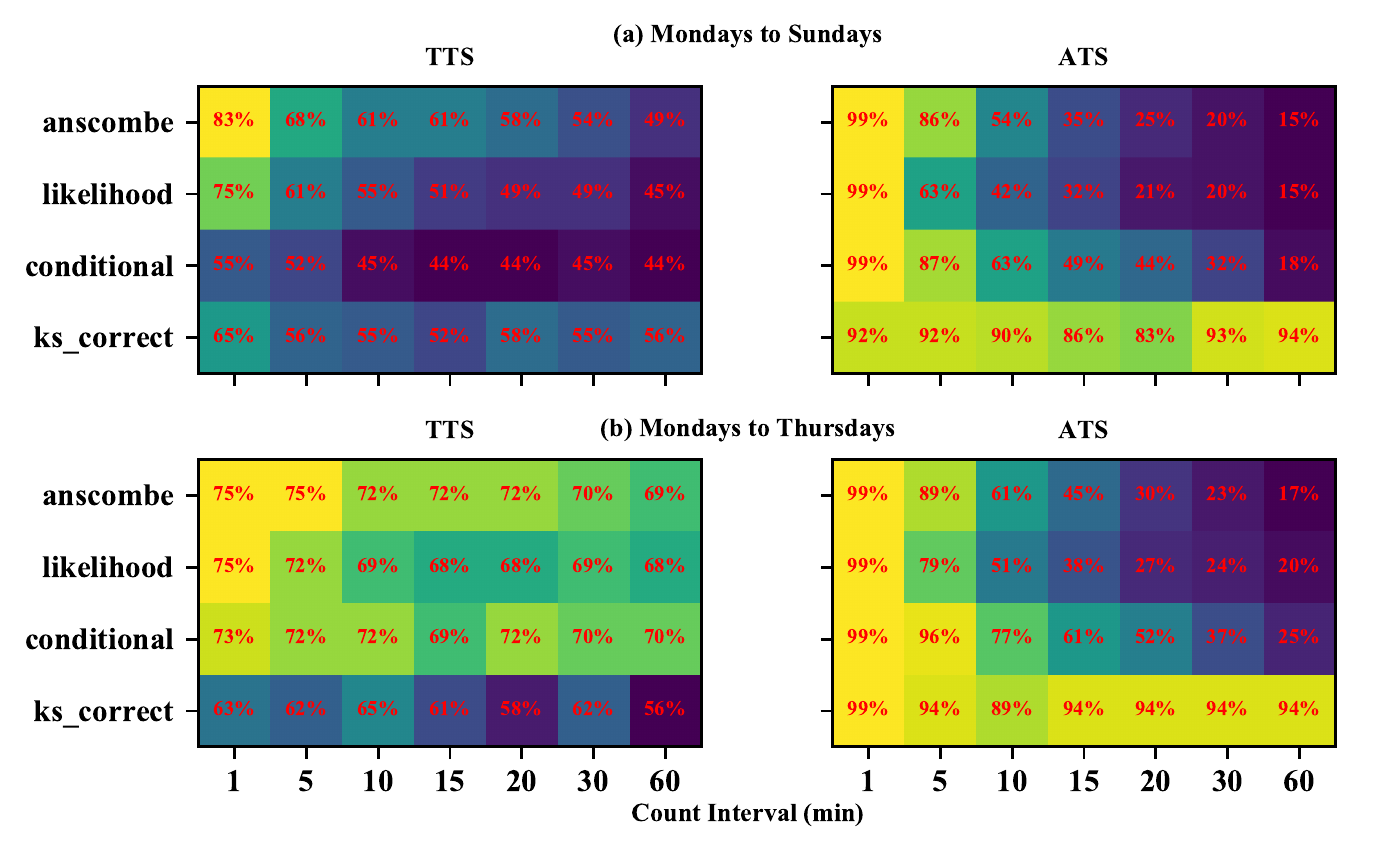}
	\caption{Hypothesis test results for vehicle arrivals at Community Districts in one-hour off peak}
	\label{vehicle-cd}
\end{figure} to Figure \ref{vehicle-census} \begin{figure}[h]
	\centering
	\includegraphics[width=1.2 \columnwidth]{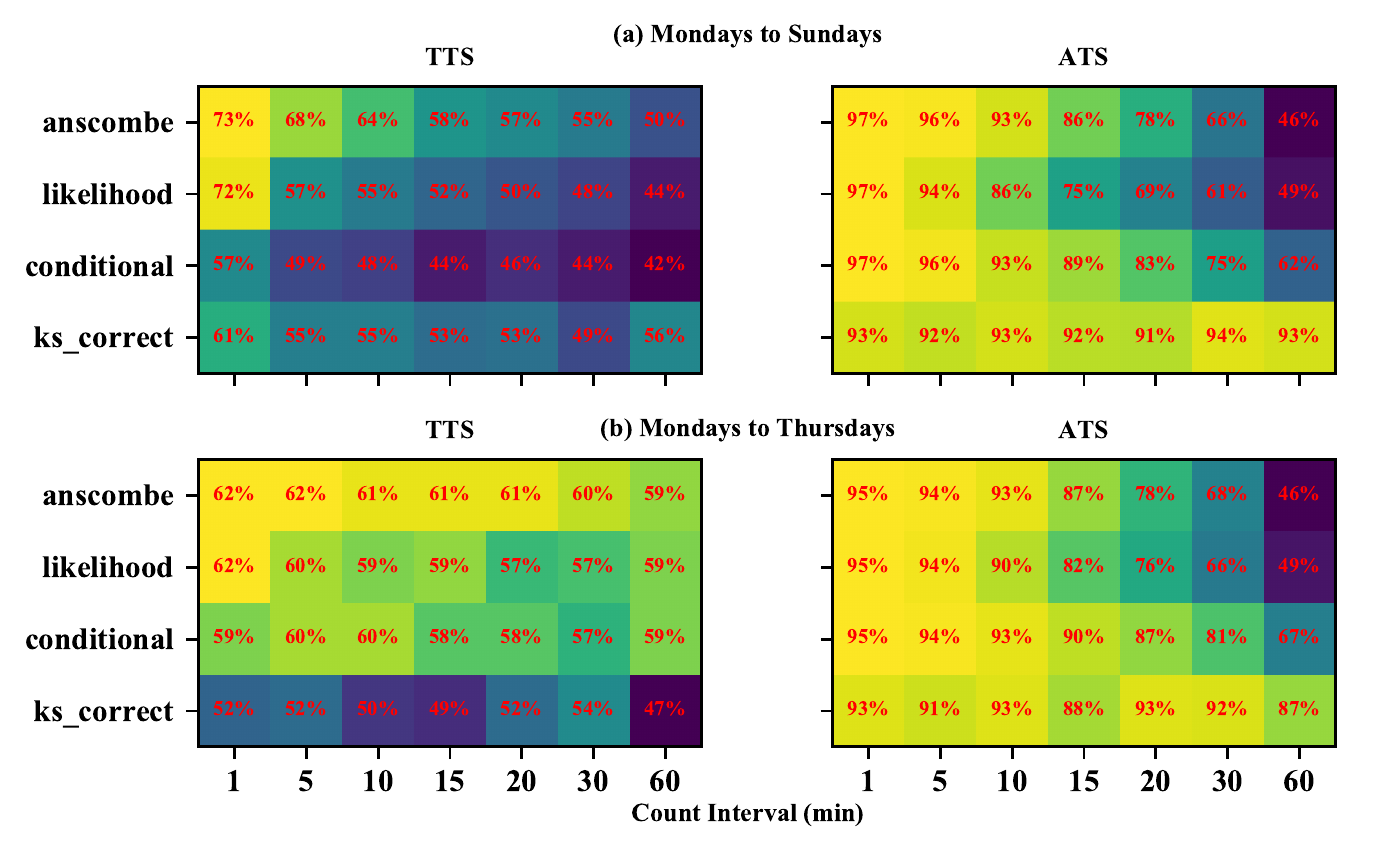}
	\caption{Hypothesis test results for vehicle arrivals at ZCTA in one-hour off peak}
	\label{vehicle-zcta}
\end{figure} \begin{figure}[h]
	\centering
	\includegraphics[width=1.2 \columnwidth]{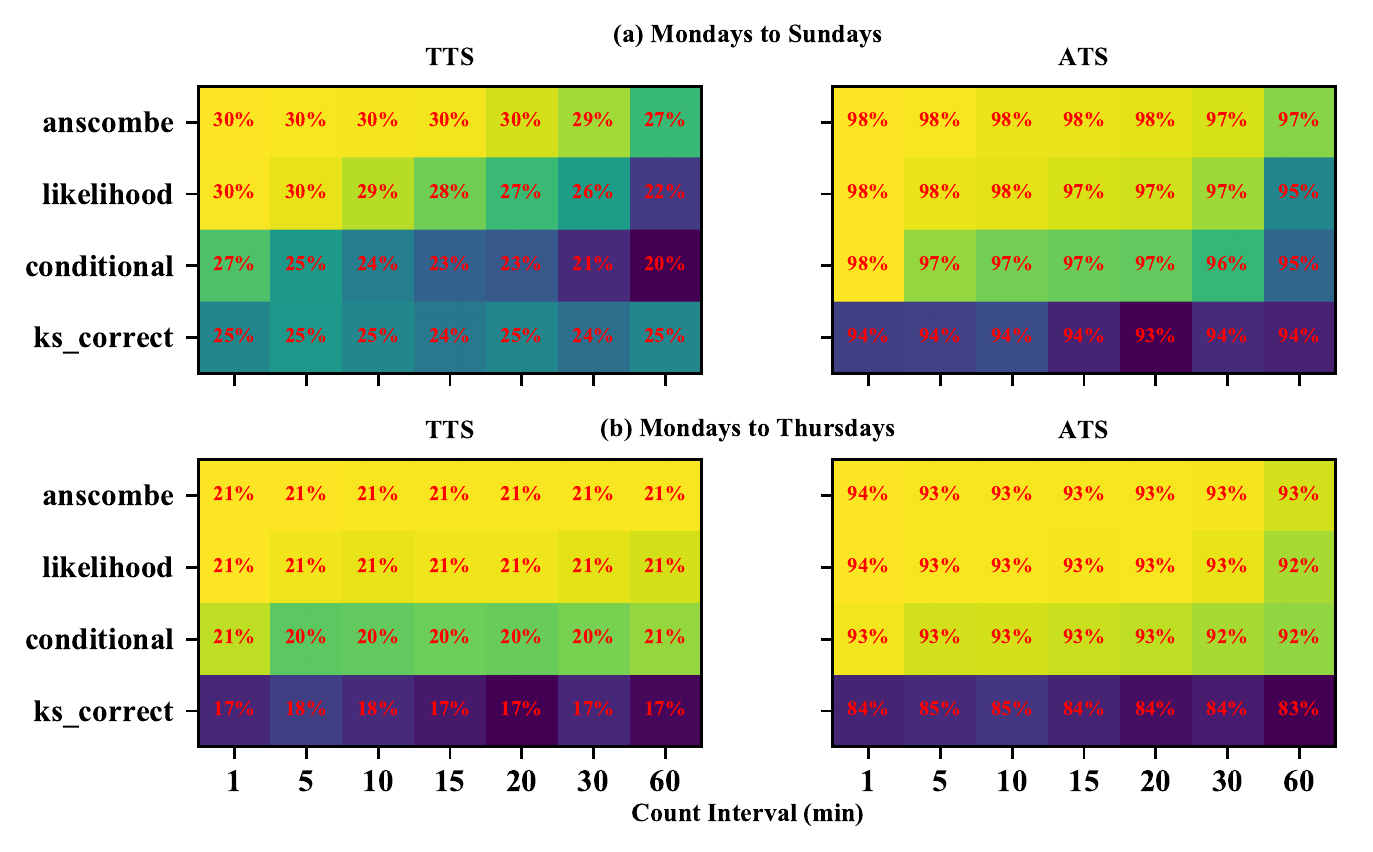}
	\caption{Hypothesis test results for vehicle arrivals at Census Tracts in one-hour off peak}
	\label{vehicle-census}
\end{figure} exhibit the percentages of zones where vehicle arrivals (newly online by ATS driver partners or a new shift of TTS driver) can be assumed as Poisson distribution at four levels of spatial scale in off peak hour, respectively. Similar as test results for passenger pickups, both community district and ZCTA aggregation reveals higher percentages. In particular, the percentages resulted from ATS vehicle arrivals, are sometimes close to 100\%. 	

To sum up, we choose community district to aggregate passenger and vehicle arrivals, and select 1 minute to count those arrivals. 

\section{Hours and Days of Interest}

Regarding the hours and days of interest, we would like to shorten our study period to weekdays (Mondays to Thursdays) and one-hour off peak, mainly depending on the following empirical evidences: 

Figure \ref{arrival-cd_peak}, \ref{arrival-period2_peak},\begin{figure}[h]
	\centering
	\includegraphics[width=1.2 \columnwidth]{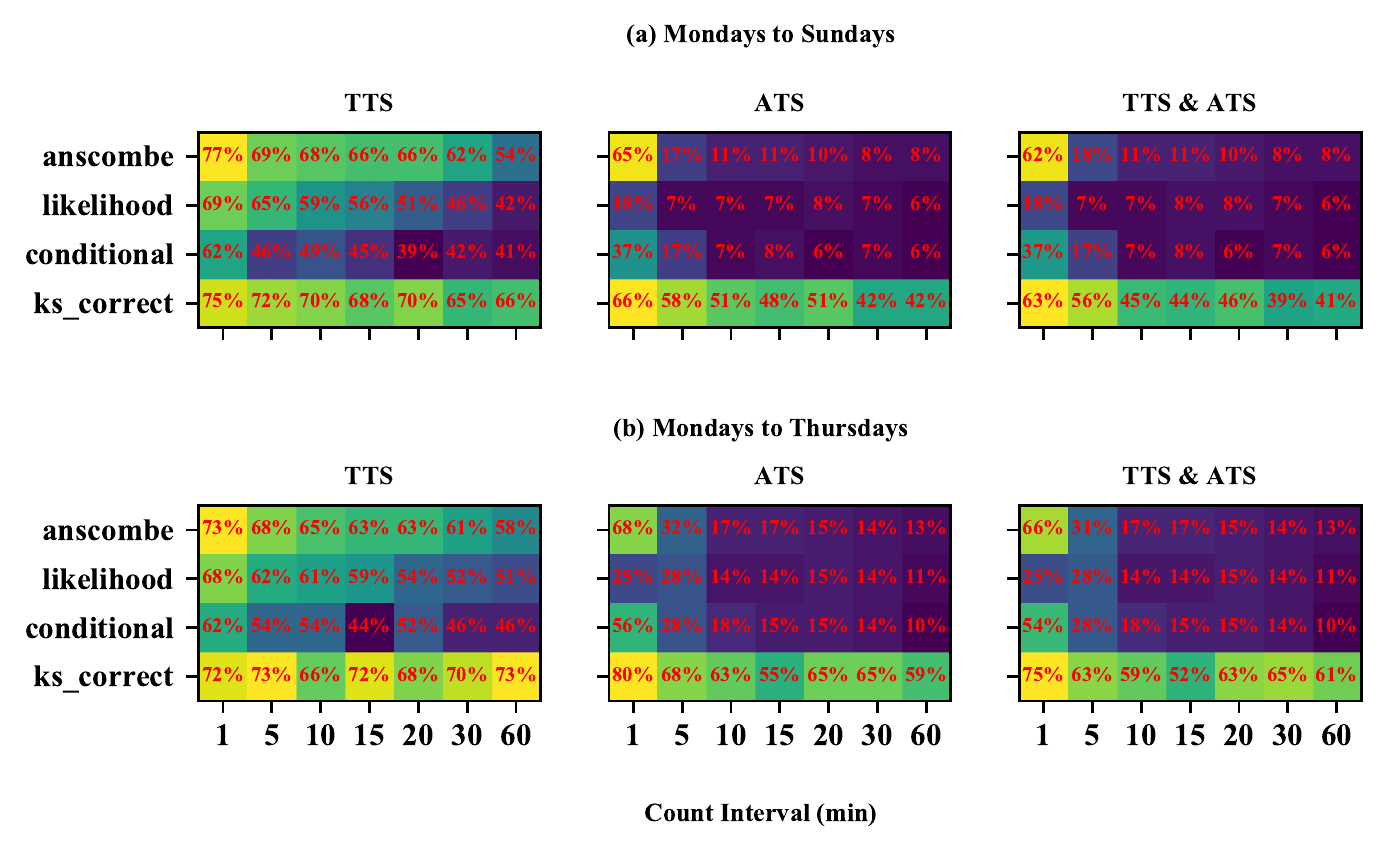}
	\caption{Hypothesis test results for passenger pickups at Community Districts in 2-hour off peak}
	\label{arrival-period2_peak}
\end{figure} and \ref{arrival-period3_peak}\begin{figure}[h]
	\centering
	\includegraphics[width=1.2 \columnwidth]{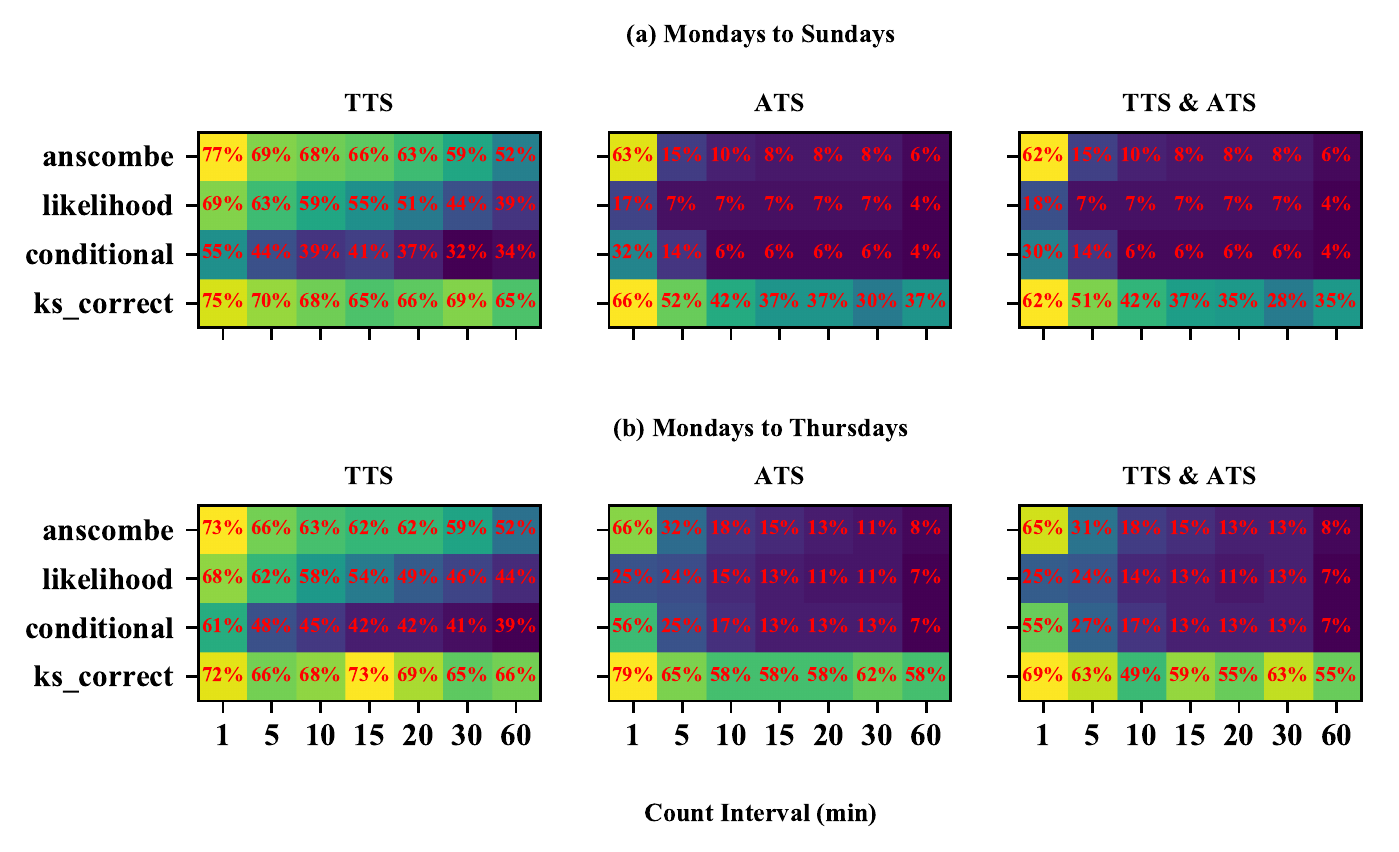}
	\caption{Hypothesis test results for passenger pickups at Community Districts in 3-hour off peak}
	\label{arrival-period3_peak}
\end{figure} compare the hypothesis results for community district aggregation of passenger pickups across different levels of peak hours, as well as day of the week. As number of hours included into peak hours increase, less community districts are not rejecting Poisson distribution. In addition, limiting to the weekdays from Monday to Thursday can slightly increase percentages of significant zones.  
Figure \ref{vehicle-cd_peak}, \ref{vehicle-period2_peak},\begin{figure}[h]
	\centering
	\includegraphics[width=1.2 \columnwidth]{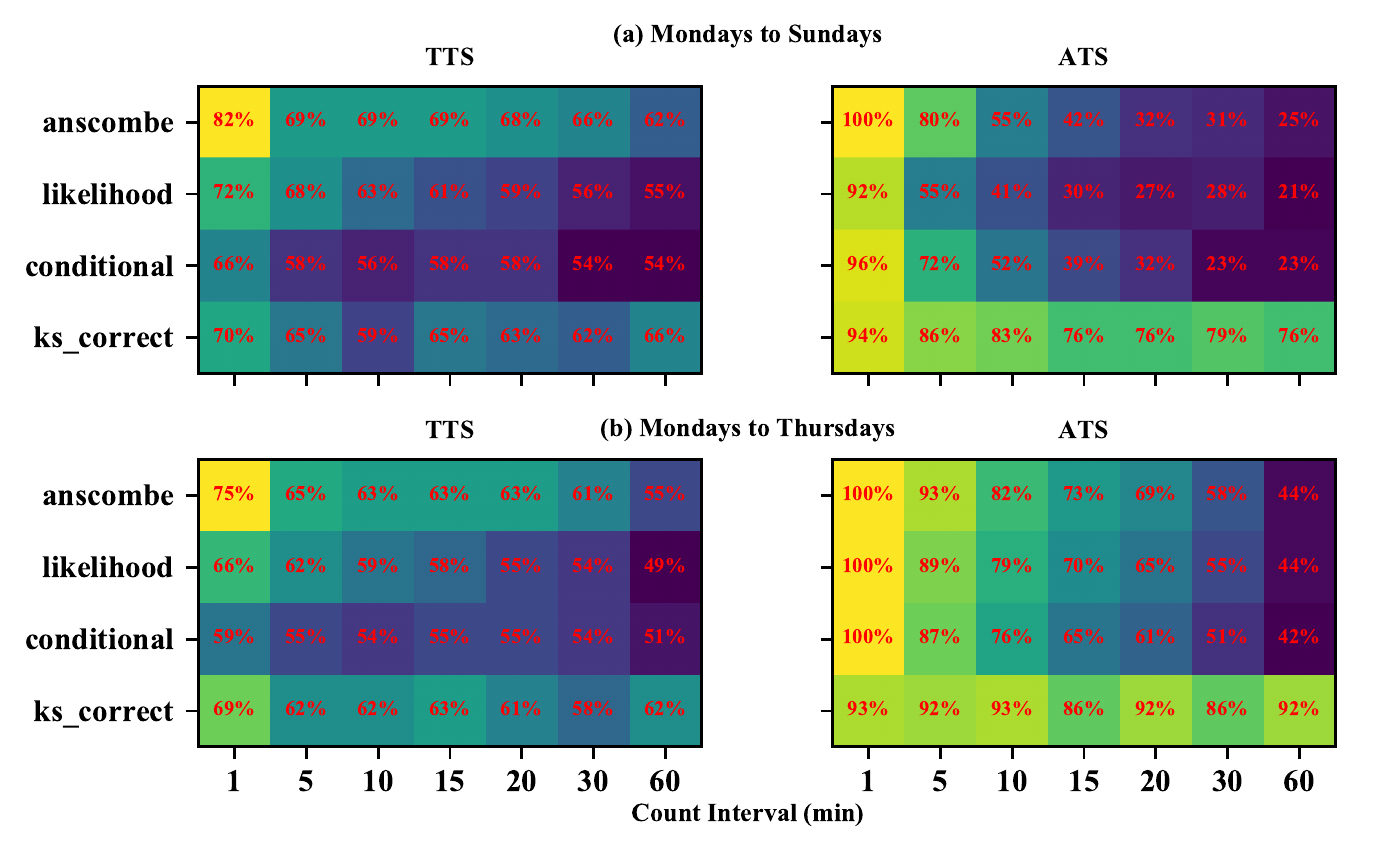}
	\caption{Hypothesis test results for vehicle arrivals at Community Districts in 2-hour off peak}
	\label{vehicle-period2_peak}
\end{figure} and \ref{vehicle-period3_peak}\begin{figure}[h]
	\centering
	\includegraphics[width=1.2 \columnwidth]{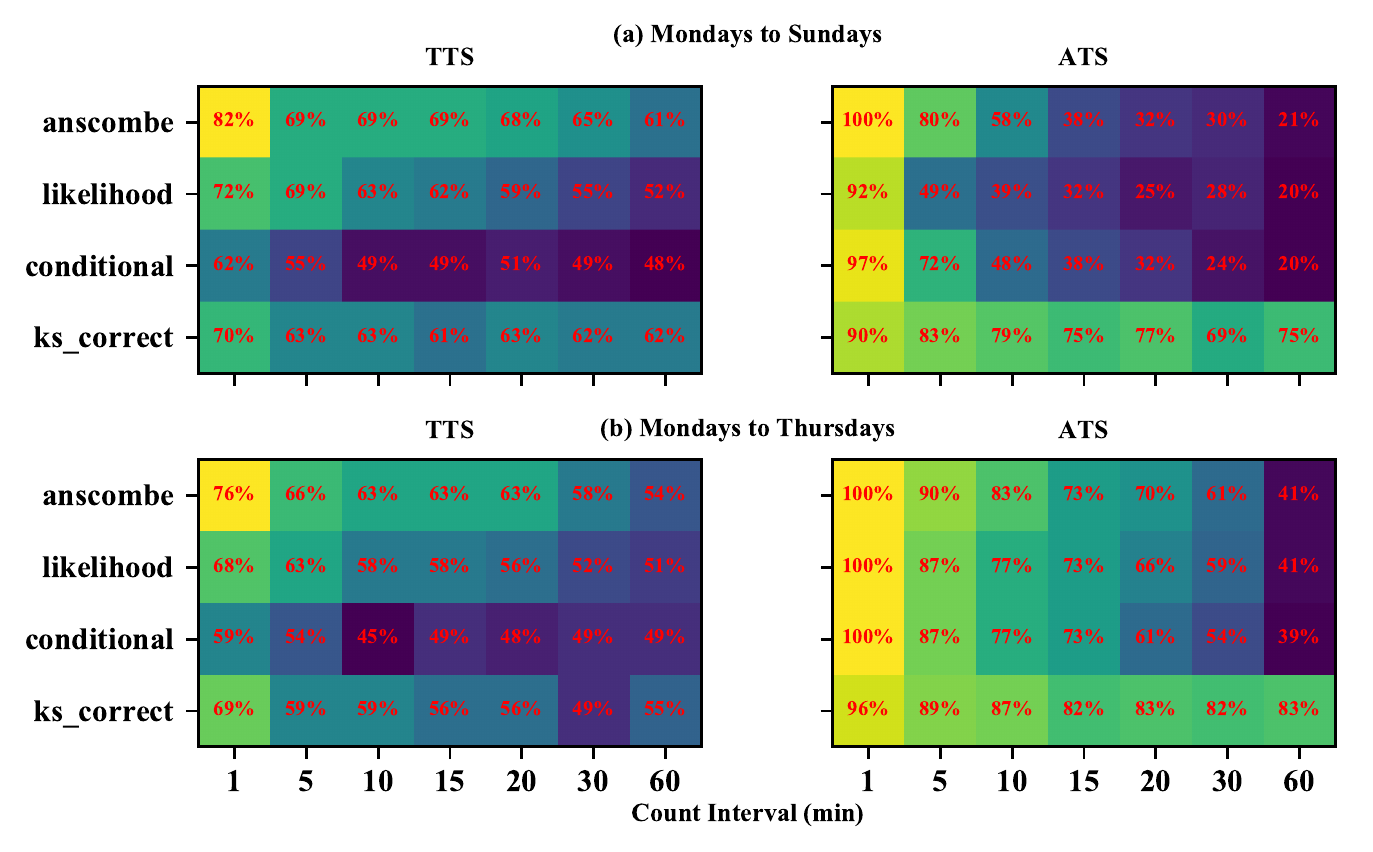}
	\caption{Hypothesis test results for vehicle arrivals at Community Districts in 3-hour off peak}
	\label{vehicle-period3_peak}
\end{figure} compare the hypothesis results for community district aggregation of vehicle arrivals across different levels of off peak hours, as well as day of the week. There are no big differences in the percentages by week of the day, as well as number of hours in off peak period. However, introducing more hours or focusing on weekdays can lead to very small increases in percentages. 

Figure \ref{arrival-cd}, \ref{arrival-period2},\begin{figure}[h]
	\centering
	\includegraphics[width=1.2 \columnwidth]{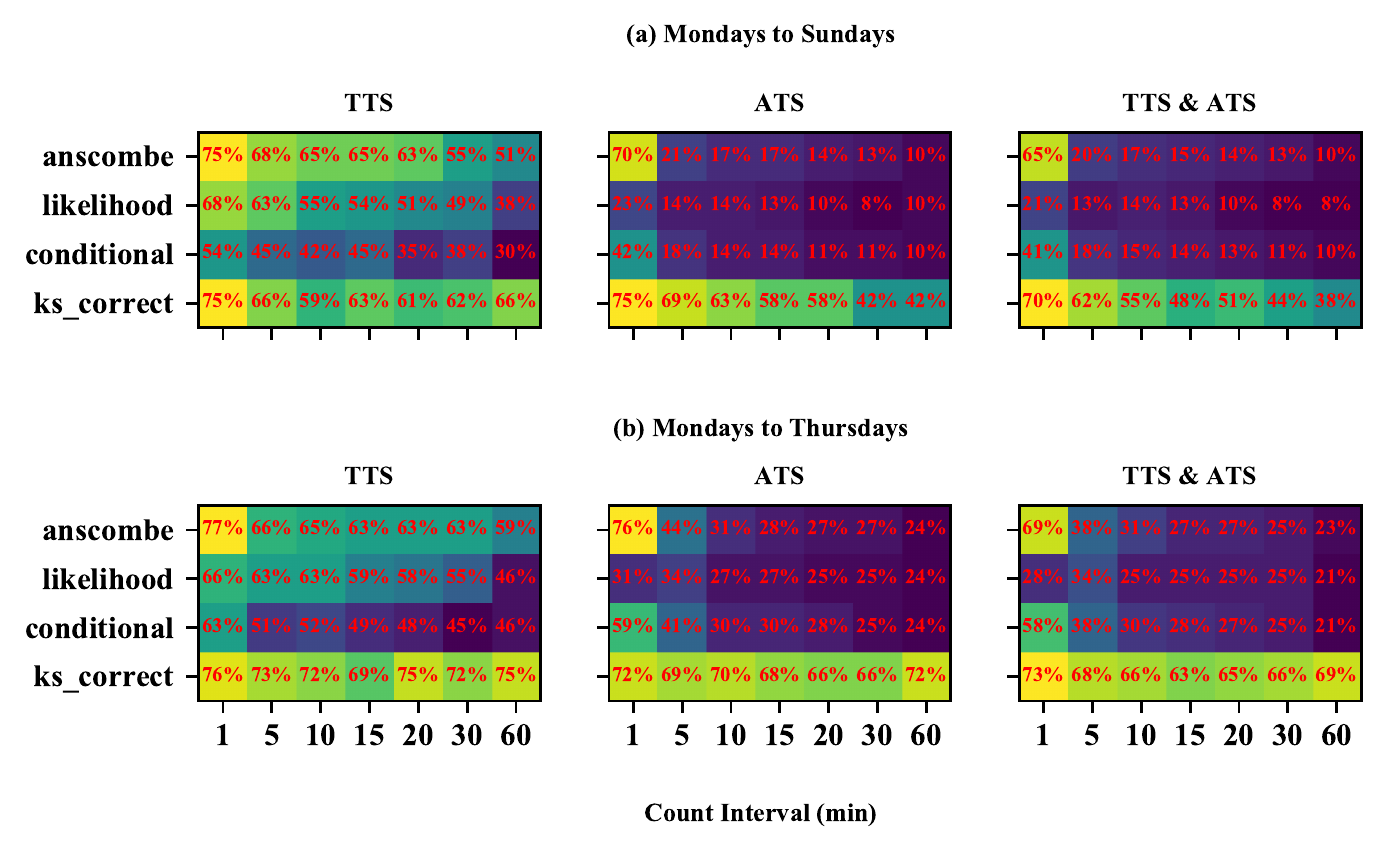}
	\caption{Hypothesis test results for passenger pickups at Community Districts in 2-hour off peak}
	\label{arrival-period2}
\end{figure} and \ref{arrival-period3}\begin{figure}[h]
	\centering
	\includegraphics[width=1.2 \columnwidth]{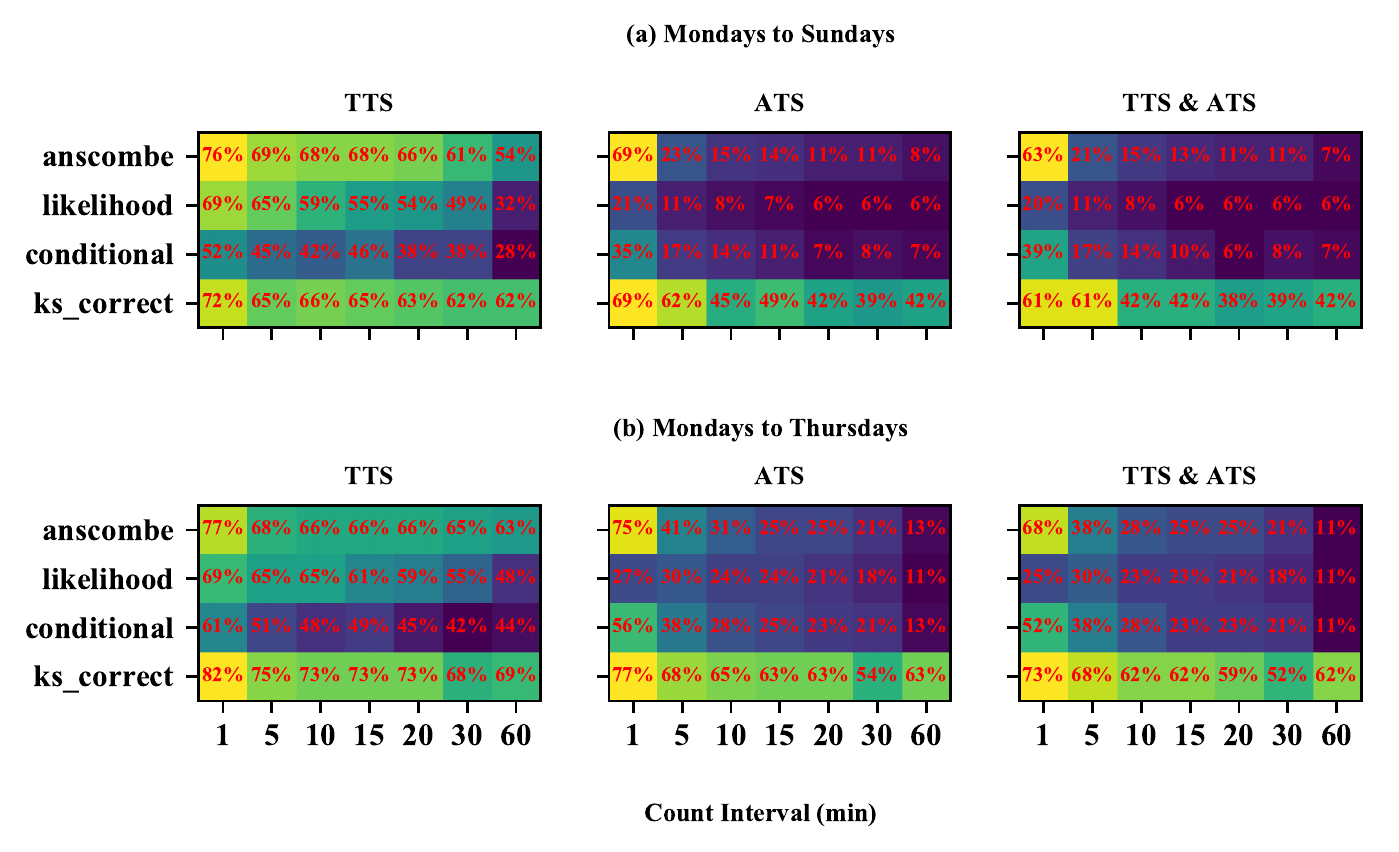}
	\caption{Hypothesis test results for passenger pickups at Community Districts in 3-hour off peak}
	\label{arrival-period3}
\end{figure} compare the hypothesis results for community district aggregation of passenger pickups across different levels of off peak hours, as well as day of the week. As number of hours included into off peak increase, less community districts are not rejecting Poisson distribution. In addition, limiting to the weekdays from Monday to Thursday can slightly increase percentages of significant zones.  
Figure \ref{vehicle-cd}, \ref{vehicle-period2},\begin{figure}[h]
	\centering
	\includegraphics[width=1.2 \columnwidth]{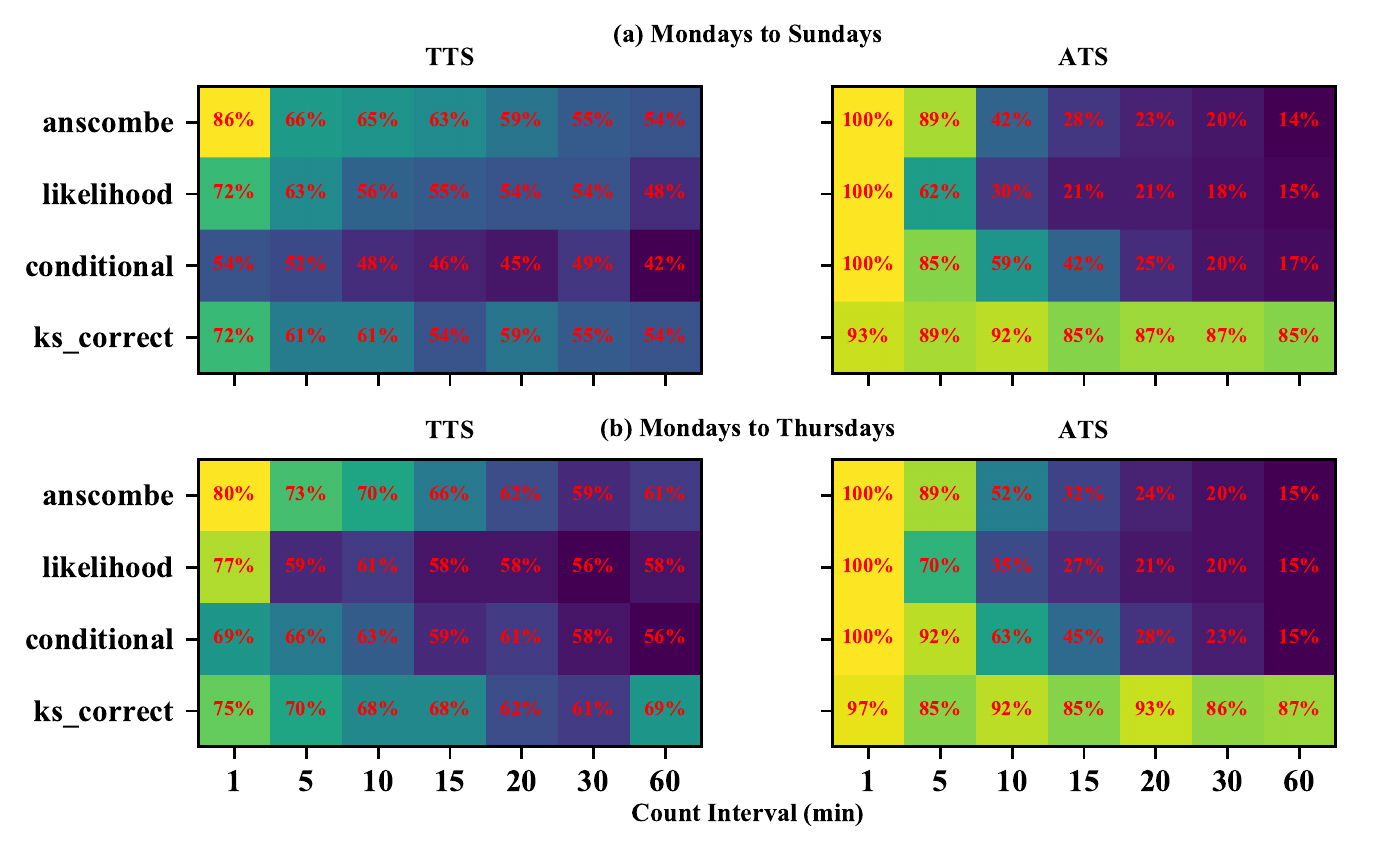}
	\caption{Hypothesis test results for vehicle arrivals at Community Districts in 2-hour off peak}
	\label{vehicle-period2}
\end{figure} and \ref{vehicle-period3}\begin{figure}[h]
	\centering
	\includegraphics[width=1.2 \columnwidth]{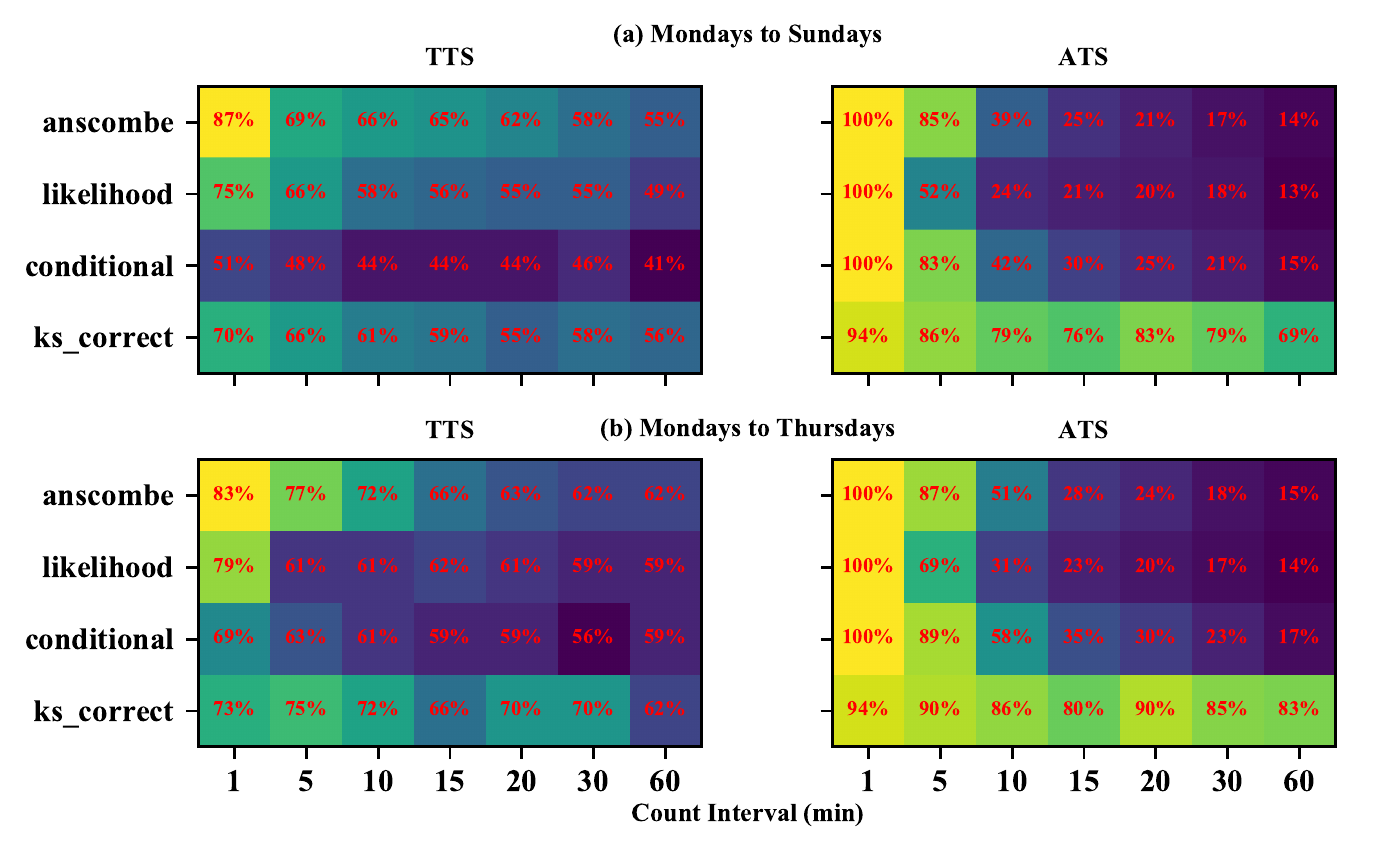}
	\caption{Hypothesis test results for vehicle arrivals at Community Districts in 3-hour off peak}
	\label{vehicle-period3}
\end{figure} compare the hypothesis results for community district aggregation of vehicle arrivals across different levels of off peak hours, as well as day of the week. There are no big differences in the percentages by week of the day, as well as number of hours in off peak period. However, introducing more hours or focusing on weekdays can lead to very small increases in percentages.

\section{Spatial Distribution}
Given the identified aggregation scales, we find the both passenger and vehicle arrivals can be assumed with Poisson distribution, in most community districts, as shown in Figure \ref{arrival-cd} and \ref{vehicle-cd}. In addition, we plot the spatial distribution of community districts not rejecting Poisson distribution in Figure \ref{arrival-spatial}\begin{figure}[!h]
	\centering
	\subfigure[1-min count interval in 1-hour off peak]{\includegraphics[width=1.2 \columnwidth]{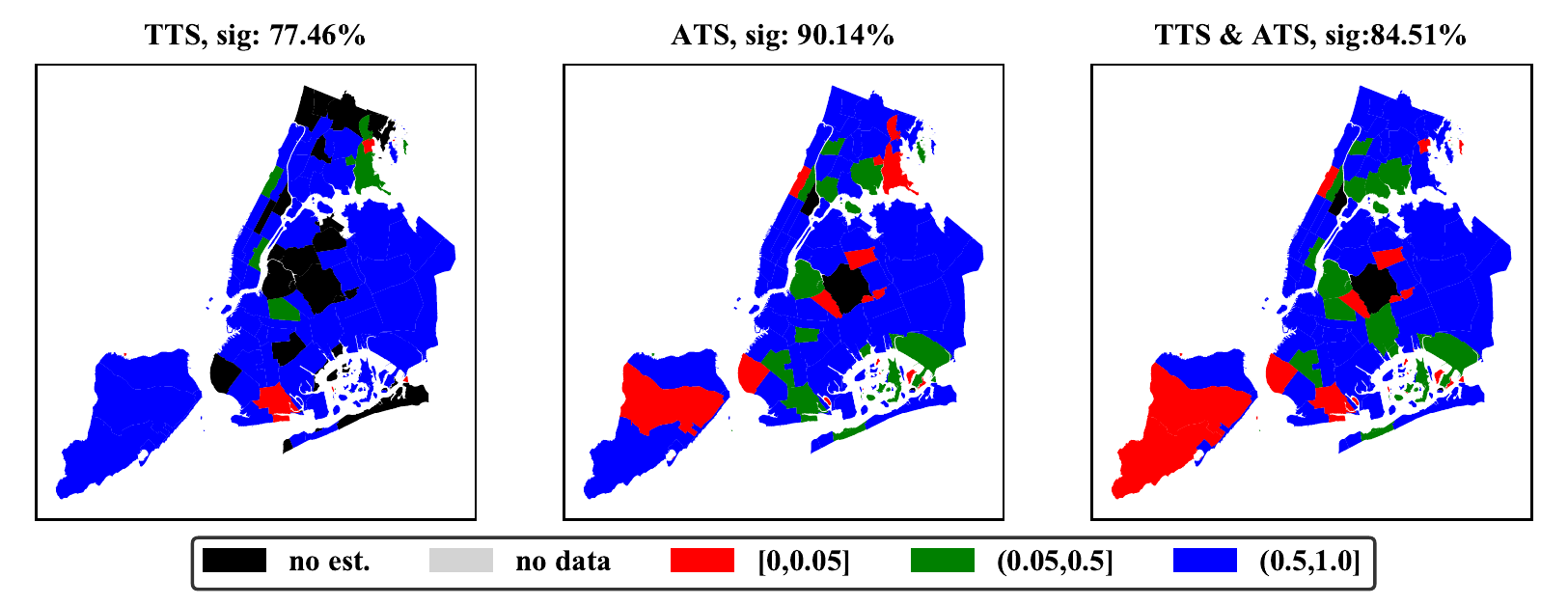}}
	\subfigure[1-min count interval in 1-hour peak]{\includegraphics[width=1.2 \columnwidth]{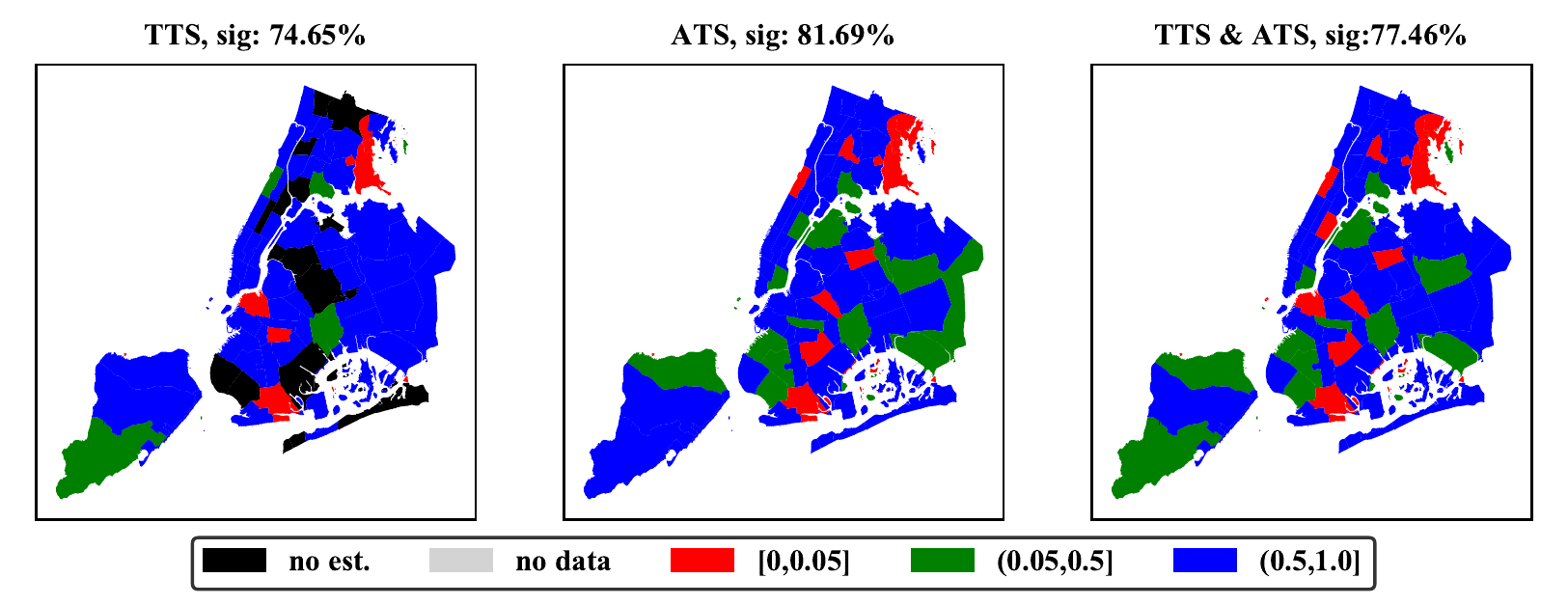}}
	\caption{Hypothesis test results for passenger pickups by Community Districts in weekdays. Note: `sig' indicates percentage of community districts not rejecting Poisson distribution, represented by blue and green color}
	\label{arrival-spatial}
\end{figure} and \ref{vehicle-spatial}\begin{figure}[!h]
	\centering
	\subfigure[1-min count interval in 1-hour off peak]{\includegraphics[width=1.2 \columnwidth]{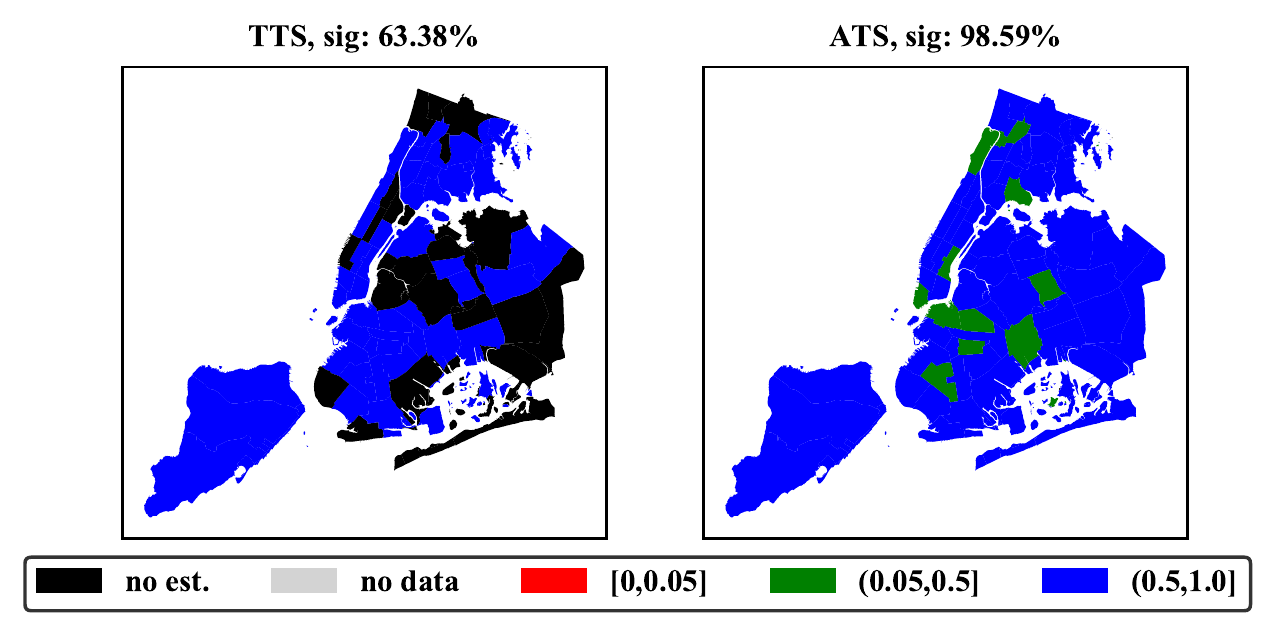}}
	\subfigure[1-min count interval in 1-hour peak]{\includegraphics[width=1.2 \columnwidth]{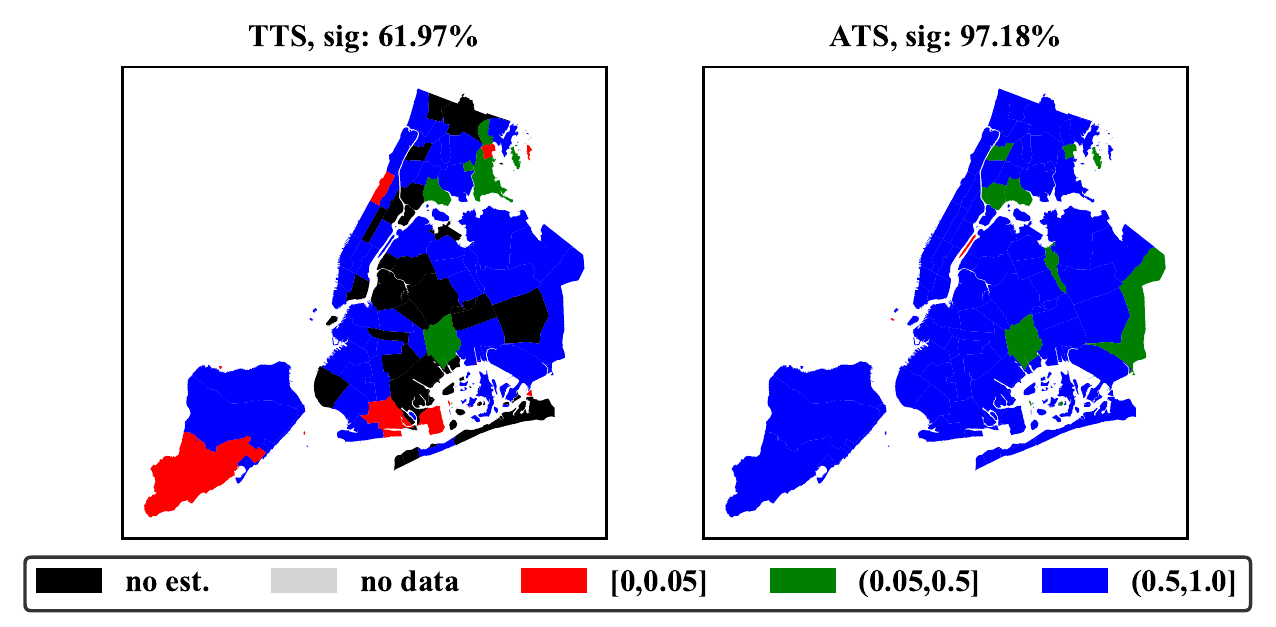}}
	\caption{Hypothesis test results for vehicle arrivals by Community Districts in weekdays. Note: `sig' indicates percentage of community districts not rejecting Poisson distribution, represented by blue and green color}
	\label{vehicle-spatial}
\end{figure}. The both figures indicate that those insignificant (i.e. rejecting Poisson assumption) community districts mainly locate in remote suburban areas, generally with rare TTS and ATS activities. In New York City, most taxi activities concentrate in Manhattan downtown and midtown, as well as two airports, Brooklyn downtown, and Queens downtown. Our hypothesis tests strongly support the Poisson assumption on passenger and vehicle arrivals in those areas. For more details on TTS and ATS activities and facts in NYC, you can refer to 2018 NYC TAXI FACT BOOK (\url{http://www.nyc.gov/html/tlc/downloads/pdf/2018_tlc_factbook.pdf}). 

\end{document}